\pdfoutput=1
\documentclass[sigconf]{aamas} 

\usepackage{balance} %

\setcopyright{ifaamas}
\acmConference[AAMAS '24]{Proc.\@ of the 23rd International Conference
on Autonomous Agents and Multiagent Systems (AAMAS 2024)}{May 6 -- 10, 2024}
{Auckland, New Zealand}{N.~Alechina, V.~Dignum, M.~Dastani, J.S.~Sichman (eds.)}
\copyrightyear{2024}
\acmYear{2024}
\acmDOI{}
\acmPrice{}
\acmISBN{}

\acmSubmissionID{1045}

\title[Tractability and Relevance of MF-RL]{When is Mean-Field Reinforcement Learning Tractable and Relevant?}

\author{Batuhan Yardim}
\affiliation{
  \institution{ETH Zürich}
  \city{Zürich}
  \country{Switzerland}}
\email{yardima@ethz.ch}

\author{Artur Goldman}
\affiliation{
  \institution{HSE University}
  \city{Moscow}
  \country{Russia}}
\email{agoldman@hse.ru}

\author{Niao He}
\affiliation{
  \institution{ETH Zürich}
  \city{Zürich}
  \country{Switzerland}}
\email{niao.he@inf.ethz.ch}

\begin{abstract}
Mean-field reinforcement learning has become a popular theoretical framework for efficiently approximating large-scale multi-agent reinforcement learning (MARL) problems exhibiting symmetry.
However, questions remain regarding the applicability of mean-field approximations: in particular, their approximation accuracy of real-world systems and conditions under which they become computationally tractable.
We establish explicit finite-agent bounds for how well the MFG solution approximates the true $N$-player game for two popular mean-field solution concepts.
Furthermore, for the first time, we establish explicit lower bounds indicating that MFGs are poor or uninformative at approximating $N$-player games assuming only Lipschitz dynamics and rewards.
Finally, we analyze the computational complexity of solving MFGs with only Lipschitz properties and prove that they are in the class of \textsc{PPAD}-complete problems conjectured to be intractable, similar to general sum $N$ player games.
Our theoretical results underscore the limitations of MFGs and complement and justify existing work by proving difficulty in the absence of common theoretical assumptions.
\end{abstract}


\usepackage[utf8]{inputenc} %
\usepackage[T1]{fontenc}    %
\usepackage{hyperref}       %

\usepackage{url}            %
\usepackage{booktabs}       %
\usepackage{amsfonts}       %
\usepackage{nicefrac}       %
\usepackage{microtype}      %
\usepackage{xcolor}         %

\usepackage[english]{babel}
\usepackage{mathtools}
\usepackage{thmtools}
\usepackage{thm-restate}
\usepackage{algorithm}[l]
\usepackage{algpseudocode}
\usepackage{multirow}
\usepackage{chngpage}
\usepackage{xfrac}

\usepackage{mathrsfs}
\usepackage{xspace}
\usepackage{bm}
\usepackage{upgreek}
\usepackage{bbm}

\usepackage{tikz} 
\usetikzlibrary{decorations.text}

\usepackage{subfig}
\usepackage{cleveref}

\newcommand{\safemath}[2]{\newcommand{#1}{\ensuremath{#2}\xspace}}

\safemath{\bma}{\mathbf{a}}
\safemath{\bmb}{\mathbf{b}}
\safemath{\bmc}{\mathbf{c}}
\safemath{\bmd}{\mathbf{d}}
\safemath{\bme}{\mathbf{e}}
\safemath{\bmf}{\mathbf{f}}
\safemath{\bmg}{\mathbf{g}}
\safemath{\bmh}{\mathbf{h}}
\safemath{\bmi}{\mathbf{i}}
\safemath{\bmj}{\mathbf{j}}
\safemath{\bmk}{\mathbf{k}}
\safemath{\bml}{\mathbf{l}}
\safemath{\bmm}{\mathbf{m}}
\safemath{\bmn}{\mathbf{n}}
\safemath{\bmo}{\mathbf{o}}
\safemath{\bmp}{\mathbf{p}}
\safemath{\bmq}{\mathbf{q}}
\safemath{\bmr}{\mathbf{r}}
\safemath{\bms}{\mathbf{s}}
\safemath{\bmt}{\mathbf{t}}
\safemath{\bmu}{\mathbf{u}}
\safemath{\bmv}{\mathbf{v}}
\safemath{\bmw}{\mathbf{w}}
\safemath{\bmx}{\mathbf{x}}
\safemath{\bmy}{\mathbf{y}}
\safemath{\bmz}{\mathbf{z}}
\safemath{\bmzero}{\mathbf{0}}
\safemath{\bmone}{\mathbf{1}}
\safemath{\bmpi}{\pmb{\pi}}
\safemath{\bmalpha}{\pmb{\alpha}}

\bmdefine{\biad}{a}
\bmdefine{\bibd}{b}
\bmdefine{\bicd}{c}
\bmdefine{\bidd}{d}
\bmdefine{\bied}{e}
\bmdefine{\bifd}{f}
\bmdefine{\bigd}{g}
\bmdefine{\bihd}{h}
\bmdefine{\biid}{i}
\bmdefine{\bijd}{j}
\bmdefine{\bikd}{k}
\bmdefine{\bild}{l}
\bmdefine{\bimd}{m}
\bmdefine{\bind}{n}
\bmdefine{\biod}{o}
\bmdefine{\bipd}{p}
\bmdefine{\biqd}{q}
\bmdefine{\bird}{r}
\bmdefine{\bisd}{s}
\bmdefine{\bitd}{t}
\bmdefine{\biud}{u}
\bmdefine{\bivd}{v}
\bmdefine{\biwd}{w}
\bmdefine{\bixd}{x}
\bmdefine{\biyd}{y}
\bmdefine{\bizd}{z}

\bmdefine{\bixid}{\xi}
\bmdefine{\bilambdad}{\lambda}
\bmdefine{\bimud}{\mu}
\bmdefine{\binud}{\nu}
\bmdefine{\bithetad}{\theta}
\bmdefine{\biomegad}{\omega}
\bmdefine{\biphid}{\phi}

\safemath{\bmia}{\biad}
\safemath{\bmib}{\bibd}
\safemath{\bmic}{\bicd}
\safemath{\bmid}{\bidd}
\safemath{\bmie}{\bied}
\safemath{\bmif}{\bifd}
\safemath{\bmig}{\bigd}
\safemath{\bmih}{\bihd}
\safemath{\bmii}{\biid}
\safemath{\bmij}{\bijd}
\safemath{\bmik}{\bikd}
\safemath{\bmil}{\bild}
\safemath{\bmim}{\bimd}
\safemath{\bmin}{\bind}
\safemath{\bmio}{\biod}
\safemath{\bmip}{\bipd}
\safemath{\bmiq}{\biqd}
\safemath{\bmir}{\bird}
\safemath{\bmis}{\bisd}
\safemath{\bmit}{\bitd}
\safemath{\bmiu}{\biud}
\safemath{\bmiv}{\bivd}
\safemath{\bmiw}{\biwd}
\safemath{\bmix}{\bixd}
\safemath{\bmiy}{\biyd}
\safemath{\bmiz}{\bizd}

\safemath{\bmxi}{\bixid}
\safemath{\bmlambda}{\bilambdad}
\safemath{\bmmu}{\bimud}
\safemath{\bmnu}{\binud}
\safemath{\bmtheta}{\bithetad}
\safemath{\bmomega}{\biomegad}
\safemath{\bmphi}{\biphid}

\safemath{\bA}{\mathbf{A}}
\safemath{\bB}{\mathbf{B}}
\safemath{\bC}{\mathbf{C}}
\safemath{\bD}{\mathbf{D}}
\safemath{\bE}{\mathbf{E}}
\safemath{\bF}{\mathbf{F}}
\safemath{\bG}{\mathbf{G}}
\safemath{\bH}{\mathbf{H}}
\safemath{\bI}{\mathbf{I}}
\safemath{\bJ}{\mathbf{J}}
\safemath{\bK}{\mathbf{K}}
\safemath{\bL}{\mathbf{L}}
\safemath{\bM}{\mathbf{M}}
\safemath{\bN}{\mathbf{N}}
\safemath{\bO}{\mathbf{O}}
\safemath{\bP}{\mathbf{P}}
\safemath{\bQ}{\mathbf{Q}}
\safemath{\bR}{\mathbf{R}}
\safemath{\bS}{\mathbf{S}}
\safemath{\bT}{\mathbf{T}}
\safemath{\bU}{\mathbf{U}}
\safemath{\bV}{\mathbf{V}}
\safemath{\bW}{\mathbf{W}}
\safemath{\bX}{\mathbf{X}}
\safemath{\bY}{\mathbf{Y}}
\safemath{\bZ}{\mathbf{Z}}

\safemath{\bZero}{\mathbf{0}}
\safemath{\bOne}{\mathbf{1}}
\safemath{\bDelta}{\mathbf{\Delta}}
\safemath{\bLambda}{\mathbf{\UpLambda}}
\safemath{\bPhi}{\mathbf{\Upphi}}
\safemath{\bSigma}{\mathbf{\Upsigma}}
\safemath{\bOmega}{\mathbf{\Upomega}}
\safemath{\bTheta}{\mathbf{\Uptheta}}

\bmdefine{\biAd}{A}
\bmdefine{\biBd}{B}
\bmdefine{\biCd}{C}
\bmdefine{\biDd}{D}
\bmdefine{\biEd}{E}
\bmdefine{\biFd}{F}
\bmdefine{\biGd}{G}
\bmdefine{\biHd}{H}
\bmdefine{\biId}{I}
\bmdefine{\biJd}{J}
\bmdefine{\biKd}{K}
\bmdefine{\biLd}{L}
\bmdefine{\biMd}{M}
\bmdefine{\biOd}{N}
\bmdefine{\biPd}{O}
\bmdefine{\biQd}{P}
\bmdefine{\biRd}{R}
\bmdefine{\biSd}{S}
\bmdefine{\biTd}{T}
\bmdefine{\biUd}{U}
\bmdefine{\biVd}{V}
\bmdefine{\biWd}{W}
\bmdefine{\biXd}{X}
\bmdefine{\biYd}{Y}
\bmdefine{\biZd}{Z}

\bmdefine{\biDelta}{\Delta}
\bmdefine{\biLambda}{\Lambda}
\bmdefine{\biPhi}{\Phi}
\bmdefine{\biSigma}{\Sigma}
\bmdefine{\biOmega}{\Omega}
\bmdefine{\biTheta}{\Theta}

\safemath{\bimA}{\biAd}
\safemath{\bimB}{\biBd}
\safemath{\bimC}{\biCd}
\safemath{\bimD}{\biDd}
\safemath{\bimE}{\biEd}
\safemath{\bimF}{\biFd}
\safemath{\bimG}{\biGd}
\safemath{\bimH}{\biHd}
\safemath{\bimI}{\biId}
\safemath{\bimJ}{\biJd}
\safemath{\bimK}{\biKd}
\safemath{\bimL}{\biLd}
\safemath{\bimM}{\biMd}
\safemath{\bimN}{\biNd}
\safemath{\bimO}{\biOd}
\safemath{\bimP}{\biPd}
\safemath{\bimQ}{\biQd}
\safemath{\bimR}{\biRd}
\safemath{\bimS}{\biSd}
\safemath{\bimT}{\biTd}
\safemath{\bimU}{\biUd}
\safemath{\bimV}{\biVd}
\safemath{\bimW}{\biWd}
\safemath{\bimX}{\biXd}
\safemath{\bimY}{\biYd}
\safemath{\bimZ}{\biZd}

\safemath{\bimDelta}{\biDelta}
\safemath{\bimLambda}{\biLambda}
\safemath{\bimPhi}{\biPhi}
\safemath{\bimSigma}{\biSigma}
\safemath{\bimOmega}{\biOmega}
\safemath{\bimTheta}{\biTheta}

\safemath{\setA}{\mathcal{A}}
\safemath{\setB}{\mathcal{B}}
\safemath{\setC}{\mathcal{C}}
\safemath{\setD}{\mathcal{D}}
\safemath{\setE}{\mathcal{E}}
\safemath{\setF}{\mathcal{F}}
\safemath{\setG}{\mathcal{G}}
\safemath{\setH}{\mathcal{H}}
\safemath{\setI}{\mathcal{I}}
\safemath{\setJ}{\mathcal{J}}
\safemath{\setK}{\mathcal{K}}
\safemath{\setL}{\mathcal{L}}
\safemath{\setM}{\mathcal{M}}
\safemath{\setN}{\mathcal{N}}
\safemath{\setO}{\mathcal{O}}
\safemath{\setP}{\mathcal{P}}
\safemath{\setQ}{\mathcal{Q}}
\safemath{\setR}{\mathcal{R}}
\safemath{\setS}{\mathcal{S}}
\safemath{\setT}{\mathcal{T}}
\safemath{\setU}{\mathcal{U}}
\safemath{\setV}{\mathcal{V}}
\safemath{\setW}{\mathcal{W}}
\safemath{\setX}{\mathcal{X}}
\safemath{\setY}{\mathcal{Y}}
\safemath{\setZ}{\mathcal{Z}}
\safemath{\emptySet}{\varnothing}

\safemath{\colA}{\mathscr{A}}
\safemath{\colB}{\mathscr{B}}
\safemath{\colC}{\mathscr{C}}
\safemath{\colD}{\mathscr{D}}
\safemath{\colE}{\mathscr{E}}
\safemath{\colF}{\mathscr{F}}
\safemath{\colG}{\mathscr{G}}
\safemath{\colH}{\mathscr{H}}
\safemath{\colI}{\mathscr{I}}
\safemath{\colJ}{\mathscr{J}}
\safemath{\colK}{\mathscr{K}}
\safemath{\colL}{\mathscr{L}}
\safemath{\colM}{\mathscr{M}}
\safemath{\colN}{\mathscr{N}}
\safemath{\colO}{\mathscr{O}}
\safemath{\colP}{\mathscr{P}}
\safemath{\colQ}{\mathscr{Q}}
\safemath{\colR}{\mathscr{R}}
\safemath{\colS}{\mathscr{S}}
\safemath{\colT}{\mathscr{T}}
\safemath{\colU}{\mathscr{U}}
\safemath{\colV}{\mathscr{V}}
\safemath{\colW}{\mathscr{W}}
\safemath{\colX}{\mathscr{X}}
\safemath{\colY}{\mathscr{Y}}
\safemath{\colZ}{\mathscr{Z}}

\safemath{\opA}{\mathbb{A}}
\safemath{\opB}{\mathbb{B}}
\safemath{\opC}{\mathbb{C}}
\safemath{\opD}{\mathbb{D}}
\safemath{\opE}{\mathbb{E}}
\safemath{\opF}{\mathbb{F}}
\safemath{\opG}{\mathbb{G}}
\safemath{\opH}{\mathbb{H}}
\safemath{\opI}{\mathbb{I}}
\safemath{\opJ}{\mathbb{J}}
\safemath{\opK}{\mathbb{K}}
\safemath{\opL}{\mathbb{L}}
\safemath{\opM}{\mathbb{M}}
\safemath{\opN}{\mathbb{N}}
\safemath{\opO}{\mathbb{O}}
\safemath{\opP}{\mathbb{P}}
\safemath{\opQ}{\mathbb{Q}}
\safemath{\opR}{\mathbb{R}}
\safemath{\opS}{\mathbb{S}}
\safemath{\opT}{\mathbb{T}}
\safemath{\opU}{\mathbb{U}}
\safemath{\opV}{\mathbb{V}}
\safemath{\opW}{\mathbb{W}}
\safemath{\opX}{\mathbb{X}}
\safemath{\opY}{\mathbb{Y}}
\safemath{\opZ}{\mathbb{Z}}
\safemath{\opZero}{\mathbb{O}}
\safemath{\identityop}{\opI}

\safemath{\veca}{\bma}
\safemath{\vecb}{\bmb}
\safemath{\vecc}{\bmc}
\safemath{\vecd}{\bmd}
\safemath{\vece}{\bme}
\safemath{\vecf}{\bmf}
\safemath{\vecg}{\bmg}
\safemath{\vech}{\bmh}
\safemath{\veci}{\bmi}
\safemath{\vecj}{\bmj}
\safemath{\veck}{\bmk}
\safemath{\vecl}{\bml}
\safemath{\vecm}{\bmm}
\safemath{\vecn}{\bmn}
\safemath{\veco}{\bmo}
\safemath{\vecp}{\bmmp}
\safemath{\vecq}{\bmq}
\safemath{\vecr}{\bmr}
\safemath{\vecs}{\bms}
\safemath{\vect}{\bmt}
\safemath{\vecu}{\bmu}
\safemath{\vecv}{\bmv}
\safemath{\vecw}{\bmw}
\safemath{\vecx}{\bmx}
\safemath{\vecy}{\bmy}
\safemath{\vecz}{\bmz}

\safemath{\veczero}{\bmzero}
\safemath{\vecone}{\bmone}
\safemath{\vecxi}{\bmxi}
\safemath{\veclambda}{\bmlambda}
\safemath{\vecmu}{\bmmu}
\safemath{\vecnu}{\bmnu}

\safemath{\vecomega}{\bmomega}
\safemath{\vectheta}{\bmtheta}
\safemath{\vecphi}{\bmphi}
\safemath{\vecpi}{\bmpi}
\safemath{\vecalpha}{\bmalpha}

\safemath{\matA}{\bA}
\safemath{\matB}{\bB}
\safemath{\matC}{\bC}
\safemath{\matD}{\bD}
\safemath{\matE}{\bE}
\safemath{\matF}{\bF}
\safemath{\matG}{\bG}
\safemath{\matH}{\bH}
\safemath{\matI}{\bI}
\safemath{\matJ}{\bJ}
\safemath{\matK}{\bK}
\safemath{\matL}{\bL}
\safemath{\matM}{\bM}
\safemath{\matN}{\bN}
\safemath{\matO}{\bO}
\safemath{\matP}{\bP}
\safemath{\matQ}{\bQ}
\safemath{\matR}{\bR}
\safemath{\matS}{\bS}
\safemath{\matT}{\bT}
\safemath{\matU}{\bU}
\safemath{\matV}{\bV}
\safemath{\matW}{\bW}
\safemath{\matX}{\bX}
\safemath{\matY}{\bY}
\safemath{\matZ}{\bZ}
\safemath{\matzero}{\bmzero}

\safemath{\matDelta}{\bDelta}
\safemath{\matLambda}{\bLambda}
\safemath{\matPhi}{\bPhi}
\safemath{\matSigma}{\bSigma}
\safemath{\matOmega}{\bOmega}
\safemath{\matTheta}{\bTheta}

\safemath{\matidentity}{\matI}
\safemath{\matone}{\matO}

\safemath{\rnda}{A}
\safemath{\rndb}{B}
\safemath{\rndc}{C}
\safemath{\rndd}{D}
\safemath{\rnde}{E}
\safemath{\rndf}{F}
\safemath{\rndg}{G}
\safemath{\rndh}{H}
\safemath{\rndi}{I}
\safemath{\rndj}{J}
\safemath{\rndk}{K}
\safemath{\rndl}{L}
\safemath{\rndm}{M}
\safemath{\rndn}{N}
\safemath{\rndo}{O}
\safemath{\rndp}{P}
\safemath{\rndq}{Q}
\safemath{\rndr}{R}
\safemath{\rnds}{S}
\safemath{\rndt}{T}
\safemath{\rndu}{U}
\safemath{\rndv}{V}
\safemath{\rndw}{W}
\safemath{\rndx}{X}
\safemath{\rndy}{Y}
\safemath{\rndz}{Z}

\safemath{\rveca}{\bimA}
\safemath{\rvecb}{\bimB}
\safemath{\rvecc}{\bimC}
\safemath{\rvecd}{\bimD}
\safemath{\rvece}{\bimE}
\safemath{\rvecf}{\bimF}
\safemath{\rvecg}{\bimG}
\safemath{\rvech}{\bimH}
\safemath{\rveci}{\bimI}
\safemath{\rvecj}{\bimJ}
\safemath{\rveck}{\bimK}
\safemath{\rvecl}{\bimL}
\safemath{\rvecm}{\bimM}
\safemath{\rvecn}{\bimN}
\safemath{\rveco}{\bomO}
\safemath{\rvecp}{\bimP}
\safemath{\rvecq}{\bimQ}
\safemath{\rvecr}{\bimR}
\safemath{\rvecs}{\bimS}
\safemath{\rvect}{\bimT}
\safemath{\rvecu}{\bimU}
\safemath{\rvecv}{\bimV}
\safemath{\rvecw}{\bimW}
\safemath{\rvecx}{\bimX}
\safemath{\rvecy}{\bimY}
\safemath{\rvecz}{\bimZ}

\safemath{\rvecxi}{\bmxi}
\safemath{\rveclambda}{\bmlambda}
\safemath{\rvecmu}{\bmmu}
\safemath{\rvectheta}{\bmtheta}
\safemath{\rvecphi}{\bmphi}

\safemath{\rmatA}{\bimA}
\safemath{\rmatB}{\bimB}
\safemath{\rmatC}{\bimC}
\safemath{\rmatD}{\bimD}
\safemath{\rmatE}{\bimE}
\safemath{\rmatF}{\bimF}
\safemath{\rmatG}{\bimG}
\safemath{\rmatH}{\bimH}
\safemath{\rmatI}{\bimI}
\safemath{\rmatJ}{\bimJ}
\safemath{\rmatK}{\bimK}
\safemath{\rmatL}{\bimL}
\safemath{\rmatM}{\bimM}
\safemath{\rmatN}{\bimN}
\safemath{\rmatO}{\bimO}
\safemath{\rmatP}{\bimP}
\safemath{\rmatQ}{\bimQ}
\safemath{\rmatR}{\bimR}
\safemath{\rmatS}{\bimS}
\safemath{\rmatT}{\bimT}
\safemath{\rmatU}{\bimU}
\safemath{\rmatV}{\bimV}
\safemath{\rmatW}{\bimW}
\safemath{\rmatX}{\bimX}
\safemath{\rmatY}{\bimY}
\safemath{\rmatZ}{\bimZ}

\safemath{\rmatDelta}{\bimDelta}
\safemath{\rmatLambda}{\bimLambda}
\safemath{\rmatPhi}{\bimPhi}
\safemath{\rmatSigma}{\bimSigma}
\safemath{\rmatOmega}{\bimOmega}
\safemath{\rmatTheta}{\bimTheta}

\usepackage{amsfonts}
\usepackage{mathrsfs}
\usepackage{xspace}
\usepackage{bm}

\newenvironment{textbmatrix}{	\setlength{\arraycolsep}{2.5pt}%
								\big[\begin{matrix}}{\end{matrix}\big]%
								\raisebox{0.08ex}{\vphantom{M}}}

\def\be{\begin{equation}}
\def\ee{\end{equation}}
\def\een{\nonumber \end{equation}}
\def\mat{\begin{bmatrix}}
\def\emat{\end{bmatrix}}
\def\btm{\begin{textbmatrix}}
\def\etm{\end{textbmatrix}}

\def\ba#1\ea{\begin{align}#1\end{align}}
\def\bas#1\eas{\begin{align*}#1\end{align*}}
\def\bs#1\es{\begin{split}#1\end{split}} 
\def\bg#1\eg{\begin{gather}#1\end{gather}} 
\def\bi#1\ei{\begin{itemize}#1\end{itemize}}

\newcommand{\lefto}{\mathopen{}\left}

\DeclareMathOperator*{\argmin}{arg\;min}		%
\DeclareMathOperator{\Prob}{\opP}			%
\DeclareMathOperator{\Exop}{\opE}			%
\DeclareMathOperator{\Varop}{\opV\!\mathrm{ar}} %

\DeclareMathOperator{\landauO}{\mathcal{O}}

\newcommand{\Var}[1]{\ensuremath{\Varop\lefto[#1\right]}} %

\newcommand{\ind}[1]{\mathbbm{1}_{\{#1\}}}

\safemath{\dirac}{\delta}					%
\safemath{\krond}{\dirac}					%
\safemath{\upto}{\uparrow}
\safemath{\downto}{\downarrow}
\safemath{\iu}{j}							%
\safemath{\ev}{\lambda}						%
\safemath{\hilseqspace}{l^{2}}				%
\newcommand{\banachfunspace}[1]{\setL^{#1}}	%
\safemath{\hilfunspace}{\banachfunspace{2}}	%

\safemath{\SNR}{\text{\sc snr}} 				%
\safemath{\No}{N_0}							%
\safemath{\Es}{E_s}							%
\safemath{\Eb}{E_b}							%
\safemath{\EbNo}{\frac{\Eb}{\No}}
\safemath{\EsNo}{\frac{\Es}{\No}}

\DeclareMathOperator{\CHop}{\ensuremath{\opH}} %
\safemath{\tvir}{\rndh_{\CHop}}				%
\safemath{\tvtf}{\rndl_{\CHop}}				%
\safemath{\spf}{\rnds_{\CHop}}				%
\safemath{\bff}{H_{\CHop}}					%

\safemath{\ircf}{r_{h}}						%
\safemath{\tftvcf}{r_{s}}					%
\safemath{\tfcf}{r_{l}}						%
\safemath{\bfcf}{r_{H}}						%

\safemath{\tcorr}{c_h}						%
\safemath{\scf}{c_{s}}						%
\safemath{\tfcorr}{c_{l}}					%
\safemath{\fcorr}{c_{H}}						%

\safemath{\mi}{I}							%
\safemath{\capacity}{C}						%

\safemath{\normal}{\mathcal{N}}			%
\safemath{\jpg}{\mathcal{CN}}			%
\safemath{\mchain}{\leftrightarrow}		%

\safemath{\dB}{\,\mathrm{dB}}
\safemath{\dBm}{\,\mathrm{dBm}}
\safemath{\Hz}{\,\mathrm{Hz}}
\safemath{\kHz}{\,\mathrm{kHz}}
\safemath{\MHz}{\,\mathrm{MHz}}
\safemath{\GHz}{\,\mathrm{GHz}}
\safemath{\s}{\,\mathrm{s}}
\safemath{\ms}{\,\mathrm{ms}}
\safemath{\mus}{\,\mathrm{\mu s}}
\safemath{\ns}{\,\mathrm{ns}}
\safemath{\meter}{\,\mathrm{m}}
\safemath{\mm}{\,\mathrm{mm}}
\safemath{\cm}{\,\mathrm{cm}}
\safemath{\m}{\,\mathrm{m}}
\safemath{\W}{\,\mathrm{W}}
\safemath{\J}{\,\mathrm{J}}
\safemath{\K}{\,\mathrm{K}}
\safemath{\bit}{\,\mathrm{bit}}

\safemath{\define}{=}			%

\safemath{\equivalent}{\sim}
\safemath{\distas}{\sim}					%
\safemath{\sdiff}{\Delta}				%

\safemath{\reals}{\mathbb{R}}
\safemath{\positivereals}{\reals_{+}}
\safemath{\integers}{\mathbb{Z}}
\safemath{\posint}{\integers_{+}}
\safemath{\naturals}{\mathbb{N}}
\safemath{\posnaturals}{\naturals_{+}}
\safemath{\complexset}{\mathbb{C}}
\safemath{\rationals}{\mathbb{Q}}

\newcommand{\dsp}[1]{\Delta_{#1}} %
\newcommand{\psp}{\Pi} %
\newcommand{\ssp}{\cS} %
\newcommand{\asp}{\cA} %

\newcommand{\rwd}{R} %
\newcommand{\rwdpol}{\overline{R}} %
\newcommand{\dnm}{P} %
\newcommand{\dnmpol}{\overline{P}} %
\newcommand{\empdist}[1]{\widehat{\mu}_{#1}} %

\newcommand{\gpop}[2]{\Gamma_{\dnm}(#1, #2)} %
\newcommand{\gpopind}[3]{\mathbf{\Gamma}_{\dnm}^{#3}(#1, #2)} 
\newcommand{\poldif}[1]{\Delta_{#1}} %

\newcommand{\lmu}{L_\mu}
\newcommand{\ls}{L_s}
\newcommand{\la}{L_a}
\newcommand{\kmu}{K_\mu}
\newcommand{\ka}{K_a}
\newcommand{\ks}{K_s}

\newcommand{\lpopmu}{L_{pop,\mu}}

\newcommand{\R}{\mathbb{R}}

\newcommand{\eps}{\varepsilon}
\newcommand{\eqdef}{ := }
\newcommand{\setsz}[1]{|#1|}

\newcommand{\cA}{\mathcal{A}}

\newcommand{\cF}{\mathcal{F}}

\newcommand{\cL}{\mathcal{L}}

\newcommand{\cO}{\mathcal{O}}

\newcommand{\cS}{\mathcal{S}}

\newcommand{\cX}{\mathcal{X}}

\newcommand{\EEN}[1]{\mathbb{E}_N\left[#1\right]}
\newcommand{\EEinf}[1]{\mathbb{E}_\infty\left[#1\right]}

\newcommand{\EE}[1]{\mathbb{E}\left[#1\right]}

\newcommand{\EEc}[2]{\mathbb{E}\left[#1\left|#2\right.\right]}

\renewcommand{\P}{\mathbb{P}}
\newcommand{\PP}[1]{\mathbb{P}\left(#1\right)}

\newcommand{\Law}{\cL}

\newcommand{\sgauss}[1]{SG(#1)}

\newcommand{\JfinNi}[1]{J_{P, R}^{H,N,(#1)}}
\newcommand{\ExpfinNi}[1]{\mathcal{E}_{P, R}^{H,N,(#1)}}

\usepackage{enumerate}
\usepackage{enumitem}

\newcommand{\RL}{\mathcal{R}}
\newcommand{\PL}{\mathcal{P}}

\newcommand{\RSim}{\mathcal{R}^{\text{Sim}}}
\newcommand{\PSim}{\mathcal{P}^{\text{Sim}}}

\newcommand{\RLin}{\mathcal{R}^{\text{Lin}}}
\newcommand{\PLin}{\mathcal{P}^{\text{Lin}}}

\newcommand{\Vstat}{V^\gamma_{P,R}}
\newcommand{\Vfin}{V_{P, R}^H}
\newcommand{\Expfin}{\setE_{P, R}^H}
\newcommand{\Lambdaop}{\Lambda^{H}_{P}}

\newcommand{\JstatN}{J_{P, R}^{\gamma,N,(i)}}
\newcommand{\JstatNi}[1]{J_{P, R}^{\gamma,N,(#1)}}
\newcommand{\JfinN}{J_{P, R}^{H,N,(i)}}
\newcommand{\ExpfinN}{\mathcal{E}_{P, R}^{H,N,(i)}}

\newcommand{\CompFHLinear}[1]{$#1$-\textsc{FH-Linear}}
\newcommand{\CompFH}[1]{$(\varepsilon, #1)$-\textsc{FH-Nash}}
\newcommand{\CompStat}{$\varepsilon$-\textsc{StatDist}}

\newcommand{\CompTwoNash}{$2$-\textsc{Nash}}
\newcommand{\CompEOL}{\textsc{End-of-the-Line}}
\newcommand{\CompGC}{$\varepsilon$-\textsc{GCircuit}}

\newcommand{\CompGCeps}[1]{$\text{#1}$-\textsc{GCircuit}}
\newcommand{\CompFHeps}[2]{$(\text{#1}, \text{#2})$-\textsc{FH-Nash}}
\newcommand{\CompStateps}[1]{$\text{#1}$-\textsc{StatDist}}

\newcommand{\ppad}{\textsc{PPAD}}
\newcommand{\ppadcomplete}{\textsc{PPAD}-complete}
\newcommand{\ppadhard}{\textsc{PPAD}-hard}
\newcommand{\polytime}{\textsc{P}}

\newcommand{\BibTeX}{\rm B\kern-.05em{\sc i\kern-.025em b}\kern-.08em\TeX}

\usepackage{pifont}

\newcommand{\fProb}{\omega_\epsilon}
\newcommand{\fRew}{\vecg}
\newcommand{\fRep}{\vech}

\makeatletter
\gdef\@copyrightpermission{
	\begin{minipage}{0.3\columnwidth}
		\href{https://creativecommons.org/licenses/by/4.0/}{\includegraphics[width=0.90\textwidth]{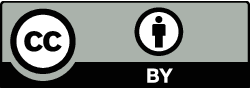}}
	\end{minipage}\hfill
	\begin{minipage}{0.7\columnwidth}
		\href{https://creativecommons.org/licenses/by/4.0/}{This work is licensed under a Creative Commons Attribution International 4.0 License.}
	\end{minipage}
	\vspace{5pt}
}
\makeatother

\begin{document}

\pagestyle{fancy}
\fancyhead{}

\maketitle

\section{Introduction}

Multi-agent reinforcement learning (MARL) finds numerous impactful applications in the real world \cite{shavandi2022multi, wiering2000multi,samvelyan2019starcraft, rashedi2016markov, matignon2007hysteretic, mao2022mean}.
Despite the urgent need in practice, MARL remains a fundamental challenge, especially in the setting with large numbers of agents due to the so-called ``curse of many agents'' \citep{wang2020breaking}.

Mean-field games (MFG), a theoretical framework first proposed by \citet{lasry2007mean} and \citet{huang2006large}, permits the theoretical study of such large-scale games by introducing mean-field simplification.
Under certain assumptions, the mean-field approximation leads to efficient algorithms for the analysis of a particular type of $N$-agent competitive game where there are symmetries between players and when $N$ is large. Such games appear widely in for instance auctions \cite{iyer2014mean}, and cloud resource management \cite{mao2022mean}.
For the mean-field analysis,  the game dynamics with $N$-players must be \emph{symmetric} (i.e., each player must be exposed to the same rules) and \emph{anonymous} (i.e., the effect of each player on the others should be permutation invariant).
Under this simplification, works such as \cite{perrin2020fictitious, anahtarci2022q, guo2019learning, yardim2023policy, perolat2022scaling, xie2021learning, cui2021approximately} and many others have analyzed reinforcement learning (RL) algorithms in the MFG limit $N\rightarrow\infty$ to obtain a tractable approximation of many agent games, providing learning guarantees under various structural assumptions.

Being a simplification, MFG formulations should ideally satisfy two desiderata: 
(1) they should be \emph{relevant}, i.e., they are good approximations of the original MARL problem and
(2) they should be \emph{tractable}, i.e., they are at least easier than solving the original MARL problem.
In this work, we would like to understand the extent to which MFGs satisfy these two requirements, and 
we aim to answer two natural questions that remain understudied:
\begin{itemize}
    \item \emph{When are MFGs good approximations of the finite player games, when are they not?} 
    In particular, are polynomially many agents always sufficient for mean-field approximation to be effective?
    \item \emph{Is solving  MFGs always computationally tractable, or more tractable than directly solving the $N$-player game?} 
    In particular, can MFGs be solved in polynomial or pseudo-polynomial time?
\end{itemize}

\subsection{Related Work}
Mean-field RL has been studied in various mathematical settings. 
In this work, we focus on two popular formulations in particular: stationary mean-field games (Stat-MFG, see e.g. \cite{anahtarci2022q, guo2019learning}) and finite-horizon MFG (FH-MFG, see e.g. \cite{perrin2020fictitious, perolat2022scaling}). 
In the Stat-MFG setting the objective is to find a stationary policy that is optimal with respect to its induced stationary distribution, while in the FH-MFG setting, a finite-horizon reward is considered with a time-varying policy and population distribution.

\textbf{Existing results on MFG relevance/approximation.}
The approximation properties of MFGs have been explored by several works in literature, as summarized in Table~\ref{tab:selected_works_approx}.
Finite-agent approximation bounds have been widely analyzed in the case of stochastic mean-field differential games \cite{carmona2013probabilistic, carmona2018probabilistic}, albeit in the differential setting and without explicit lower bounds.
Recent works~\cite{anahtarci2022q, cui2021approximately} have established that Stat-MFG Nash equilibria (Stat-MFG-NE) asymptotically approximate the NE of $N$-player symmetric dynamic games under continuity assumptions.
The result by \citet{saldi2018markov}, as the basis of subsequent proofs, shows asymptotic convergence for a large class of MFG variants and only requires continuity of dynamics and rewards as well as minor technical assumptions such as compactness and a form of local Lipschitz continuity. 
However, such asymptotic convergence guarantees leave the question unanswered if the MFG models are realistic in real-world games.
Many games such as traffic systems, financial markets, etc. naturally exhibit large $N$, however, if $N$ must be astronomically large for good approximation, the real-world impact of the mean-field analysis will be limited.
Recently, \cite{yardim2023stateless} provided finite-agent approximation bounds of a special class of stateless MFG, which assumes no state dynamics.
We complement existing work on approximation properties of both Stat-MFG and FH-MFG by providing explicit upper and lower bounds for approximation.

\textbf{Existing results on MFG tractability.} The tractability of solving MFGs as a proxy for MARL has been also heavily studied in the RL community under various classes of structural assumptions.
Since finding approximate Nash equilibria for normal form games is \ppadcomplete{}, a class believed to be computationally intractable \cite{daskalakis2009complexity, chen2009settling}, solving the mean-field approximation in many cases can be a tractable alternative.
We summarize recent work for computationally (or statistically) solving the two types of MFGs below, with an in-depth comparison also provided in Table~\ref{tab:selected_works_comp}.

For Stat-MFG,  under a contraction assumption RL algorithms such as Q-learning \cite{zaman2023oracle, anahtarci2022q}, policy mirror ascent \cite{yardim2023policy}, policy gradient methods \cite{guo2022general}, soft Q-learning \cite{cui2021approximately} and fictitious play \cite{xie2021learning} have been shown to solve Stat-MFG with statistical and computational efficiency.
However, all of these guarantees require the game to be heavily regularized as pointed out in \cite{cui2021approximately, yardim2023policy}, inducing a non-vanishing bias on the computed Nash.
Moreover, in some works the population evolution is also implicitly required to be contractive under all policies (see e.g. \cite{guo2019learning, yardim2023policy}), further restricting the analysis to sufficiently smooth games.
While \cite{guo2022mf} has proposed a method that guarantees convergence to MFG-NE under differentiable dynamics, the algorithm converges only when initialized sufficiently close to the solution.
To the best of our knowledge, there are neither RL algorithms that work without regularization nor evidence of difficulty in the absence of such strong assumptions: we complement the line of work by showing that unless dynamics are sufficiently smooth, Stat-MFG is both computationally intractable and a poor approximation.

A separate line of work analyzes the finite horizon problem. 
In this case, when the dynamics are population-independent and the payoffs are monotone the problem is known to be tractable.
Algorithms such as fictitious play \cite{perrin2020fictitious} and mirror descent \cite{perolat2022scaling} have been shown to converge to Nash in corresponding continuous-time equations.
Recent work has also focused on the statistical complexity of the finite-horizon problem in very general FH-MFG problems \cite{huang2023statistical}, however, the algorithm proposed is in general computationally intractable.
In terms of computational tractability and the approximation properties, our work complements these results by demonstrating that (1) when dynamics depend on the population as well an exponential approximation lower bound exists, and (2) in the absence of monotonicity, the FH-MFG is provably as difficult as solving an $N$-player game.

Finally, we note that there are several other settings and MFG solution concepts have been analyzed.
For instance, a certain class of infinite horizon MFG has been shown to be equivalent to concave utility RL, proving finite-time computational guarantees \cite{geist2021concave}.

\subsection {Our Contribution} 
In this work, we formalize and provide answers to the two aforementioned fundamental questions, first focusing on the approximation properties of MFG in Section~\ref{sec:main_relevance} and later on the computational tractability of MFG in Section~\ref{sec:main_complexity}.  
Our contributions are summarized as follows.

Firstly, we introduce explicit finite-agent approximation bounds for finite horizon and stationary MFGs (Table~\ref{tab:selected_works_approx}) in terms of exploitability in the finite agent game.
In both cases, we prove explicit upper bounds which quantify how many agents a symmetric game must have to be well-approximated by the MFG, which has been absent in the literature to the best of our knowledge. 
Our approximation results only require a minimal Lipschitz continuity assumption of the transition kernel and rewards.
For FH-MFG, we prove a $\mathcal{O}\left(\frac{(1 - L^H) H^2}{(1-L)\sqrt{N}}\right)$ upper bound for the exploitabilty where $L$ is the Lipschitz modulus of the population evolution operator: the upper bound exhibits an exponential dependence on the horizon $H$.
For the Stat-MFG we show that a $\mathcal{O}\left(\frac{(1-\gamma)^{-3}}{\sqrt{N}}\right)$ approximation bound can be established, but only if the population evolution dynamics are non-expansive.
Next, for the first time, we establish explicit lower bounds for the approximation proving the shortcomings of the upper bounds are fundamental.
For the FH-MFG, we show that unless $N \geq \Omega(2^H)$, an exploitability linear in horizon $H$ is unavoidable when deploying the MFG solution to the $N$ player game: hence in general the MFG equilibrium becomes irrelevant quickly as the problem horizon increases.
For Stat-MFG we establish an $\Omega(N^{\log_2 \gamma})$ lower bound when the population dynamics are not restricted to non-expansive population operators, showing that a large discount factor $\gamma$ also rapidly deteriorates the approximation efficiency.
Our lower bounds indicate that in the worst case, the number of agents required for the approximation can grow exponentially in the problem parameters, demonstrating the limitations of the MFG approximation.

Finally, from the computational perspective, we establish that both finite-horizon and stationary MFGs can be \ppadcomplete{} problems in general, even when restricted to certain simple subclasses (Table~\ref{tab:selected_works_comp}).
This shows that both MFG problems are in general as hard as finding a Nash equilibrium of $N$-player general sum games.
Furthermore, our results imply that unless \ppad{}=\polytime{} there are no polynomial time algorithms for solving FH-MFG and Stat-MFG, a result indicating computational intractability.

\begin{table*}[t]
  \begin{tabular}{llll} \toprule
    \textbf{Work} & \textbf{MFG type} & \textbf{Key Assumptions} & \textbf{Approximation Rate (in Exploitability)} \\ \midrule
    \citeauthor{carmona2013probabilistic}, \citeyear{carmona2013probabilistic} & Other\textsuperscript{a} & Affine drift, Lipschitz derivatives & $\mathcal{O}(N^{-1 /(d+4)})$ ($d$ dimension of state space)\\
    \citeauthor{saldi2018markov}, \citeyear{saldi2018markov} & Other\textsuperscript{b} & Continuity & $o(1)$ (asymptotic: convergence as $N\rightarrow \infty$) \\
     \citeauthor{anahtarci2022q}, \citeyear{anahtarci2022q} & Stat-MFG & Lipschitz $P,R$ + Regularized + Contractive $\Gamma_P$ & $o(1)$ (asymptotic: convergence as $N\rightarrow \infty$) \\
    \citeauthor{cui2021approximately}, \citeyear{cui2021approximately} & Stat-MFG & Continuity & $o(1)$ (asymptotic: convergence as $N\rightarrow \infty$) \\
    \citeauthor{yardim2023stateless}, \citeyear{yardim2023stateless} & Other\textsuperscript{c} & Lipschitz $P,R$ & $\mathcal{O}(\sfrac{1}{\sqrt{N}})$  \\
     \midrule
    \textbf{Theorem~\ref{theorem:upper_approx_fin}} & FH-MFG & Lipschitz $P,R$ & $\cO\left( \frac{H^2(1-L^H)}{(1-L)\sqrt{N}}\right)$, $L$ Lipschitz modulus of $\Gamma_P$ \\
    \textbf{Theorem~\ref{theorem:lower_approx_fin}} & FH-MFG & Lipschitz $P,R$ & $\Omega(H)$ unless $N \geq \Omega(2^H)$\\
    \textbf{Theorem~\ref{theorem:upper_approx_stat}}  & Stat-MFG & Lipschitz $P,R$ + Non-expansive $\Gamma_P$ & $\cO(\sfrac{(1-\gamma)^{-3}}{\sqrt{N}})$ \\
    \textbf{Theorem~\ref{theorem:lower_approx_stat}}  & Stat-MFG & Lipschitz $P,R$  & $\Omega(N^{-\log_2 \gamma^{-1}}))$ \\
    \bottomrule
  \end{tabular}
  \caption{Selected approximation results for MFG. Notes: \textsuperscript{a} stochastic differential MFG, \textsuperscript{b} infinite-horizon discounted setting with non-stationary policies, \textsuperscript{c} stateless/static MFG setting. } 
  \label{tab:selected_works_approx}
\end{table*}

\begin{table*}[t]
  \centering
  \begin{tabular}{llllll} \toprule
    \textbf{Work} & \textbf{MFG Type}  & \textbf{Key Assumptions} & \textbf{Iteration/Sample Complexity result} \\ \midrule
    \citeauthor{anahtarci2022q}, \citeyear{anahtarci2022q} & Stat-MFG  & Lipschitz $P,R$ + Regularization + Contractive $\Gamma_P$  & $\widetilde{\mathcal{O}}(\varepsilon^{-4 |\setA|})$ samples, $\mathcal{O}(\log \varepsilon^{-1})$ iterations \\
    \citeauthor{geist2021concave}, \citeyear{geist2021concave} & Other\textsuperscript{a} & Concave potential  & $\mathcal{O}(\varepsilon^{-2})$ iterations \\
    \citeauthor{perrin2020fictitious}, \citeyear{perrin2020fictitious} & FH-MFG & Monotone $R$, $\mu$-independent $P$  & $\mathcal{O}(\varepsilon^{-1})$ (continuous time analysis) \\
    \citeauthor{perolat2022scaling}, \citeyear{perolat2022scaling} & FH-MFG  & Monotone $R$, $\mu$-independent $P$ & $\mathcal{O}(\varepsilon^{-1})$ (continuous time analysis)\\
    \citeauthor{zaman2023oracle}, \citeyear{zaman2023oracle} & Stat-MFG  & Lipschitz $P,R$ + Regularization + Contractive $\Gamma_P$ & $\mathcal{O}(\varepsilon^{-4})$ samples \\
    \citeauthor{cui2021approximately}, \citeyear{cui2021approximately} & Stat-MFG  & Lipschitz $P,R$ + Regularization & $\mathcal{O}(\log \varepsilon^{-1})$ iterations\\
    \citeauthor{yardim2023stateless}, \citeyear{yardim2023stateless} & Other\textsuperscript{b}  & Monotone and Lipschitz $R$ & $\mathcal{O}(\varepsilon^{-2})$ samples ($N$-player) \\
    \citeauthor{yardim2023policy}, \citeyear{yardim2023policy} & Stat-MFG  & Lipschitz $P,R$ + Regularization + Contractive $\Gamma_P$ & $\mathcal{O}(\varepsilon^{-2})$ samples ($N$-player) \\
    \midrule
    \textbf{Theorem~\ref{theorem:compstat_ppad}} & Stat-MFG & Lipschitz $P,R$  & \ppadcomplete{} \\
    \textbf{Theorem~\ref{theorem:compfh_ppad}} & FH-MFG  & Lipschitz $P,R$ + $\mu$-independent $P$  & \ppadcomplete{} \\
    \textbf{Theorem~\ref{theorem:compfh_linear_ppad}} & FH-MFG  & Linear $P,R$ + $\mu$-independent $P$  & \ppadcomplete{} \\
    \bottomrule
  \end{tabular}
  \caption{Selected results for computing MFG-NE from literature. 
  In the assumptions column, contractive $\Gamma_P$ indicates that for all $\pi\in\Pi$, $\Gamma_P(\cdot, \pi)$ is a contraction, and regularization indicates that a non-vanishing bias is present.
  Notes: \textsuperscript{a} infinite-horizon, population dependence through the discounted state distribution. \textsuperscript{b} stateless/static MFG.}
  \label{tab:selected_works_comp}
\end{table*}

\section{Mean-Field Games: Definitions, Solution Concepts}

\paragraph{Notation.}
Throughout this work, we assume $\setS, \setA$ are finite sets.
For a finite set $\setX$, $\Delta_\setX$ denotes the set of probability distributions on $\setX$.
The norm used will not fundamentally matter for our results, we choose to equip $\Delta_\setS, \Delta_\setA$ with the norm $\| \cdot \|_1$.
We define the set of Markov policies $\Pi := \{ \pi:\setS \rightarrow \Delta_\setA\}$, $\Pi_H := \{ \{ \pi_h\}_{h=0}^{H-1} : \pi_h \in \Pi, \forall h\}$ and $\Pi_H^N := \{ \{ \pi_h^i\}_{h=0, i = 0}^{H-1, N} : \pi_h^i \in \Pi, \forall h\}$. 
For policies $\pi,\pi'\in\psp$ denote $\|\pi-\pi'\|_1=\sup_{s\in\ssp}\|\pi(\cdot|s)-\pi'(\cdot|s)\|_1$. 
We denote $d(x,y) \eqdef \ind{x\neq y}$ for $x,y$ in $\setA$ or $\setS$.
For $\vecpi\in\Pi^N, \pi' \in \Pi$, we define $( \pi', \vecpi^{ -i}) \in\Pi^N$ as the policy profile where the $i$-th policy has been replaced by $\pi'$.
Likewise, for $\vecpi\in\Pi^N_H, \vecpi' \in \Pi_H$, we denote by $( \vecpi', \vecpi^{ -i}) \in\Pi^N_H$ the policy profile where the $i$-th player's policy has been replaced by $\vecpi'$.
For any $N \in \mathbb{N}_{\geq 0}$, $[N] := \{1, \ldots, N \}$.

MFGs introduce a dependence on the population distribution over states of the rewards and dynamics.
We will strictly consider Lipschitz continuous rewards and dynamics, which is a common assumption in literature \cite{guo2019learning, anahtarci2022q, yardim2023policy, xie2021learning}, formalized below.

\begin{definition}[Lipschitz dynamics, rewards]
    \label{def:dynamics_rewards}
    For some $L \geq 0$, we define the set of $L$-Lipschitz reward functions and state transition dynamics as
    \begin{align*}
        \RL_L := \Big\{ R: \setS \times \setA \times \Delta_\setS \rightarrow &[0,1] \, : \, |R(s,a, \mu) - R(s,a, \mu')| \\
            &\leq L \| \mu - \mu'\|_1, \forall s,a,\mu, \mu' \Big\}, \\
        \PL_L := \Big\{ P : \setS \times \setA \times \Delta_\setS \rightarrow &\Delta_\setS\, : \, \| P(s,a,\mu) - P(s,a,\mu')\|_1 \\
            &\leq L \| \mu - \mu' \|_1, \forall s,a,\mu, \mu' \Big\}.
    \end{align*}
    Moreover, we define the set of Lipschitz rewards and dynamics as $\RL := \bigcup_{L \geq 0} \RL_L, \quad \PL := \bigcup_{L \geq 0} \PL_L$ respectively.
\end{definition}

We note that there are interesting MFGs with non-Lipschitz dynamics and rewards, however, even the existence of Nash is not guaranteed in this case.
Lipschitz continuity is a minimal assumption under which solutions to MFG always exist, and as our aim is to prove lower bounds and difficulty we will adopt this assumption.
Solving MFG with non-Lipschitz dynamics is more challenging than Lipschitz continuous MFG (the latter being a subset of the former), hence our difficulty results will apply.

\paragraph{Operators.}
We will define the useful population operators $\Gamma_P: \Delta_\setS \times \Pi \rightarrow \Delta_\setS$, $\Gamma_P^H: \Delta_\setS \times \Pi \rightarrow \Delta_\setS$, and
$\Lambda^{H}_{P}: \Delta_\setS \times \Pi_{H} \rightarrow \Delta_\setS^{H}$ as 
\begin{align*}
    \Gamma_P(\mu, \pi)  &:= \sum_{s \in \setS, a \in \setA} \mu(s) \pi(a|s) P(\cdot|s,a,\mu), \\
    \Gamma^H_P(\mu, \pi) &:= \underbrace{\Gamma_P(\dots \Gamma_P(\Gamma_P(\mu, \pi), \pi) \dots), \pi)}_{H\text{ times}},\\
    \Lambdaop(\mu_0, \vecpi) &:= \big\{\underbrace{\Gamma_P(\dots \Gamma_P(\Gamma_P(\mu_0, \pi_0), \pi_1) \dots, \pi_{h-1})}_{h \text{ times}} \big\}_{h=0}^{H-1} 
\end{align*}
for all $n\in \mathbb{N}_{> 0}, \pi \in \Pi, \vecpi = \{ \pi_h \}_{h=0}^{H-1} \in \Pi_H, P\in \PL, \mu_0 \in \Delta_\setS$.

Finally, we will need the following Lipschitz continuity result for the $\Gamma_P$ operator.
\begin{lemma}{\cite[Lemma~3.2]{yardim2023policy}}
    Let $P \in \PL_{K_\mu}$ for $K_\mu >0$ and
    {\small
    \begin{align*}
        \ks \eqdef \sup_{\substack{s,s'\\ a,\mu}}
        \left\|
        \dnm(s,a,\mu)-\dnm(s',a,\mu)
        \right\|_1, 
        \ka \eqdef \sup_{\substack{a,a'\\ s,\mu}}
        \left\|
        \dnm(s,a,\mu)-\dnm(s,a',\mu)
        \right\|_1.
    \end{align*}}
    Then it holds for all $\mu, \mu' \in \Delta_\setS, \pi,\pi' \in \Pi$ that:
    \label{lemma:lipschitz_gpop}
    \[
    \|\gpop{\mu}{\pi}
    -
    \gpop{\mu'}{\pi'}\|_1
    \leq 
    \lpopmu \|\mu-\mu'\|_1
    +
    \frac{\ka}{2}\|\pi-\pi'\|_1,
    \]
    where $\lpopmu \eqdef (\kmu+\frac{\ks}{2}+\frac{\ka}{2})$ for all $\pi,\pi' \in\psp$, $\mu,\mu' \in\dsp{\ssp}$.
\end{lemma}
In particular, in our settings, Lemma~\ref{lemma:lipschitz_gpop} indicates that $\Gamma_P$ is always Lipschitz continuous if $P \in \setP$, a property which will become significant for approximation analysis.

We will be interested in two classes of MFG solution concepts that lead to different analyses: infinite horizon stationary MFG Nash equilibrium (Stat-MFG-NE) and finite horizon MFG Nash equilibrium (FH-MFG-NE).
The first problem widely studied in literature is the stationary MFG equilibrium problem, see for instance \cite{anahtarci2022q, yardim2023policy, guo2019learning, guo2022general,xie2021learning}.
We formalize this solution concept below.

\begin{definition}[Stat-MFG]
A stationary MFG (Stat-MFG) is defined by the tuple $(\setS, \setA, P, R, \gamma)$ for Lipschitz dynamics and rewards $P \in \PL, R \in \RL$, discount factor $\gamma \in (0,1)$.
For any $(\mu, \pi) \in \Delta_\setS \times \Pi$, we define the $\gamma$-discounted infinite horizon expected reward as
    \begin{align*}
        \Vstat (\mu, \pi) := \Exop \left[ \sum_{t=0}^\infty \gamma^t R(s_t, a_t, \mu) \middle| \substack{s_0 \sim \mu, \quad a_t \sim \pi(s_t) \\ s_{t+1} \sim P(s_t, a_t, \mu) }\right].
    \end{align*}
A policy-population pair $(\mu^*, \pi^*) \in \Delta_\setS \times \Pi$ is called a Stat-MFG Nash equilibrium if the two conditions hold:
    \begin{align*}
        \textit{Stability: } \quad &\mu^* = \Gamma_P(\mu^*, \pi^*), \notag \\
        \textit{Optimality: } \quad &\Vstat (\mu^*, \pi^*) = \max_{\pi \in \Pi} \Vstat(\mu^*, \pi). 
        \tag{Stat-MFG-NE}
    \end{align*}
\end{definition}

The second MFG concept that we will consider has a finite time horizon, and is also common in literature \cite{perolat2015approximate, perrin2020fictitious, lauriere2022scalable, huang2023statistical}.
In this case, the population distribution is permitted to vary over time, and the objective is to find an optimal non-stationary policy with respect to the population distribution it induces.
We formalize this problem and the corresponding solution concept below.

\begin{definition}[FH-MFG]
\label{def:fh_mfg}
A finite horizon MFG problem (FH-MFG) is determined by the tuple $(\setS, \setA, H, P, R, \mu_0)$ where $H \in \mathbb{Z}_{> 0}$, $P\in \PL, R\in \RL, \mu_0 \in \Delta_\setS$.
For $\vecpi = \{\pi_h\}_{h=0}^{H} \in \Pi _ H, \vecmu = \{ \mu_h \}_{h=0}^{H-1} \in \Delta_\setS^{H}$, define the expected reward and exploitability as
    \begin{align*}
        \Vfin \left( \vecmu, \vecpi \right) & := \Exop \left[ \sum_{h=0}^{H-1} R(s_h, a_h, \mu_h) \middle| \substack{s_0 \sim \mu_0 , \quad a_h \sim \pi_h(s_h)\\ s_{h+1} \sim P(s_h, a_h, \mu_h)} \right] , \\
       \Expfin (\vecpi) &:= \max_{\vecpi' \in \Pi^H} \Vfin( \Lambdaop (\mu_0, \vecpi), \vecpi') - \Vfin ( \Lambdaop (\mu_0, \vecpi), \vecpi ).
    \end{align*}
    Then, the FH-MFG Nash equilibrium is defined as:
    \begin{align}
        \textit{Policy } &\vecpi^* = \{ \pi^*_h\}_{h=0}^{H-1} \in \Pi _ H \text{ such that } \notag \\
        & \Expfin(\{\pi^*_h\}_{h=0}^{H-1}) = 0. \tag{FH-MFG-NE}
    \end{align}
\end{definition}

\section{Approximation Properties of MFG} \label{sec:main_relevance}

As established in literature, the reason the FH-MFG and Stat-MFG problems are studied is the fact that they can approximate the NE of certain symmetric games with $N$ players, establishing the main relevance of the formulations in the real world.
Such results are summarized in Table~\ref{tab:selected_works_approx}.

In this section, we study how efficient this convergence is and also related lower bounds.
For these purposes, we first define the corresponding \emph{finite-player} game of each mean-field game problem: to avoid confusion, we call these games \emph{symmetric anonymous dynamic games} (SAG).
Afterwards, for each solution concept, we will first establish (1) an upper bound on the approximation error (i.e. the exploitability) due to the mean-field, and (2) a lower bound demonstrating the worst-case rate.
We will present the main outlines of proofs, and postpone computation-intensive derivations to the supplementary material of the paper.

\subsection{Approximation Analysis of FH-MFG}\label{sec:main_relevance_fh}

Firstly, we define the finite-player game that is approximately solved by the FH-MFG-NE.

\begin{definition}[$N$-FH-SAG]
    \label{def:n_fh_mfg}
    An $N$-player finite horizon SAG ($N$-FH-SAG) is determined by the tuple $(N, \setS, \setA, H, P, R, \mu_0)$ such that $N\in \mathbb{Z}_{> 0}, H \in \mathbb{Z}_{> 0}$, $P\in \PL, R\in \RL, \mu_0 \in \Delta_\setS$.
For any $\vecpi = \{\pi_h^i\}_{h=0,\ldots,H-1, i \in [N]} \in \Pi_H^N$, we define the expected mean reward and exploitability of player $i$ as
    \begin{align*}
        \JfinN \left( \vecpi \right) & := \Exop \left[ \sum_{h=0}^{H-1} R(s_h^i, a_h^i, \widehat{\mu}_h) \middle| \substack{\forall j : s_0^j \sim \mu_0 , \quad a_h ^j \sim \pi_h^j(s_h^j)\\ s_{h+1}^j \sim P(s_h^j, a_h^j, \widehat{\mu}_h), \widehat{\mu}_h := \frac{1}{N}\sum_j \vece_{s^j_h}} \right] , \\
        \ExpfinN(\vecpi) &:= \max_{\vecpi' \in \Pi^H} \JfinN ( \vecpi', \vecpi^{ -i}) - \JfinN ( \vecpi ) .
    \end{align*}
    Then, the $N$-FH-SAG Nash equilibrium is defined as:
    \begin{align}
        \textit{$N$-tuple of policies } &\{ \pi^{(i),*}_h\}_{h=0}^{H-1} \in \Pi^N _ H \text{ such that } \notag \\
        \forall i: \ExpfinN &(\{\pi^*_h\}_{h=0}^{H-1}) = 0. \tag{$N$-FH-SAG-NE}
    \end{align}
    If instead $\ExpfinN(\vecpi) \leq \delta$ for all $i$, then $\vecpi$ is called a $\delta$-$N$-FH-SAG Nash equilibrium.
\end{definition}

The above definition corresponds to a real-world problem as the function $\JfinN$ expresses the expected total payoff of each player: hence a $\delta$-$N$-MFG-NE is a Nash equilibrium of a concrete $N$-player game in the traditional game theoretical sense.
Also, note that now in the definition transition probabilities and rewards depend on $\widehat{\mu}_h$ which is the $\cF(\{s_h^i\}_i)=\cF_h$-measurable random vector of the empirical state distribution at time $h$ of all agents.

Firstly, we provide a positive result well-known in literature: the $N$-FH-SAG is approximately solved by the FH-MFG-NE policy.
Unlike some past works, we establish an explicit rate of convergence in terms of $N$ and problem parameters.

\begin{theorem}[Approximation of $N$-FH-SAG]\label{theorem:upper_approx_fin}
    Let $(\setS, \setA, H, P, R, \mu_0)$ be a FH-MFG with $P \in \setP, R\in\setR$ and with a FH-MFG-NE $\vecpi^* \in \Pi_H$, and for any $N \in \mathbb{N}_{>0}$ let $\vecpi_N^* := (\underbrace{\vecpi^*,\ldots, \vecpi^*}_{\text{$N$ times}}) \in \Pi_H^N$.
    Let $L > 0$ be the Lipschitz constant of $\Gamma_P$ in $\mu$, and let $\setG_N := (N, \setS, \setA, H, P, R, \mu_0)$ be the corresponding $N$-player game.
    Then:
    \begin{enumerate}
        \item If $L = 1$, then for all $i\in [N]$, $\ExpfinN(\vecpi^*_N) \leq \landauO( \frac{H^3}{\sqrt{N}})$, that is, $\vecpi^*_N$ is a $\cO( \frac{H^3}{\sqrt{N}})$-NE of $\setG_N $. 
        \item If $L \neq 1$, then for all $i\in [N]$, $\ExpfinN(\vecpi^*_N) \leq \landauO\left( \frac{H^2 (1 - L^H)}{(1 - L) \sqrt{N}}\right)$, that is, $\vecpi^*_N$ is a $\landauO\left(\frac{H^2 (1 - L^H)}{(1 - L) \sqrt{N}}\right)$-NE of $\setG_N $. 
    \end{enumerate}
\end{theorem}
\begin{proof}\emph{(sketch)}
    Certain aspects of our proof will mirror the techniques introduced by \cite{saldi2018markov}, although we establish an explicit bound.
    We first bound the expected empirical population deviation given by $\Exop[\| \widehat{\mu}_h - \mu^{\vecpi}_h\|_1] = \mathcal{O} \left( \frac{L^h}{\sqrt{N}}\right)$ with an inductive concentration argument: at each step $h+1$, given past states $\widehat{\mu}_h$, the empirical distribution $\widehat{\mu}_h$ is a sum of $N$ independent identically distributed sub-Gaussian random variables. 
    Next, by utilizing the Lipschitz property of rewards and bounding deviation from the theoretical rewards the result follows in two computational steps: 
    (1) we show that $\left|\JfinNi{1}(\vecpi,\dots,\vecpi) - \Vfin(\Lambdaop(\mu_0, \vecpi), \vecpi)
    \right| \leq \mathcal{O}(\sfrac{1}{\sqrt{N}})$, and similarly
    (2) we show that for any policy sequence $\vecpi' \in \Pi_h$, we have $\left|\JfinNi{1}(\vecpi',\vecpi,\dots,\vecpi) - \Vfin(\Lambdaop(\mu_0,\vecpi), \vecpi')
    \right| \leq \mathcal{O}(\sfrac{1}{\sqrt{N}})$.
    The result follows by definition of exploitability, with explicit constants shown in the appendix.
\end{proof}

$\Gamma_P$ in Theorem~\ref{theorem:upper_approx_fin} is always $L$-Lipschitz in $\mu$ for some $L$ by Lemma~\ref{lemma:lipschitz_gpop}.
When $L>1$, the upper bound $\mathcal{O}\left(\sfrac{(1 + L^H)H^2}{\sqrt{N}}\right)$ has an exponential dependence on the Lipschitz constant of the operator $\Gamma_P$. 
However, for games with longer horizons, the upper bound might require an unrealistic amount of agents $N$ to guarantee a good approximation due to the exponential dependency.
Next, we establish a worst-case result demonstrating that this is not avoidable without additional assumptions.

\newcommand{\sleft}{s_{\text{Left}}}
\newcommand{\sright}{s_{\text{Right}}}
\newcommand{\slA}{s_{\text{LA}}}
\newcommand{\slB}{s_{\text{LB}}}
\newcommand{\srA}{s_{\text{RA}}}
\newcommand{\srB}{s_{\text{RB}}}

\newcommand{\actA}{a_{\text{A}}}
\newcommand{\actB}{a_{\text{B}}}

\begin{figure}[h!]
    \centering
    \begin{tikzpicture}[node distance={10mm}, thick, main/.style = {draw, circle}] 
        \node[main] (1) {$\sleft$}; 
        \node[main, draw=none] [right of=1] (9) {};
        \node[main, draw=none] [right of=9] (12) {};
        \node[main, draw=none] [right of=12] (7) {};
        \node[main, draw=none] [right of=7] (10) {};
        \node[main, draw=none] [right of=10] (11) {}; 
        \node[main] (2) [above of=7] {$\slB$}; 
        \node[main] (3) [above of=2] {$\slA$}; 
        \node[main] (4) [below of=7] {$\srA$}; 
        \node[main] (5) [below of=4] {$\srB$}; 
        \node[main] (6) [right of=11]  {$\sright$}; 
        \draw[->] (1) to [out=45,in=180,looseness=1.1] node[above,sloped] {$\ind{a=\actB}$} (2); 
        \draw[->] (1) to [out=90,in=180,looseness=1.1] node[near end, above,sloped] {$\ind{a=\actA}$} (3);
        \draw[->] (6) to [out=235,in=0,looseness=1.1] node[above,sloped] {$\ind{a=\actA}$} (4); 
        \draw[->] (6) to [out=270,in=0,looseness=1.1] node[near end, above,sloped] {$\ind{a=\actB}$} (5);
        \draw[->, orange] (5) to [out=310,in=290,looseness=1.1] (6);
        \draw[->, teal] (5) to [out=230,in=250,looseness=1.1] (1);
        \draw[->, orange] (4) to [out=70,in=200,looseness=1.1] (6);
        \draw[->, teal] (4) to [out=110,in=340,looseness=1.1] (1);
        \draw[->, orange] (2) to [out=310,in=160,looseness=1.1]  (6);
        \draw[->, teal] (2) to [out=230,in=20,looseness=1.1] (1);
        \draw[->, orange] (3) to [out=70,in=70,looseness=1.1] (6);
        \draw[->, teal] (3) to [out=110,in=110,looseness=1.1] (1);
    \end{tikzpicture} 
    \caption{Visualization of the counterexample. All orange edges have probability $\omega_\eps(\mu(\srA)+\mu(\srB))$, green edges have probability $\omega_\eps(\mu(\slA)+\mu(\slB))$ independent of action taken. Edges with probability $0$ are not drawn.}
    \label{fig:MFG}
\end{figure}
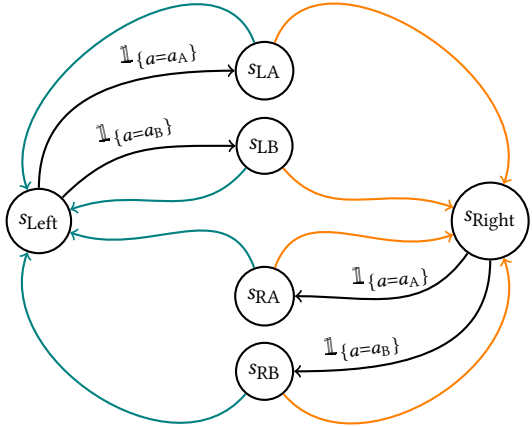

\begin{theorem}[Approximation lower bound for $N$-FH-SAG]\label{theorem:lower_approx_fin}
    There exists $\setS, \setA$ and $P \in \PL_8, R \in \RL_2, \mu_0 \in \Delta_\setS$ such that the following hold:
    \begin{enumerate}
        \item For each $H > 0$, the FH-MFG defined by $(\setS, \setA, H, P, R, \mu_0)$ has a \emph{unique} solution $\vecpi^*_H$  (up to modifications on zero-probability sets),
        \item For any $H, h > 0$, in the $N$-FH-SAG it holds that $\Exop_{H}[ \| \widehat{\mu}_h - \Lambdaop(\mu_0, \vecpi^*_H)_h\|_1] \geq \Omega \left(\min\{\ 1, \frac{2^H}{\sqrt{N}} \}\right)$.
        \item For any $H, N > 0$ either $N \geq \Omega(2^H)$, or for each player $i\in [N]$ it holds that $\ExpfinN (\vecpi^*_H, \ldots, \vecpi^*_H) \geq \Omega(H) $.
                
    \end{enumerate}
\end{theorem}
\begin{proof}\emph{(sketch)}
    We provide the basic idea of the proof and leave the cumbersome computations to the appendix.
    The proof is constructive: we construct an explicit FH-MFG where the statements hold, depicted in Figure~\ref{fig:MFG}.
    The FH-MFG will have 6 states and two actions defined as sets $\setS = \{ \sleft, \sright, \slA, \slB,\srA, \srB \}$ and $
    \setA = \{\actA, \actB \}.$    
    We define the initial state distribution with $\vecmu_0(\sleft) = \vecmu_0(\sright) = \sfrac{1}{2}$.
    The colored state transition probabilities are given by the function:
    \begin{align*}
        \fProb(x) &= \begin{cases}
            1, \quad x > \sfrac{1}{2} + \epsilon \\
            0, \quad x < \sfrac{1}{2} - \epsilon \\
            \frac{1}{2} + \frac{x - \sfrac{1}{2}}{2\epsilon}, \quad x \in [\sfrac{1}{2} - \epsilon, \sfrac{1}{2} + \epsilon]
        \end{cases}.
    \end{align*}
    The uniform policy over all actions $\vecpi^*$ at all states will be the unique FH-MFG-NE for all $H$, and the mean-field population distribution for all even $h$ will be $\mu_h^*(\sleft) = \mu^*_h(\sright) = \sfrac{1}{2}$.
    However, for finite $N$, using an anti-concentration bound on the binomial, we can show that with probability at least $\sfrac{1}{10}$, $\| \mu_0^* - \widehat{\mu}_h \|_1 \geq \sfrac{1}{\sqrt{N}}$.
    Using the fact that $\omega_\epsilon$ is $(2\epsilon)^{-1}$-expansive in the interval $[\sfrac{1}{2} - \epsilon, \sfrac{1}{2} + \epsilon]$, we can then show that the empirical population distribution exponentially diverges from the mean-field, that is $\Exop[\| \mu_{2h}^* - \widehat{\mu}_{2h} \|_1] \geq \Omega(\sfrac{5^h}{\sqrt{N}})$ until time $K:= \log_5 \sqrt{N}$.
    Moreover, with a series of concentration bounds, it can be shown that within an expected number of $\landauO(\log N)$ steps, all agents will converge to either $\sleft$ or $\sright$ during even rounds.
    Only the colored transitions are defined to have non-zero rewards, whose definition (provided in the supplementary) guarantees that the exploitability suffered scales linearly with $H$ after $N$ agents concentrate on the same state in even steps.
\end{proof}

This result shows that without further assumptions, the FH-MFG solution might suffer from exponential exploitability in $H$ in the $N$-player game.
In such cases, to avoid the concrete $N$-player game from deviating from the mean-field behavior too fast, either $H$ must be small or $P$ must be sufficiently smooth in $\mu$. 
We note that the typical assumption in the finite-horizon setting that $P \in \PL_0$ (see e.g. \cite{perrin2020fictitious, geist2021concave}) avoids this lower bound since in this case $\Gamma_P(\cdot, \pi)$ is simply multiplication by a stochastic matrix which is always non-expansive ($L=1$).
We also note at the expense of simplicity a stronger counter-example inducing exploitability $\Omega(H)$ unless $N \geq \Omega((L-\epsilon)^H)$ for all $\epsilon>0$ can be constructed, where $P\in\PL_L$.

\paragraph{A remark.}
The proof of Theorem~\ref{theorem:lower_approx_fin} in fact suggests that for finite $N$ and large horizon $H$, there exists a time-homogenous policy $\widebar{\pi}^\star \in \Pi$ different than the FH-MFG solution such that for $\widebar{\vecpi}^{\star}_H := \{ \widebar{\pi}^\star\}_{h=0}^{H-1}\in\Pi_H$, the time-averaged exploitability of $\widebar{\vecpi}^{\star}_H$ is small: $\forall i\in[N]:H^{-1}\ExpfinN (\widebar{\vecpi}^{\star}_H , \ldots, \widebar{\vecpi}^{\star}_H ) \leq \landauO(H^{-1} \log_2 N).$

\subsection{Approximation Analysis of Stat-MFG}\label{sec:main_relevance_stat}

Similarly, we introduce the $N$-player game corresponding to the Stat-MFG solution concept.

\begin{definition}[$N$-Stat-SAG]
    An $N$-player stationary SAG ($N$-Stat-SAG) problem is defined by the tuple $(N, \setS, \setA, P, R, \gamma)$ for Lipschitz dynamics and rewards $P \in \PL, R \in \RL$, discount factor $\gamma \in (0,1)$.
For any $(\mu, \vecpi) \in \Delta_\setS \times \Pi^N $, the $N$-player $\gamma$-discounted infinite horizon expected reward is defined as:
    \begin{align*}
        \JstatN (\mu, \vecpi) := \Exop \left[ \sum_{t=0}^\infty \gamma^t R(s_t^i, a_t^i, \widehat{\mu}_t) \middle| \substack{ a_t^j \sim \pi^j(s_t^j), \widehat{\mu}_t := \frac{\sum_j \vece_{ \text{\tiny $s^j_h$}}}{N} \\ s_0^j \sim \mu, s_{t+1}^i \sim P(s_t^i, a_t^i, \widehat{\mu}_t) }\right].
    \end{align*}
    A policy profile-population pair $(\mu^*, \vecpi^*) \in \Delta_\setS \times \Pi^N$ is called an $N$-Stat-SAG Nash equilibrium if:
    \begin{align*}
        &\JstatN (\mu^*, \vecpi^*) = \max_{\pi \in \Pi} \JstatN(\mu^*, (\pi, \vecpi^{*, -i})). 
        \tag{$N$-Stat-SAG-NE}
    \end{align*}
    If instead $\JstatN (\mu^*, \vecpi^*) \geq \max_{\pi \in \Pi} \JstatN(\mu^*, (\pi, \vecpi^{*, -i})) - \delta$, then we call $\mu^*, \pi^*$ a $\delta$-$N$-Stat-SAG Nash equilibrium.

\end{definition}

\begin{theorem}[Approximation of $N$-Stat-SAG]\label{theorem:upper_approx_stat}
Let $(\setS, \setA, H, P, R, \gamma)$ be a Stat-MFG and $(\mu^*, \pi^*) \in \Delta_\setS \times \Pi$ be a corresponding Stat-MFG-NE.
Furthermore, assume that $\Gamma_P(\cdot, \pi)$ is non-expansive in the $\ell_1$ norm for any $\pi$, that is, $\| \Gamma_P(\mu, \pi) - \Gamma_P(\mu', \pi) \| _ 1 \leq \| \mu - \mu'\|_1.$
Then, $(\mu^*, \vecpi^*) \in \Delta_\setS \times \Pi^N$ is a $\landauO\left(\frac{1}{\sqrt{N}}\right)$ Nash equilibrium for the $N$-player game where $\vecpi^*_N := (\pi^*, \ldots, \pi^*)$, that is, for all $i$,
\begin{align*}
    \JstatN (\mu^*, \vecpi^*_N) \geq \max_{\pi \in \Pi} \JstatN (\mu^*, (\pi, \vecpi^{*, -i}_N)) - \landauO\left(\frac{(1-\gamma)^{-3} }{\sqrt{N}}\right).
\end{align*}
\end{theorem}
\begin{proof}\emph{(sketch)}
    Let $(\mu^*, \pi^*)$ be a Stat-MFG-NE.
    The proof method is very similar to the FH-MFG case: we first bound the expected deviation from the stable distribution $\mu^*$ given by $\Exop[\| \widehat{\mu} - \mu^*\|_1]$.
    The truncated expected rewards can be controlled using similar arguments to the FH-MFG case, and an application of the dominated convergence theorem yields the exploitability for the infinite horizon discounted setting.
\end{proof}

We also establish an approximation lower bound for the $N$-Stat-SAG.
In this case, the question is if the non-expansive $\Gamma_P$ assumption is necessary for the optimal $\mathcal{O}(\sfrac{1}{\sqrt{N}})$ rate.
The below results affirm this: in for Stat-MFG-NE with expansive $\Gamma_P$, we suffer from an exploitability of $\omega(\sfrac{1}{\sqrt{N}})$ in the $N$-agent case.

\begin{theorem}[Lower bound for $N$-Stat-SAG]\label{theorem:lower_approx_stat}
    For any $N \in \mathbb{N}_{> 0}, \gamma \in (\sfrac{1}{\sqrt{2}}, 1)$ there exists $\setS, \setA$ with $|\setS|=6, |\setA|=2$ and $P \in \PL_{7}, R \in \RL_3$ such that:
    \begin{enumerate}
        \item The Stat-MFG $(\setS, \setA, P, R, \gamma)$ has a \emph{unique} NE $\mu^*, \pi^*$,
        \item For any $N$ and $\vecpi_N^* := (\pi^*, \ldots, \pi^*) \in \Pi^N$, it holds that $\JstatN (\vecpi_N^*) \leq \max_\pi \JstatN (\pi, \vecpi_N^{*,-i}) - \Omega(N^{-\log_2 \gamma^{-1}})$.
    \end{enumerate}
\end{theorem}
\begin{proof}\emph{(sketch)}
    The counter-example will be similar to the case in the FH-MFG, with minor modifications to make the Stat-MFG-NE unique.
    Intuitively, due to the same anti-concentration bound as before for $T = \log_2 \sqrt{N}$, at times $t = 0, 2, 4, \ldots, T-1$ the population deviation from $\mu^*$ can be lower bounded by $\Exop[ \| \empdist{t} - \mu^*\|_1] \geq \Omega(\frac{2^t}{\sqrt{N}})$.
    By the design of reward functions, this yields an exploitability of 
    \begin{align*}
        \Omega\left(\frac{1 + 2\gamma^2 + \ldots + (2\gamma^2)^{T-1}}{\sqrt{N}} \right) = \Omega\left(N^{-\log_2 \gamma^{-1}} \right).
    \end{align*}
    The proof is postponed to the supplementary material.
\end{proof}

The result above shows that unless the relevant $\Gamma_P$ operator is contracting in some potential, in general, the exploitability of the Stat-MFG-NE in the $N$-player game might be very large unless the effective horizon $(1-\gamma)^{-1}$ is small.
Hence, in these cases, the mean-field Nash equilibrium might be uninformative regarding the true NE of the $N$ player game. 
In the case of Stat-MFG, our lower bound is even stronger in the sense that the exploitability no longer decreases with $\mathcal{O}(\sfrac{1}{\sqrt{N}})$ for large $\gamma$.
For a sufficiently long effective horizon $(1-\gamma)^{-1}$ and large enough Lipschitz constant $L$, the rate in terms of $N$ can be arbitrarily slow.
Furthermore, if we take the ergodic limit $\gamma \rightarrow 1$, we will observe a non-vanishing exploitability $\Omega(1)$ for \emph{all} finite $N$.

\section{Computational Tractability of MFG}\label{sec:main_complexity}

The next fundamental question for mean-field reinforcement learning will be whether it is always computationally easier than finding an equilibrium of a $N$-player general sum normal form game.
We focus on the computational aspect of solving mean-field games in this section, and not statistical uncertainty: we assume we have full knowledge of the MFG dynamics.
We will show that unless additional assumptions are introduced (as typically done in the form of contractivity or monotonicity), solving MFG can in general be as hard as finding $N$-player general sum Nash.

We will prove that the problems are \ppadcomplete{}, where \ppad{} is a class of computational problems studied in the seminal work by \citet{papadimitriou1994complexity}, containing the complete problem of finding $N$-player Nash equilibrium in general sum normal form games and finding the fixed point of continuous maps \cite{daskalakis2009complexity,chen2009settling}.
The class \ppad{} is conjectured to contain difficult problems with no polynomial time algorithms \cite{beame1995relative, goldberg2011survey}, hence our results can be seen as a proof of difficulty.
Our results are significant since they imply that the MFG problems studied in literature are in the same complexity class as general-sum $N$-player normal form games or $N$-player Markov games \cite{daskalakis2023complexity}.
Once again, several computation-intensive aspects of our proofs will be postponed to the supplementary material.

Due to a technical detail, we will prove the complexity results for a subset of possible reward and transition probability functions.
We formalize this subset of possible rewards and dynamics as ``simple'' rewards/dynamics and also linear rewards, defined below.

\begin{definition}[Simple/Linear Dynamics and Rewards]
$R\in\RL$ and $P\in\PL$ are said to be \emph{simple} if for any $s,s'\in \setS, a \in \setA$, $P(s'|s,a, \mu)$ and $R(s,a,\mu)$ are functions of $\mu$ that are expressible as finite combinations of arithmetic operations $+, -, \times, \frac{\cdot}{\cdot}$ and functions $\max\{\cdot, \cdot\}, \min\{\cdot, \cdot\}$ of coordinates of $\mu$.
They are called \emph{linear} if $P(s'|s,a,\mu)$ and $R(s,a,\mu)$ are linear functions of $\mu$ for all $s,a,s'$.
The set of simple rewards and dynamics are denoted by $\RSim$ and $\PSim$ respectively, and the set of linear rewards and transitions are denoted $\RLin,\PLin$ respectively.
\end{definition}

\textbf{A note on simple functions.}
We define simple functions as above as in general there is no known efficient encoding of a Lipschitz continuous function as a sequence of bits.
This is significant since a Turing machine accepts a finite sequence of bits as input.
To solve this issue, we prove a slightly stronger hardness result that even games where $P(s'|s,a,\mu), R(s,a,\mu)$ are Lipschitz functions with strong structure are \ppadcomplete{}.
Since we are proving hardness, other larger classes of $P,R$ including $\PSim, \RSim$ will have similar intractability.
See also arithmetic circuits with $\max,\min$ gates \cite{daskalakis2011continuous} for a similar idea.

\subsection{The Complexity Class \ppad{}}

The \ppad{} class is defined by the complete problem \textsc{End-of-The-Line} \cite{daskalakis2009complexity}, whose formal definition we defer to the appendix as it is not used in our proofs.

\begin{definition}[\ppad{}, \ppadhard{}, \ppadcomplete{}]
    The class \ppad{} is defined as all search problems that can be reduced to \textsc{End-of-The-Line} in polynomial time.
    If \textsc{End-of-The-Line} can be reduced to a search problem $\setS$ in polynomial time, then $\setS$ is called \ppadhard{}.
    A search problem $\setS$ is called \ppadcomplete{} if it is both a member of PPAD and it is \ppadhard{}.
\end{definition}

While \textsc{End-of-the-Line} defines the problem class \ppad{}, it is hard to construct direct reductions to it.
We will instead use two problems that are known to be \ppadcomplete{} (and hence can be equivalently used to define \ppad{}): solving generalized circuits and finding a NE for an $N$-player general sum game.

\begin{definition}[Generalized Circuits \cite{rubinstein2015inapproximability, daskalakis2023complexity}]
    A generalized circuit $\setC = (\setV, \setG)$ is a finite set of nodes $\setV$ and gates $\setG$.
    Each gate $G \in \setG$ is characterized by the tuple $G(\theta|v_1, v_2|v)$ where $G \in \{ G_\leftarrow, G_{\times, +}, \setG_{<}\}$, $\theta \in \mathbb{R}^\star$ is a parameter (possibly of length 0), $v_1, v_2 \in V \cup \{ \perp \}$ are the input nodes (with $\perp$ indicating an empty input) and $v \in V$ it the output node of the gate.
    The collection of gates $\setG$ satisfies the property that if $G_1(\theta|v_1, v_2| v), G_2(\theta'|v'_1, v'_2| v') \in G$ are distinct gates, then $v\neq v'$. 
\end{definition}

Such circuits define a set of constraints on values assigned to each gate, and finding such an assignment will be the associated computational problem for such a circuit desription.
We formally define the \CompGC{} problem to this end.
\CompGC{} is a standard complete problem for the class \ppad{}, and we will work with it for our reductions.
We will use the shorthand notation $x = y \pm \varepsilon$ to indicate that $x \in [y - \varepsilon, y+\varepsilon]$ for $x,y\in\mathbb{R}$.

\begin{definition}[\CompGC{} \cite{rubinstein2015inapproximability}]
    Given a generalized circuit $\setC = (\setV, \setG)$, a function $p:V \rightarrow [0,1]$ is called an $\varepsilon$-satisfying assignment if:
    \begin{itemize}
        \item For every gate $G\in\setG$ of the form $G_{\leftarrow}(\zeta||v)$ for $\zeta\in{0,1}$, it holds that $p(v) = \zeta \pm \varepsilon$,
        \item For every gate $G\in\setG$ of the form $G_{\times,+}(\alpha, \beta | v_1, v_2 | v)$ for $\alpha, \beta \in [-1, 1]$, it holds that 
            \begin{align*}
                p(v) \in [\max\{\min\{0, \alpha p(v_1) + \beta p(v_2)\}\}] \pm \varepsilon,
            \end{align*}
        \item For every gate $G\in\setG$ of the form $G_{<}(|v_1, v_1|v)$ it holds that
            \begin{align*}
                p(v) =  \begin{cases}
                            &1 \pm \varepsilon, \quad p(v_1) \leq p(v_2) - \varepsilon, \\
                            &0 \pm \varepsilon, \quad p(v_1) \geq p(v_2) + \varepsilon.
                        \end{cases}
            \end{align*}
    \end{itemize}
    The $\varepsilon$-\textsc{GCircuit} problem is defined as follows:
    \begin{align*}
        \textit{Given generalized circuit } \setC, \textit{ find an } \varepsilon\textit{-satisfying assignment of $\setC$.}
    \end{align*}
\end{definition}

$\varepsilon$-\textsc{GCircuit} is one of the prototypical hard instances of \ppad{} problems as the result below suggests.

\begin{theorem}\cite{rubinstein2015inapproximability}
    There exists $\varepsilon > 0$ such that \CompGC{} is \ppadcomplete{}.
\end{theorem}

In other words, \CompGC{} is representative of the most difficult problem in \ppad{} which suggests intractability.
The $\varepsilon$-\textsc{GCircuit} computational problem will be used in our proofs by reducing an arbitrary generalized circuit into solving a particular MFG.

We will also use the general sum 2-player Nash computation problem, which is the standard problem of finding an approximate Nash equilibrium of a general sum bimatrix game.

\begin{definition}[\CompTwoNash{}]
    Given $\varepsilon > 0$, $K_1, K_2 \in \mathbb{N}_{>0}$, payoff matrices $A, B \in [0,1]^{K_1, K_2}$, find an approximate Nash equilibrium $(\sigma_1, \sigma_2) \in \Delta_{K_1} \times \Delta_{K_2}$ such that
    \begin{align*}
        \max_{\sigma \in \Delta_{K_1}} \sum_{i \in [K_1]} \sum_{j \in [K_2]} A_{i, j} \sigma(i) \sigma_2(j) -  \sum_{i \in [K_1]} \sum_{j \in [K_2]} A_{i, j} \sigma_1(i) \sigma_2(j) &\leq \varepsilon \\
        \max_{\sigma \in \Delta_{K_2}} \sum_{i \in [K_2]} \sum_{a \in [K_2]} B_{i,j} \sigma_1(i) \sigma(j) - \sum_{i \in [K_1]} \sum_{j \in [K_2]} B_{i, j} \sigma_1(i) \sigma_2(j) &\leq \varepsilon
    \end{align*}
\end{definition}

The following is the well-known result that even the $2$-Nash general sum problem is \ppadcomplete{}.
In fact, any $N$-player general sum normal form game is \ppadcomplete{}.

\begin{theorem}
    \cite{chen2009settling} 
    \CompTwoNash{} is \ppadcomplete{}.
\end{theorem}

\subsection{Complexity of Stat-MFG}\label{sec:complexity_main_stat}

Next, we provide our difficulty results for the Stat-MFG problem.
Notably, for Stat-MFG, the stability subproblem of finding a stable distribution for a fixed policy $\pi$ itself is \ppadhard{}.
Even without considering the optimality conditions, finding a stable distribution in general for a fixed policy is intractable, unless additional assumptions are introduced (e.g. $\Gamma_P$ is contractive or non-expansive).
We define the computational problem below and state the results.

\begin{definition}[\CompStat{}]
    Given finite state-action sets $\setS,\setA$, simple dynamics $P \in \PSim$ and policy $\pi$, find $\mu^* \in \Delta_\setS$ such that $\| \Gamma_P(\mu^*, \pi) - \mu^*\|_\infty \leq \frac{\varepsilon}{|\setS|}$.
\end{definition}

The computational problem as described above is to find an approximate fixed point of $\Gamma_P(\cdot, \pi)$ which corresponds to an approximately stable distribution of policy $\pi$.
We show that \CompStat{} is \ppadcomplete{} for some fixed constant $\varepsilon$.

\begin{theorem}[\CompStat{} is \ppadcomplete{}]\label{theorem:compstat_ppad}
For some $\varepsilon > 0$, the problem \CompStat{} is \ppadcomplete{}. 
\end{theorem}
\begin{proof}\emph{(sketch)}
\newcommand{\sbase}{s_\text{base}}
    The reduction from \CompStat{} to a fixed point problem (or the Sperner problem \cite{daskalakis2009complexity}) is straightforward, showing \CompStat{} is in \ppad{}.
    The main challenge of the proof is showing \CompStat{} is simultaneously \ppadhard{}.
    This is achieved by showing any \CompGC{} problem can be reduced to a \CompStat{} for some $\varepsilon'$.
    For simplicity, we reduce \CompGC{} to finding the stable distribution of a transition kernel $P(s'|s,\mu)$.
    Given a generalized circuit $\setC = (\setV, \setG)$, we construct a Stat-MFG that has one base state $\sbase$, one additional state $s_v$ for each $v\in\setV$ that is the output of a gate.
    Let $\theta := \frac{1}{8V}, B:= \frac{1}{4}$.
    Also define the function $u_\alpha(x) := \max\{ 0, \min \{ \alpha, x\}\}$ for any $\alpha \in [0,1]$.
    We present the construction and defer the analysis to the appendix: any gate of the form $G_{\leftarrow}(\zeta||v)$, we will add one state $s_v$ such that $P(\sbase|s_v, \mu) = 1$, $P(s_v | \sbase, \mu) = \frac{\zeta \theta}{\max\{B, \,\mu(\sbase)\}}$.
    For any weighted addition gate $G_{\times,+}(\alpha, \beta | v_1, v_2 | v)$, we add a state $s_v$ such that $P(\sbase|s_v, \mu) = 1$ and $P(s_v | \sbase, \mu) =  \frac{u_\theta(\alpha \mu(v_1) + \beta \mu(v_2))}{\max\{B, \,\mu(\sbase)\}}$.
    Finally, for each comparison gate $G_{<}(|v_1, v_1|v)$, also add a state $s_v$ and define the transition probabilities: 
    \begin{align*}
        P(s_v | \sbase, \mu) &= \frac{ \theta p_{\varepsilon/8} (\theta^{-1} \mu(s_1), \theta^{-1} \mu(s_2))}{\max\{B, \,\mu(\sbase)\}} , \\ 
        P(s_v | s_v, \mu) &= 0, \quad P(\sbase | s_v, \mu) = 1,
\end{align*}
    where $p_\varepsilon (x, y) := u_1 \left( \frac{1}{2} +  \varepsilon^{-1} (x - y) \right)$.
    Once all gates are added, the construction is completed by defining $P(\sbase|\sbase,\mu) = 1 - \sum_{s' \in \setS} P(s'|\sbase,\mu)$.
    Simple computation verifies that for any \emph{exact} stationary distribution $\mu^*$ of the above $P$, an exact assignment the the generalized circuit can be read by the map $v \rightarrow u_1(\frac{\mu^*(s_v)}{\theta})$.
\end{proof}

As a corollary, there is no polynomial time algorithm for \CompStat{} unless \ppad{}=\polytime{}, which is conjectured to be not the case.

\begin{corollary}
    There exists a $\varepsilon > 0$ such that there exists no polynomial time algorithm for \CompStat{}, unless \polytime{} = \ppad{}.
\end{corollary}

Most notably, these results show that the stable distribution oracle of \cite{cui2021approximately} might be intractable to compute in general, and the shared assumption that $\Gamma_P(\cdot, \pi)$ is contractive in some norm found in many works \cite{xie2021learning, anahtarci2022q, yardim2023policy} might not be trivial to remove without sacrificing tractability.

\subsection{Complexity of FH-MFG}\label{sec:complexity_main_fh}

We will show that finding an $\varepsilon$ solution to the finite horizon problem is also \ppadcomplete{}, in particular even if we restrict our attention to the case when $H=2$ and the transition probabilities $P$ do not depend on $\mu$.
We formalize the structured computational FH-MFG problem.

\begin{definition}[\CompFH{H}]
    Given simple reward function $R\in\setR^{\text{Sim}}$, transition matrix $P(s'|s,a)$, and initial distribution $\mu_0\in\Delta_\setS$, find a time dependent policy $\{\pi_h\}_{h=0}^{H-1}$ such that $\Expfin (\{\pi_h\}_{h=0}^{H-1}) \leq \sfrac{\varepsilon}{|\setS|}$.
\end{definition}

Our result in the case of the finite horizon MFG problem is that even in the case of $H=2$, the problem is \ppadcomplete{}.

\begin{theorem}[\CompFH{2} is \ppadcomplete{}]\label{theorem:compfh_ppad}
There exists an $\varepsilon > 0$ such that the problem \CompFH{2} is \ppadcomplete{}.
\end{theorem}
\begin{proof}\emph{(sketch)}
    Once again, showing \CompFH{2} is in \ppad{} is simple: it follows from the fact that a FH-MFG-NE is a fixed point of an easy-to-compute function (see e.g. \cite{huang2023statistical}).
    To show that \CompFH{2} is also \ppadhard{}, for an arbitrary generalized circuit $\setC = (\setV, \setG)$ we construct a FH-MFG whose $\delta$-NE will be $\delta'$-satisfying assignments for $\setC$ for some $\delta$'.
\end{proof}

\begin{corollary}
    There exists a $\varepsilon > 0$ such that there exists no polynomial time algorithm for \CompFH{2}, unless \polytime = \ppad.
\end{corollary}

These results for the FH-MFG show that the (weak) monotonicity assumption present in works such as \cite{perrin2020fictitious, perolat2022scaling} might also be necessary, as in the absence of any structural assumptions the problems are provably difficult.

Finally, we also show that even if $R(s,a,\mu)$ is a linear function of $\mu$ for all $s,a$ (that is, $R\in\RLin$), the intractability holds, although not for fixed $\varepsilon$.
We define the linear computational problem below.

\begin{definition}[\CompFHLinear{H}]\label{theorem:compfh_linear_ppad}
    Given $\varepsilon > 0$, linear reward function $R\in\RLin$, transition matrix $P(s'|s,a)$, find a time dependent policy $\{\pi_h\}_{h=0}^{H-1}$ such that $\Expfin (\{\pi_h\}_{h=0}^{H-1}) \leq \varepsilon.$
\end{definition}

\begin{theorem}[\CompFHLinear{2} is \ppadcomplete{}]
The problem \CompFHLinear{2}  is \ppadcomplete{}.
\end{theorem}
\begin{proof}\emph{(sketch)}
    \newcommand{\sbaseA}{s_{\text{base}}^1}
    \newcommand{\sbaseB}{s_{\text{base}}^2}
    In this case, we provide a reduction from \CompTwoNash{}.
    For a given \CompTwoNash{} instance $K_1, K_2 \in \mathbb{N}_{>0}$ with payoff matrices $A, B \in [0,1]^{K_1, K_2}$, we construct an FH-MFG with one initial state for each player and one additional state for each strategy of each of the players, resulting in a FH-MFG with $K_1 + K_2 + 2$ states, $\setS := \{\sbaseA, \sbaseB, s^1_1,\ldots, s^1_{K_1}, s^2_{1}, \ldots, s^2_{K_2}\}$.
    We set $\mu_0(\sbaseA) = \mu_0(\sbaseB) = \sfrac{1}{2}$.
    The action set will consist of $\max \{K_1, K_2 \}$ actions.
    In the first round, an agent starting from $\sbaseA$ will be transitioned to one of states $s^1_1,\ldots, s^1_{K_1}$ depending on the action picked receiving zero reward, and likewise and agent starting from $\sbaseB$ will transition to one of states $s^2_1,\ldots, s^2_{K_2}$.
    In the second round, the agent will receive a population-dependent reward regardless of the action player, which is equal to the expected utility of an action (a linear function).
    We postpone the cumbersome details relating to error analysis and dealing with the case $K_1 \neq K_2$ to the appendix.
\end{proof}

We emphasize that for \CompFHLinear{2} the accuracy $\varepsilon$ is also an input of the problem: hence the existence of a pseudo-polynomial time algorithm is not ruled out.

\section{Discussion and Conclusion}

We provided novel results on when mean-field RL is relevant for real-world applications and when it is tractable from a computational perspective.
Our results differ from existing work by provably characterizing cases where MFGs might have practical shortcomings.
From the approximation perspective, we show clear conditions and lower bounds on when the MFGs efficiently approximate real-world games.
Computationally, we show that even simple MFGs can be as hard as solving $N$-player general sum games.

We emphasize that our results do not discard MFGs, but rather identify potential bottlenecks (and conditions to overcome these) when using mean-field RL to compute a good approximate NE.

\begin{acks}
This project is supported by Swiss National Science Foundation (SNSF) under the framework of NCCR Automation and SNSF Starting Grant. A.Goldman was supported by the ETH Student Summer Research Fellowship.
\end{acks}

\bibliographystyle{ACM-Reference-Format} 
\bibliography{references}

\newpage
\appendix

\section{MFG Approximation Results}

\subsection{Preliminaries}\label{section:concentr_and_anticoncentr}

To establish explicit upper bounds on the approximation rate, we will use standard concentration tools.
\begin{definition}[Sub-Gaussian]
    \label{def:sg_rv}
    Random variable $\xi$ is called sub-Gaussian with variance proxy $\sigma^2$ if  $\forall\lambda\in\R: \quad
    \EE{
    e^{\lambda(\xi-\EE{\xi})}
    }
    \leq 
    e^{\frac{\lambda^2\sigma^2}{2}}$.
    In this case, we write $\xi\in\sgauss{\sigma^2}$.
\end{definition}

It is easy to show that if $\xi\in\sgauss{\sigma^2}$, then $\alpha \xi\in\sgauss{\alpha^2\sigma^2}$ for any constant $\alpha \in \mathbb{R}$.
Furthermore, if $\xi_1, \dots, \xi_n$ are independent random variables with $\xi_i\in\sgauss{\sigma_i^2}$, then $\sum_{i}\xi_i\in\sgauss{\sum_i \sigma_i^2}$.
Finally, if $\xi$ is almost surely bounded in $[a,b]$, then $\xi_i\in\sgauss{\sfrac{(b-a)^2}{4}}$.
We also state the well-known Hoeffding concentration bound and a corollary, Lemma~\ref{lemma:sgevbound}.

\begin{lemma}[Hoeffding inequality \cite{mcdiarmid1989method}]
\label{lemma:hoefd_ineq}
Let $\xi\in\sgauss{\sigma^2}$. 
Then for any $t>0$ it holds that
$\PP{
|\xi-\EE{\xi}|\geq t
}\leq 
2e^{-\frac{t^2}{2\sigma^2}}$.
\end{lemma}

\begin{lemma}
\label{lemma:sgevbound}
    Let $\xi\in\sgauss{\sigma^2}$. Then
\begin{align*}
    \EE{
    |\xi-\EE{\xi}|
    }\leq\sqrt{2\pi\sigma^2}, \quad
    \EE{
    (\xi-\EE{\xi})^2
    }\leq 4\sigma^2
\end{align*}
\end{lemma}
\begin{proof}
    \begin{align*}
        \EE{
        |\xi-\EE{\xi}|
        }
        &=
        \int_0^\infty 
        \P(|\xi-\EE{\xi}|\geq t) dt \\
        &
        \stackrel{(I)}{\leq}
        2\int_0^\infty 
        e^{-\frac{t^2}{2\sigma^2}} dt
        =
        \sqrt{2\pi\sigma^2}
    \end{align*}
    Inequality $(I)$ is true due to \Cref{lemma:hoefd_ineq}. 
    Likewise,
    \begin{align*}
        \EE{
        (\xi-\EE{\xi})^2
        }
        &=
        \int_0^\infty 
        \P((\xi-\EE{\xi})^2\geq t) dt\\
        &=
        \int_0^\infty 
        \P(|\xi-\EE{\xi}|\geq \sqrt{h}) dt\\
        &\stackrel{(II)}{\leq}
        2\int_0^\infty 
        e^{-\frac{h}{2\sigma^2}} dt
        =
        4\sigma^2
    \end{align*}
\end{proof}

Establishing lower bounds for the mean-field approximation of the $N$-player game will be more challenging as it will require different tools.
To establish lower bounds, we will need to use the following anti-concentration result for the binomial distribution.

\begin{lemma}[Anti-concentration for binomial]
\label{lemma:anticoncentration}
Let $N \in \mathbb{N}_{> 0}$ and $X \sim \operatorname{Binom}(N,p)$ be drawn from a binomial distribution for some $p\in [\sfrac{1}{2}, 1]$.
Then, $\Prob\left[X\geq\frac{N}{2} + \frac{ \sqrt{N}}{2}\right] \geq \frac{1}{20}.$
\end{lemma}
\begin{proof}
    For $k_0 := \left\lceil \frac{N}{2} + \frac{ \sqrt{N}}{2} \right\rceil$, we will lower bound $\sum_{k=k_0}^{N} \binom{N}{k} p^{k} (1-p)^{N-k}$ when $N$ is large enough.
    If $k_0 < \lceil Np \rceil$, then the probability in the statement above is bounded below trivially by $\sfrac{1}{2}$ since $\lfloor Np \rfloor$ lower bounds the median of the binomial \cite{kaas1980mean}.
    Otherwise, if $k_0 \geq \lceil Np \rceil$, then the function $\widebar{p} \rightarrow \widebar{p}^{k} (1-\widebar{p})^{N-k}$ is increasing in $\widebar{p}$ in the interval $[0, p]$. 
    As $\sfrac{1}{2} \in [0,p]$, it is then sufficient to assume $p = \sfrac{1}{2}$, and to upper bound $\Prob\left[ \frac{N}{2} - \frac{ \sqrt{N}}{2} < X < \frac{N}{2} + \frac{ \sqrt{N}}{2}\right]$ by $\sfrac{9}{10}$ as the binomial probability mass is symmetric around $\frac{N}{2}$ when $p=\sfrac{1}{2}$.

    First assuming $N$ is even, we obtain by monotonicity $\binom{N}{k} \leq\binom{N}{\sfrac{N}{2}} $.
    Using the Stirling bound $\sqrt{2 \pi} k^{k+\frac{1}{2}} \mathrm{e}^{-k} \leq k ! \leq \mathrm{e} k^{k+\frac{1}{2}} \mathrm{e}^{-k}$, we further upper bound $\binom{N}{\sfrac{N}{2}} \leq \frac{e}{\pi} \frac{2^N}{\sqrt{N}}$, resulting in the bound $\Prob\left[ \frac{N}{2} - \frac{ \sqrt{N}}{2} < X < \frac{N}{2} + \frac{ \sqrt{N}}{2}\right] \leq 2^{-N} \sqrt{N} \binom{N}{\sfrac{N}{2}} \leq \frac{e}{\pi} \leq \sfrac{9}{10} $, since there are at most $\sqrt{N}$ binomial coefficients being summed.
    Finally, assume $N = 2m+1$ is odd, then by the binomial formula $\binom{2m+1}{m+1} = \binom{2m}{m+1} + \binom{2m}{m} \leq 2 \binom{2m}{m} \leq \frac{2e}{\pi} \frac{2^{2m}}{\sqrt{2m}}$.
    Hence we have the bound on the sum $\Prob\left[ \frac{N}{2} - \frac{ \sqrt{N}}{2} < X < \frac{N}{2} + \frac{ \sqrt{N}}{2}\right] \leq \frac{e\sqrt{N}}{\pi} \frac{1}{\sqrt{N-1}}$.
    It is easy to verify that for $N \geq 16$, $\frac{e\sqrt{N}}{\pi\sqrt{N-1}} \leq \sfrac{9}{10}$, and the case when $N <16$ and $N$ is odd follows by manual computation.
\end{proof}

Finally, we prove slightly more general upper bounds than presented in the main text that approximates the exploitability of an \emph{approximate} MFG-NE in a finite population setting.
Hence we define the following notions approximate FH-MFG and Stat-MFG.

\begin{definition}[$\delta$-FH-MFG-NE]
Let $(\setS, \setA, H, P, R, \mu_0)$ be a FH-MFG.
    Then, a $\delta$-FH-MFG Nash equilibrium is defined as:
    \begin{align}
        \textit{Policy } &\vecpi_\delta^* = \{ \pi^*_{\delta,h}\}_{h=0}^{H-1} \in \Pi _ H \text{ such that } \notag \\
        & \Expfin(\{\pi^*_{\delta,h}\}_{h=0}^{H-1}) \leq \delta. \tag{$\delta$-FH-MFG-NE}
    \end{align}
\end{definition}

\begin{definition}[$\delta$-Stat-MFG-NE]
Let $(\setS, \setA, P, R, \gamma)$ be a Stat-MFG.
A policy-population pair $(\mu_\delta^*, \pi_\delta^*) \in \Delta_\setS \times \Pi$ is called a $\delta$-Stat-MFG Nash equilibrium if the two conditions hold:
    \begin{align*}
        \textit{Stability: } \quad &\mu_\delta^* = \Gamma_P(\mu_\delta^*, \pi_\delta^*), \notag \\
        \textit{Optimality: } \quad &\Vstat(\mu_\delta^*, \pi_\delta^*) \geq \max_{\pi \in \Pi} \Vstat(\mu_\delta^*, \pi) - \delta. 
        \tag{$\delta$-Stat-MFG-NE}
    \end{align*}
\end{definition}

\subsection{Upper Bound for FH-MFG: Extended Proof of Theorem~\ref{theorem:upper_approx_fin}}

Throughout this section we work with fixed $\dnm\in \PL_{\kmu}$ and $\rwd\in \RL_{\lmu}$. 
For any $\cX$ valued random variable $x$ denote $\Law(x)(\cdot)\in\dsp{\cX}$ as the distribution of $x$.
We start by introducing some notation.

For given $\rwd$ and $\dnm$ define the following constants:
    \begin{align*}
        \ls &\eqdef \sup_{s,s',a,\mu}
        \left|
        \rwd(s,a,\mu)-\rwd(s',a,\mu)
        \right|, \\
        \la &\eqdef \sup_{s,a,a',\mu}
        \left|
        \rwd(s,a,\mu)-\rwd(s,a',\mu)
        \right|, \\
        \ks &\eqdef \sup_{s,s',a,\mu}
        \left\|
        \dnm(\cdot|s,a,\mu)-\dnm(\cdot|s',a,\mu)
        \right\|, \\
        \ka &\eqdef \sup_{s,a,a',\mu}
        \left\|
        \dnm(\cdot|s,a,\mu)-\dnm(\cdot|s,a',\mu)
        \right\|.
    \end{align*} 
$\rwd$ and $\dnm$ are bounded due to \Cref{def:dynamics_rewards}, thus all constants $\ka,\ks,\la,\ls$ are finite and well-defined, and it always holds that $\ks,\ka\leq 2$ and $\ls,\la\leq 1$.
With the above definition of constants, the more general Lipschitz condition holds:
$\forall\ s,s'\in\ssp$, $a,a'\in\asp$, $\mu,\mu'\in\dsp{\ssp}$
    \begin{align*}
    \|\dnm(\cdot | s,a,\mu)
    -\dnm(\cdot|s',a',\mu')\|_1 \leq& 
    \kmu \|\mu-\mu'\|_1 +
    \ks d(s,s') \\
        &+
    \ka d(a,a'),\\
    |\rwd( s,a,\mu)
    -\rwd(s',a',\mu')| \leq& 
    \lmu \|\mu-\mu'\|_1 +
    \ls d(s,s') \\
    &+
    \la d(a,a').
\end{align*}
We also introduce the shorthand notation for any $s\in \setS, u\in \Delta_\setA, \mu \in \Delta_\setS$:
\begin{align*}
    \dnmpol(\cdot|s,u,\mu)&\eqdef \sum_{a\in\asp}u(a)\dnm(\cdot|s,a,\mu), \\
    \rwdpol(s,u,\mu)&\eqdef \sum_{a\in\asp}u(a)\rwd(s,a,\mu).
\end{align*}
By \cite[Lemma~C.1]{yardim2023policy}, it holds that
\begin{align}
\|\dnmpol(\cdot|s,u,\mu)
-
\dnmpol(\cdot|s',u',\mu')\|_1
\leq &
\kmu\|\mu-\mu'\|_1+\ks d(s,s') \notag \\
    &+\frac{\ka}{2}\|u-u'\|_1, \notag \\
|\rwdpol(s,u,\mu)
-
\rwdpol(s',u',\mu')|
\leq &\lmu \|\mu-\mu'\|_1
+
\ls d(s,s') \notag \\
&+\frac{\la}{2}\|u-u'\|_1 .\label{eq:lipschitz_pr_policy_avg}
\end{align}
We will define a new operator for tracking the evolution of the population distribution over finite time horizons for a time-varying policy $\forall \vecpi = \{ \pi_h \}_{h=0}^{H-1} \in \Pi_H$:
\begin{align*}
    \gpopind{\mu}{\vecpi}{h}&\eqdef 
    \underbrace{\gpop{\dots \gpop{\gpop{\mu}{\pi_{0}}}{\pi_{1}} \dots}{\pi_{h-1}}}_{h\text{ times}} \\
    &= \mu_{h}^{\vecpi} = \Lambdaop(\mu_0, \vecpi)_h,
\end{align*}
so $\gpopind{\mu}{\vecpi}{0} = \mu_0$.
By repeated applications of Lemma~\ref{lemma:lipschitz_gpop}, we obtain the Lipschitz condition:
\begin{align}
        \|
        &\gpopind{\mu}{\{\pi_{i}\}_{i=0}^{n-1}}{n}-\gpopind{\mu'}{\{\pi_{i}'\}_{i=0}^{n-1}}{n}
        \|_1 \notag\\
        &\leq
        \lpopmu 
        \|
        \gpopind{\mu}{\{\pi_{i}\}_{i=0}^{n-2}}{n-1}-\gpopind{\mu'}{\{\pi_{i}'\}_{i=0}^{n-2}}{n-1}
        \|_1\notag\\
        &+
        \frac{\ka}{2}\|\pi_{n-1}-\pi_{n-1}'\|_1 \notag\\
        &\leq 
        \lpopmu^{n} \|\mu-\mu'\|_1
        +
        \frac{\ka}{2}
        \sum_{i=0}^{n-1}\lpopmu^{n-1-i}\|\pi_{i}-\pi_{i}'\|_1,  \label{eq:gpop_composition_lip}
\end{align}
where $\lpopmu = (\kmu+\frac{\ks}{2}+\frac{\ka}{2})$.

The proof will proceed in three steps:
\begin{itemize}
    \item \textbf{Step 1.} Bounding the expected deviation of the empirical population distribution from the mean-field distribution $\EE{\|\empdist{h}-\mu_h^{\vecpi}\|_1}$ for any given policy $\vecpi$.
    \item \textbf{Step 2. } Bounding difference of $N$ agent value function $\JfinN$ and the infinite player value function $\Vfin$.
    \item \textbf{Step 3. } Bounding the exploitability of an agent when each of $N$ agents are playing the FH-MFG-NE policy.
\end{itemize}

\textbf{Step 1: Empirical distribution bound.}
Due to its relevance for a general connection between the FH-MFG and the $N$-player game, we state this result in the form of an explicit bound.

\begin{lemma}
    \label{lemma:fh_empdist}
    Suppose for the $N$-FH-MFG $(N, \setS, \setA, N, P, R, \gamma)$, agents $i=1, \ldots,N$ follow policies $\vecpi^i=\{\pi_h^i\}_{h}$.
    Let $\widebar{\vecpi}=\{\widebar{\pi}_h\}_h \in \Pi^H$ be arbitrary and $\vecmu^{\widebar{\vecpi}} := \{\mu_h^{\widebar{\vecpi}}\}_{h=0}^{H-1} = \Lambdaop(\mu_0, \widebar{\vecpi})$.
    Then for all $h\in \{0, \ldots, H-1\}$, it holds that:
    \begin{align*}
        \EE{\|\empdist{h}-\mu_h^{\widebar{\vecpi}}\|_1}\leq
            \frac{1-\lpopmu^{h+1}}{1-\lpopmu}
            \setsz{\ssp}\sqrt{\frac{\pi}{2N}}
            +
            \frac{\ka}{2N}
            \sum_{i=0}^{h-1} \lpopmu^{h-i-1}
            \poldif{\pi_{i}},
    \end{align*}
    where $\poldif{h} \eqdef \frac{1}{N}\sum_i \|\widebar{\pi}_h-\pi_h^i\|_1$
\end{lemma}
\begin{proof}
    The proof will proceed inductively over $h$.
    First, for time $h=0$, we have
    \begin{align*}
         \EE{
        \|\empdist{0}-\mu_0\|_1
        }
        =
        \sum_{s\in\ssp}
        \EE{
        \left|\frac{1}{N}\sum_{i=1}^N (\ind{s_0^i=s}-\mu_0(s))\right|}
        \leq
        \setsz{\ssp}\sqrt{\frac{\pi}{2N}},
    \end{align*}
    where the last line is due to \Cref{lemma:sgevbound} and the fact that $\ind{s_0^i=s}$ are bounded (hence subgaussian) random variables, and that in the finite state space we have $\EE{\ind{s_0^i=s}}=\mu_0(s)$.

    Next, denoting the $\sigma$-algebra induced by the random variables $(\{s^i_{h} \})_{i, h' \leq h}$ as $\mathcal{F}_h$, we have that:
    \begin{align}
         &\EEc{\|
            \empdist{h+1}-\mu_{h+1}^{\widebar{\vecpi}}
            \|_1}{\cF_h} \notag \\
             \leq & 
            \underbrace{
                \EEc{\|
                    \EEc{\empdist{h+1}}{\cF_h}-\Gamma_P(\empdist{h}, \widebar{\pi}_h)
                    \|_1}{\cF_h}
            }_{(\square)} \notag
            \\
            & + \underbrace{\EEc{\|
            \empdist{h+1}-\EEc{\empdist{h+1}}{\cF_h}
            \|_1}{\cF_h}}_{(\triangle)} + \underbrace{\EEc{\|
            \Gamma_P(\empdist{h}, \widebar{\pi}_h)-\mu_{h+1}^{\widebar{\vecpi}}
            \|_1}{\cF_h} }_{(\heartsuit)} \label{eq:main_inductive_bound_lemma_dist}
    \end{align}
    We upper bound the three terms separately.
    For $(\triangle)$, it holds that
    \begin{align*}
        (\triangle) = &\EEc{
        \|
        \empdist{h+1}-
        \EEc{\empdist{h+1}}{\cF_h}
        \|_1
        }{\cF_h} \\    
        =&
        \sum_{s\in\ssp}
        \EEc{
        |
        \empdist{h+1}(s)-
        \EEc{\empdist{h+1}(s)}{\cF_h}
        |
        }{\cF_h} \leq \setsz{\ssp}\sqrt{\frac{\pi}{2N}},
    \end{align*}
    since each $\empdist{h+1}(s)$ is an average of independent subgaussian random variables given $\cF_h$. 
    Specifically, each indicator is bounded $\ind{s^i_{h+1}= s}\in[0,1]$ a.s. and therefore is sub-Gaussian with $\ind{s^i_{h+1}= s}\in\sgauss{1/4}$. Thus we get $\empdist{h+1}(s)\in\sgauss{1/(4N)}$ and apply bound on expected value discussed in \Cref{section:concentr_and_anticoncentr}.
    
    Next, for $(\square)=\|\EEc{\empdist{h+1}}{\cF_h}-\Gamma_P(\empdist{h}, \widebar{\pi}_h)\|_1$, we note that
    \begin{align*}
        \EEc{\empdist{h+1}(s)}{\cF_h}
        =
        \EEc{ 
        \frac{1}{N}\sum_{i=1}^N
        \ind{s_{h+1}^i=s}
        }{\cF_h}
        =
        \frac{1}{N}\sum_{i=1}^N
        \dnmpol(s|s_h^i,\pi_h^i(s_h^i),\empdist{h}),
    \end{align*}
    therefore
    \begin{align*}
        (\square)
        &=
        \left\|\frac{1}{N}\sum_{i=1}^N 
        \dnmpol(\cdot|s_h^i,\pi_h^i(\cdot|s_h^i),\empdist{h})
        -
        \sum_{s'}\empdist{h}(s')\dnmpol(\cdot|s',\pi_h(\cdot|s'),\empdist{h})\right\|_1\\
        &=
        \left\|\frac{1}{N}\sum_{i=1}^N 
        \left(\dnmpol(\cdot|s_h^i,\pi_h^i(\cdot|s_h^i),\empdist{h})
        -
        \dnmpol(\cdot|s_h^i,\pi_h(\cdot|s_h^i),\empdist{h})\right)\right\|_1\\
        &\leq 
        \frac{1}{N}\sum_{i=1}^N 
        \|
        \dnmpol(\cdot|s_h^i,\pi_h^i(\cdot|s_h^i),\empdist{h})
        -
        \dnmpol(\cdot|s_h^i,\pi_h(\cdot|s_h^i),\empdist{h})\|_1\\
        &\stackrel{(I)}{\leq} 
        \frac{\ka}{2N}\sum_{i=1}^N \|\pi_h^i(\cdot|s_h^i) - \pi_h(\cdot|s_h^i)\|_1
        \leq 
        \frac{\ka}{2}\poldif{h},
    \end{align*}
    where (I) follows from the Lipschitz property \eqref{eq:lipschitz_pr_policy_avg}.
    Finally, the last term $(\heartsuit)$ can be bounded using:
    \begin{align*}
        (\heartsuit) = &\EEc{\|
            \Gamma_P(\empdist{h}, \widebar{\pi}_h)-
            \Gamma_P(\mu_{h}^{\widebar{\vecpi}}, \widebar{\pi}_h)
            \|_1}{\cF_h} 
            \leq \lpopmu \|\empdist{h} - \mu_{h}^{\widebar{\vecpi}} \|_1.
    \end{align*}
    To conclude, merging the bounds on the three terms in Inequality~\eqref{eq:main_inductive_bound_lemma_dist} and taking the expectations we obtain:
    \begin{align*}
        \EE{\| \empdist{h+1} - \mu_{h+1}^{\widebar{\vecpi}}\|_1} \leq \lpopmu \EE{\| \empdist{h} - \mu_{h}^{\widebar{\vecpi}}\|_1} + \setsz{\ssp}\sqrt{\frac{\pi}{2N}} + \frac{K_a \poldif{h}}{2}.
    \end{align*}
    Induction on $h$ yields the statement of the lemma.

\end{proof}

\textbf{Step 2: Bounding difference of $N$ agent value function.}
Next, we bound the difference between the $N$-player expected reward function $\JfinNi{1}$ and the infinite player expected reward function $\Vfin$.
For ease of reading, expectations, probabilities, and laws of random variables will be denoted $\Exop_{\infty}, \Prob_{\infty}, \Law_{\infty}$ respectively over the infinite player finite horizon game and $\Exop_{N}, \Prob_{N}, \Law_{N}$ respectively over the $N$-player game. 
We use the regular notation $\Exop[\cdot], \Prob[\cdot], \Law(\cdot)$ without subscripts if the underlying randomness is clearly defined.
We state the main result of this step in the following lemma.

\begin{lemma}
    \label{lemma:bound_n_mfg_similar_reward_fh}
    Suppose $N$-FH-MFG agents follow the same sequence of policies $\vecpi=\{\pi_h\}_{h=0}^{H-1}$. 
    Then
    \begin{multline*}
    \left|\JfinNi{1}(\vecpi,\dots,\vecpi) - \Vfin(\Lambdaop(\mu_0, \vecpi), \vecpi)
    \right|\\
    \leq
    (\lmu + \frac{\ls}{2})
    \setsz{\ssp}\sqrt{\frac{\pi}{2N}}
    \sum_{h=0}^{H-1}\frac{1-\lpopmu^{h+1}}{1-\lpopmu}.
    \end{multline*}
\end{lemma}
\begin{proof}
    Due to symmetry in the $N$ agent game, any permutation $\sigma: [N] \rightarrow [N]$ of agents does not change their distribution, that is $\Law_N(s_h^1,\dots, s_h^N) = \Law_N(s_h^{\sigma(1)},\dots, s_h^{\sigma(N)})$.
    We can then conclude that:
    \begin{align*}
        \EEN{
        \rwd(s_h^1, a_h^1, \empdist{h})
        }
        &=
        \frac{1}{N}\sum_{i=1}^N
        \EEN{
        \rwd(s_h^i, a_h^i, \empdist{h})
        } \\
        &=
        \EEN{
            \sum_{s
            \in\setS}\empdist{h}(s)\rwdpol(s, \pi_h(s), \empdist{h}).
        }
    \end{align*}
    Therefore, we by definition:
    \begin{align*}
        \JfinNi{1}(\vecpi,\dots,\vecpi) = \EEN{\sum_{h=0}^{H-1}  \sum_{s
            \in\setS}\empdist{h}(s)\rwdpol(s, \pi_h(s), \empdist{h})}.
    \end{align*}
    Next, in the FH-MFG, under the population distribution $\{\mu_h\}_{h=0}^{H-1} = \Lambdaop(\mu_0,\vecpi)$ we have that for all $h \in 0, \ldots, H-1$,
    \begin{align*}
        \Prob_\infty(s_0 = \cdot) &= \mu_0, \\ \Prob_\infty(s_{h+1} = \cdot) &= \sum_{s\in\setS}\Prob_\infty(s_{h} = s)\Prob_\infty(s_{h} = \cdot|s_{h} = s) \\
        &= \Gamma_P(\Prob_\infty(s_{h} = \cdot), \pi_h),
    \end{align*}
    so by induction $\Prob_\infty(s_h = \cdot) = \mu_h$.
    Then we can conclude that
    \begin{align*}
        \Vfin(\Lambdaop(\mu_0,\vecpi), \vecpi) &= \EEinf{\sum_{h=0}^{H-1} R(s_h, \pi_h(s_h), \mu_h)}\\ 
        &= \sum_{h=0}^{H-1} \sum_{s\in\setS} \mu_h(s) R(s, \pi_h(s), \mu_h).
    \end{align*}
    Merging the two equalities for $J, V$, we have the bound:
    \begin{align*}
        \quad&|\JfinNi{1}(\vecpi,\dots,\vecpi) - \Vfin(\Lambdaop(\mu_0,\vecpi), \vecpi)
        | \\
            = &\left| \EEN{\sum_{h=0}^{H-1}  \sum_{s
            \in\setS}\empdist{h}(s)\rwdpol(s, \pi_h(s), \empdist{h})} - \sum_{h=0}^{H-1} \sum_{s\in\setS} \mu_h(s) R(s, \pi_h(s), \mu_h)\right| \\
        \leq & \EEN{ \sum_{h=0}^{H-1} \left| \sum_{s
            \in\setS}\left(\empdist{h}(s)\rwdpol(s, \pi_h(s), \empdist{h}) - \mu_h(s) R(s, \pi_h(s), \mu_h) \right) \right|} \\
        \leq &\EEN{ \sum_{h=0}^{H-1} \left(\frac{L_s}{2}\|\mu_h -\empdist{h}\|_1 + L_\mu \|\mu_h -\empdist{h}\|_1\right)}.
    \end{align*}
    The statement of the lemma follows by an application of Lemma~\ref{lemma:fh_empdist}.
\end{proof}

\textbf{Step 3: Bounding difference in policy deviation.}
Finally, to conclude the proof of the main theorem of this section, we will prove that the improvement in expectation due to single-sided policy changes are at most of order $\cO\left(\frac{1}{\sqrt{N}}\right)$.

\begin{lemma}
    \label{lemma:bound_n_mfg_one_diff_reward}
    Suppose $\vecpi=\{\pi_h\}_{h=0}^{H-1} \in \Pi^H$ and
    $\vecpi' = \{\pi_h'\}_{h=0}^{H-1} \in \Pi^H$ arbitrary policies, and $\vecmu^{\vecpi} := \Lambdaop(\mu_0, \vecpi)$ is the population distribution induced by $\vecpi$.
    Then
    \begin{multline*}
    \left|\JfinNi{1}(\vecpi',\vecpi,\dots,\vecpi) - \Vfin(\Lambdaop(\mu_0, \vecpi), \vecpi')
    \right|\\
    \leq 
        \sum_{h=0}^{H-1}
        \biggl(
        \frac{\lmu }{2}
        \EE{
            \|\empdist{h}-\mu^{\vecpi}_h\|_1
        }
        +
        \kmu
        \sum_{h'=0}^{h-1}
        \EE{\|\empdist{h'}-\mu_{h'}^{\vecpi}\|_1}
        \biggr).
    \end{multline*}
\end{lemma}
\begin{proof}
    Define the random variables $\{s^i_h, a_h^i\}_{i,h}, \{\widehat{\mu}_h\}_h$ as in the definition of $N$-FH-SAG (Definition~\ref{def:n_fh_mfg}).
    In addition, define the random variables $\{s_h, a_h\}_h$ evolving according to the FH-MFG with population $\vecmu^{\vecpi} := \{ \mu^{\vecpi}_h \}_h := \Lambdaop(\mu_0, \vecpi)$ and representative policy $\vecpi'$, independent from the random variables $ \{s^i_h, a_h^i\}_{i,h}$. 
    Hence $s_0 \sim \mu_0, a_h \sim \pi'(\cdot|s_h), s_{h+1} \sim P(\cdot|s_h, a_h, \mu^{\vecpi}_h)$.
    Define also for simplicity
    \begin{align*}
        E_N := \left|\JfinNi{1}(\vecpi',\vecpi,\dots,\vecpi) - \Vfin(\Lambdaop(\mu_0 \vecpi), \vecpi')\right|.
    \end{align*}

    With these definitions, we have
    \begin{align}
        E_N &= \left| \Exop \left[ \sum_{h=0}^{H-1} R(s_h, a_h, \mu^{\vecpi}_h) - \sum_{h=0}^{H-1} R(s_h^1, a_h^1, \widehat{\mu}_h) \right] \right| \notag \\
            &\leq \sum_{h=0}^{H-1} \left| \Exop \left[  R(s_h, a_h, \mu^{\vecpi}_h) -  R(s_h^1, a_h^1, \widehat{\mu}_h) \right] \right|. \label{eq:sum_rh_fh_part3}
    \end{align}
    Furthermore, for any $h\in \{ 0, \ldots, H-1\}$,
    \begin{align*}
        | \Exop & \left[  R(s_h, a_h, \mu^{\vecpi}_h) -  R(s_h^1, a_h^1, \widehat{\mu}_h) \right] | \\
            \leq & \left| \Exop \left[  R(s_h, a_h, \mu^{\vecpi}_h) -  R(s_h^1, a_h^1, \mu^{\vecpi}_h) \right] \right|\\
            & + \left| \Exop \left[  R(s_h^1, a_h^1, \mu^{\vecpi}_h) -  R(s_h^1, a_h^1, \widehat{\mu}_h) \right] \right| \\
        \leq &\left| \Exop \left[  R(s_h, \pi'_h(s_h), \mu^{\vecpi}_h) -  R(s_h^1, \pi'_h(s_h^1), \mu^{\vecpi}_h) \right] \right|\\
            & + L_\mu  \Exop \left[   \| \mu^{\vecpi}_h - \widehat{\mu}_h \|_1 \right] \\
        \leq & \frac{1}{2} \| \Prob[s_h = \cdot] - \Prob[s_h^1 = \cdot] \|_1 + L_\mu  \Exop \left[   \| \mu^{\vecpi}_h - \widehat{\mu}_h \|_1 \right] ,
    \end{align*}
    where the last line follows since $R$ is bounded in $[0,1]$.
    Replacing this in Equation~\eqref{eq:sum_rh_fh_part3},
    \begin{align}
        E_N \leq \frac{1}{2} \sum_h \| \Prob[s_h = \cdot] - \Prob[s_h^1 = \cdot] \|_1 + L_\mu \sum_h \Exop \left[   \| \mu^{\vecpi}_h - \widehat{\mu}_h \|_1 \right]. \label{eq:fh_part3_EN_sum}
    \end{align}
   The first sum above we upper bound in the rest of the proof inductively.

    Firstly, by definitions of $N$-FH-SAG and FH-MFG, both $s_0^1$ and $s_0$ have distribution $\mu_0$, hence $\| \Prob[s_0 = \cdot] - \Prob[s_0^1 = \cdot] \|_1 = 0$.
    Assume that $h \geq 1$.
    We note that $P$ takes values in $\Delta_\setS$ and the random vector $\widehat{\mu}_h$ takes values in the discrete set $\{ \frac{1}{N} u : u \in \{0,\ldots, N \}^\setS, \sum_s u(s) = N \} \subset \Delta_\setS$, hence we have the bounds:
    \begin{align*}
        \| &\Prob[s_{h+1} = \cdot] - \Prob[s^1_{h+1} = \cdot] \|_1 \\
        \leq & \left\| \sum_{s,\mu} P(s, \pi'_h(s), \mu) \Prob[s_h^1=s, \widehat{\mu}_h=\mu] - \sum_{s} P(s, \pi_h'(s), \mu^{\vecpi}_h) \Prob[s_h=s] \right \|_1 \\
        \leq & \left\| \sum_{s} P(s, \pi_h'(s), \mu^{\vecpi}_h) \Prob[s_h^1=s] - \sum_{s} P(s, \pi_h'(s), \mu^{\vecpi}_h) \Prob[s_h=s] \right \|_1  \\
            & + \left\| \sum_{s,\mu} \left(P(s, \pi_h'(s), \mu) - P(s, \pi_h'(s), \mu^{\vecpi}_h) \right) \Prob[s_h^1=s, \widehat{\mu}_h=\mu] \right \|_1 \\
        \leq &  \left\| \Prob[s_h^1= \cdot ] - \Prob[s_h= \cdot] \right \|_1 +  \sum_{s,\mu} K_\mu \left\|\mu - \mu^{\vecpi}_h \right\|_1 \Prob[s_h^1=s, \widehat{\mu}_h=\mu] \\
        \leq &  \left\| \Prob[s_h^1= \cdot] - \Prob[s_h= \cdot]\right \|_1 +  K_\mu \Exop \left[ \left\|\widehat{\mu}_h^{\vecpi} - \mu^{\vecpi}_h \right\|_1 \right]
    \end{align*}
    where the last two lines follow from the fact that $P$ is $K_\mu$ Lipschitz in $\mu$ and stochastic matrices are non-expansive in the total-variation norm over probability distributions.
    By induction, we conclude that for all $h \geq 0$, it holds that:
    \begin{align*}
        \| &\Prob[s_{h} = \cdot] - \Prob[s^1_{h} = \cdot] \|_1 \leq K_\mu \sum_{h'=0}^{h} \Exop \left[ \left\|\widehat{\mu}_{h'}^{\vecpi} - \mu^{\vecpi}_{h'} \right\|_1 \right].
    \end{align*}

   Placing this result into Equation~\eqref{eq:fh_part3_EN_sum}, we obtain the statement of the lemma.

\end{proof}

Since $ \EE{\|\empdist{h'}-\mu_{h'}^{\vecpi}\|_1}$ above in the theorem is of the order of $\mathcal{O}\left(\sfrac{1}{\sqrt{N}}\right)$ by the result in step 1, the result above allows us to bound exploitability in the $N$-FH-SAG.

\textbf{Conclusion and Statement of Result.}
Finally, we can merge the results up until this stage to upper bound the exploitability.
By definition of the FH-MFG-NE, we have:
    \[
    \delta\geq 
    \max_{\vecpi' \in \Pi^H} \Vfin( \Lambdaop (\mu_0, \vecpi_\delta), \vecpi') - \Vfin ( \Lambdaop (\mu_0,\vecpi_\delta), \vecpi_\delta )
    \]
The upper bounds on the deviation between $\Vfin$ and $\JfinNi{1}$ from the previous steps directly yields the statement of the theorem.
We state it below for completeness.

\begin{theorem}
    \label{theorem:expoitability_convergence_simple}
    It holds that
    \begin{align*}
        \ExpfinNi{1}(\vecpi_\delta,\dots,\vecpi_\delta)
        \leq 
        2\delta 
        +
        \frac{C_1}{\sqrt{N}}
        +
        \frac{C_2}{N}
        =
        O\left(\delta + \frac{1}{\sqrt{N}}\right) 
    \end{align*}
    where $\vecpi_\delta$ is a $\delta$-FH-MFG Nash equilibrium and
    \begin{align*}
        C_1
        =
        \setsz{\ssp}\sqrt{\frac{\pi}{2}}
        \left(
        (2\lmu+\frac{\ls}{2})\sum_{h=0}^{H-1}\frac{1-\lpopmu^{h+1}}{1-\lpopmu}
        +
        \kmu\sum_{h=0}^{H-1}\sum_{i=0}^{h-1}
        \frac{1-\lpopmu^{i+1}}{1-\lpopmu}
        \right)
    \end{align*}
    \begin{align*}
        C_2
        =
        \lmu\ka\sum_{h=0}^{H-1}
        \frac{1-\lpopmu^{h}}{1-\lpopmu}
        +
        \ka\kmu\sum_{h=0}^{H-1}\sum_{i=0}^{h-1}
        \frac{1-\lpopmu^{i}}{1-\lpopmu},
    \end{align*}
    where we use shorthand notation $\frac{1-\lpopmu^k}{1-\lpopmu} := k-1$ when $\lpopmu = 1$.
\end{theorem}

\textbf{A note on constants.}
    Note that constants $C_1,C_2$ in \Cref{theorem:expoitability_convergence_simple} depend on horizon with $\frac{H^2}{1-\lpopmu}$ if $\lpopmu< 1$, with $H^3$ if $\lpopmu = 1$ and with $H^2\frac{1-\lpopmu^{H+1} }{1-\lpopmu}$ if $\lpopmu > 1$.

\subsection{Lower Bound for FH-MFG: Extended Proof of Theorem~\ref{theorem:lower_approx_fin}}

The proof will be by construction: we will explicitly define an FH-MFG where the optimal policy for the $N$-agent game diverges quickly from the FH-MFG-NE policy.

\textbf{Preliminaries.}
We first define a few utility functions.
Define $\fRew: \Delta_2 \rightarrow B_{\infty,+}^2 := \{ \vecx \in \mathbb{R}^2 : \| \vecx\|_\infty = 1,\ x_1,x_2\geq 0\}$ and $\fRep: \Delta_2 \rightarrow [0,1]^2$ as follows:
\begin{align*}
    \fRew(x_1, x_2) &:= \begin{pmatrix}
                \fRew_1(x_1, x_2)\\
                \fRew_2(x_1, x_2)
            \end{pmatrix} := \begin{pmatrix}
                \frac{x_1}{\max\{ x_1, x_2\}}\\
                \frac{x_2}{\max\{ x_1, x_2\}}
            \end{pmatrix}, \\
    \fRep(x_1, x_2) &:= \begin{pmatrix}
                \fRep_1(x_1, x_2)\\
                \fRep_2(x_1, x_2)
            \end{pmatrix} := \begin{pmatrix}
                \max\{ 4x_2, 1\} \\
                \max\{ 4x_1, 1\}
            \end{pmatrix}.
\end{align*}
Furthermore, for any $\epsilon > 0$ we define $\fProb: [0,1] \rightarrow [0,1]$ as:
\begin{align*}
    \fProb(x) &= \begin{cases}
        1, \quad x > \sfrac{1}{2} + \epsilon \\
        0, \quad x < \sfrac{1}{2} - \epsilon \\
        \frac{1}{2} + \frac{x - \sfrac{1}{2}}{2\epsilon}, \quad x \in [\sfrac{1}{2} - \epsilon, \sfrac{1}{2} + \epsilon]
    \end{cases}.
\end{align*}
$\epsilon \in (0, \sfrac{1}{2})$ will be specified later.

It is straightforward to verify that $\vecg$ has an inverse in its domain given by
\begin{align*}
    \vecg^{-1}(x_1, x_2) = \left( \frac{x_1}{x_1+x_2}, \frac{x_2}{x_1+x_2}\right), \forall (x_1, x_2) \in B^2_{\infty,+}.
\end{align*}
Furthermore, it holds for $\vecx = (x_1, x_2)\in B^2_{\infty,+}, \vecy=(y_1, y_2)\in B^2_{\infty,+}$
\begin{align*}
    \| \vecg^{-1}&(\vecx) - \vecg^{-1}(\vecy) \|_1 \\
    = &\left| \frac{x_1}{x_1+x_2} - \frac{y_1}{y_1+y_2}\right| + \left|\frac{x_2}{x_1+x_2} - \frac{y_2}{y_1+y_2} \right| \\
    = &\left| \frac{ x_1 (y_2 - x_2) + x_2 (x_1 - y_1)}{(x_1+x_2)(y_1+y_2)}\right| + \left|\frac{x_2 (y_1 - x_1) + x_1(x_2 - y_2)}{(x_1+x_2)(y_1+y_2)} \right| \\
    \leq &2 \| \vecx - \vecy\|_1,
\end{align*}
and likewise for $\vecu,\vecv \in \Delta_2$, letting $u_+ := \max\{u_1, u_2\}, 
v_+ := \max\{v_1, v_2\}, 
$
\begin{align*}
     \| \vecg(\vecu) - \vecg(\vecv) \|_1 = &\left| \frac{u_1}{u_+} - \frac{v_1}{v_+}\right| + \left|\frac{u_2}{u_+} - \frac{v_2}{v _+} \right| \\
        = & \left| \frac{u_1 v_+ - v_1 u_+}{u_+ v_+} \right| + \left|\frac{u_2 v_+ - u_+ v_2}{u_+ v_+}\right| \leq 2 \| \vecu - \vecv\|_1.
\end{align*}
This follows from considering cases and observation that $u_+\geq \sfrac{1}{2}$, $v_+\geq \sfrac{1}{2}$.
Then for all $\vecu,\vecv \in \Delta_2$, $\vecg, \vech$ have the bi-Lipschitz and Lipschitz properties:
\begin{align}
    \label{eq:bilipschitzg}
    \frac{1}{2} \| \vecu - \vecv \|_1 \leq \|\fRew(\vecu) - \fRew(\vecv) \|_1 \leq 2 \| \vecu - \vecv \|_1, \\
    \label{eq:bilipschitzh}
    \|\fRep(\vecu) - \fRep(\vecv) \|_1 \leq 4 \| \vecu - \vecv \|_1.
\end{align}
Likewise, $\fProb$, being piecewise linear, also satisfies the Lipschitz condition: $|\fProb(x) - \fProb(y)| \leq \frac{1}{2\epsilon} |x - y|, \quad \forall x,y \in [0,1].$

\textbf{Defining the FH-MFG.}
We take a particular FH-MFG with 6 states, 2 actions.
Define the state-actions sets:
\begin{align*}
    \setS &= \{ \sleft, \sright, \slA, \slB,\srA, \srB \}, \quad
    \setA = \{\actA, \actB \}.
\end{align*}
Intuitively, the ``main'' states of the game are $\sleft, \sright$ and the $4$ states $\slA, \slB,\srA, \srB $ are dummy states that keep track of which actions were taken by which percentage of players used to introduce a dependency of the rewards on the distribution of agents over actions as well as states.
Define the initial probabilities $\mu_0$ by:
\begin{align*}
    \vecmu_0(\sleft) &= \vecmu_0(\sright) = \sfrac{1}{2}, \\
    \vecmu_0(\slA) &= \vecmu_0(\srA) = \vecmu_0(\srA) = \vecmu_0(\srB) = 0.
\end{align*}
When at the states $\sleft, \sright$, the transition probabilities are defined for all $\mu\in\Delta_\setS$ by:
\begin{align*}
    P(\slA | \sleft, \actA, \mu) &= 1, \quad P(\slB | \sleft, \actB, \mu) = 1, \\
    P(\srA | \sright, \actA, \mu) &= 1, \quad P(\srB | \sright, \actB, \mu) = 1.
\end{align*}
That is, the agent transitions to one of $\{\slA, \srA, \srB, \slB\}$ to remember its last action and left-right state.
When at states $\{\slA, \srA, \srB, \slB\}$, the transition probabilities are:
\begin{align*}
    \text{If } s &\in \{\slA, \slB, \srA, \srB\}: \\ P(s' | s, a, \mu) &= \begin{cases}
        \fProb(\mu(\slA) + \mu(\slB)), \text{ if } s' = \sleft \\
        \fProb(\mu(\srA) + \mu(\srB)), \text{ if } s' = \sright
    \end{cases}, \forall \mu, a .
\end{align*}
The other non-defined transition probabilities are of course $0$.

Finally, let $\alpha, \beta > 0$ such that $\alpha + \beta < 1$ (to be also defined later).
The reward functions are defined for all $\mu\in\Delta_\setS$ as follows:
\begin{align*}
    R(\sleft, \actA, \mu) = &R(\sleft, \actB, \mu) = 0, \\
    R(\sright, \actA, \mu) = &R(\sright, \actB, \mu) = 0, \\
    \begin{pmatrix}
        R(\slA, \actA, \mu) \\
        R(\slB, \actA, \mu)
    \end{pmatrix} = &(1 - \alpha - \beta) \fRew\big(\mu(\slA) + \mu(\slB), \mu(\srA) + \mu(\srB) \big) \\
                &+ \alpha \fRep(\mu(\slA), \mu(\slB)) \\
    \begin{pmatrix}
        R(\slA, \actB, \mu) \\
        R(\slB, \actB, \mu)
    \end{pmatrix} = &(1 - \alpha - \beta)\fRew\big(\mu(\slA) + \mu(\slB), \mu(\srA) + \mu(\srB) \big) \\
        &+ \alpha \fRep(\mu(\slA), \mu(\slB)) 
        + \beta \vecone \\
    \begin{pmatrix}
        R(\srA, \actA, \mu) \\
        R(\srB, \actA, \mu)
    \end{pmatrix} = &(1 - \alpha - \beta)\fRew\big(\mu(\srA) + \mu(\srB), \mu(\slA) + \mu(\slB) \big) \\
        &+ \alpha\fRep(\mu(\srA), \mu(\srB)) \\
    \begin{pmatrix}
        R(\srA, \actB, \mu) \\
        R(\srB, \actB, \mu)
    \end{pmatrix} = &(1 - \alpha - \beta)\fRew\big(\mu(\srA) + \mu(\srB), \mu(\slA) + \mu(\slB) \big) \\
        &+ \alpha\fRep(\mu(\srA), \mu(\srB)) 
        + \beta \vecone
\end{align*}
Note that only at odd steps do the agents get a reward, and at this step, it does not matter which action the agent plays, only the state among $\{\slA, \slA, \srA, \srB \}$ and the population distribution.
The parameters $\epsilon, \alpha, \beta$ of the above FH-MFG are ``free'' parameters to be specified later.

\textbf{A minor remark.}
The arguments of $\fRew$ above will be with probability one in the set $\Delta_2$ at odd-numbered time steps, but to formally satisfy the Lipschitz condition $R\in\RL_2$ one can for instance replace $\fRew\big(\mu(\srA) + \mu(\srB), \mu(\slA) + \mu(\slB) \big)$ with $\fRew\big(\mu(\srA) + \mu(\srB) + \mu(\sleft), \mu(\slA) + \mu(\slB) + \mu(\sright)\big)$ in the definitions, which will not impact the analysis since at odd timesteps $\mu(\sright) = \mu(\sleft) = 0$ for both the FH-MFG and $N$-FH-SAG.

Note that with these definitions, $P \in \PL_{\sfrac{1}{2\epsilon}}, R \in \RL_{2}$ since only $\forall\ s,s'\in\setS,a,a'\in\setA,\mu,\mu'\in\Delta_\setS$, we have by the definitions:
\begin{align}
    \label{eq:lower_bound_P_Lipschitz}
    \|P(\cdot|s,a,\mu)-P(\cdot|s',a',\mu')\|_1 &\leq 2d(s,s')+2d(a,a')+ \frac{1}{2\epsilon}\|\mu-\mu'\|_1, \\
    |R(s,a,\mu)-R(s',a',\mu')| &\leq d(s,s') + d(a,a')
        +2\|\mu-\mu'\|_1, \label{eq:lower_bound_R_Lipschitz}
\end{align}
for any $\alpha,\beta > 0 $ with $\alpha + \beta < 1$ and $\alpha < \frac{1}{4}$,
using the Lipschitz conditions in \eqref{eq:bilipschitzg}, \eqref{eq:bilipschitzh}.

\textbf{Step 1: Solution of the FH-MFG.}
Next, we solve the infinite player FH-MFG and show that the policy $\vecpi^*_H := \{\pi^*_h\}_{h=0}^{H-1}$ given by:
\begin{align*}
    \pi_h^*(a|s) := \begin{cases}
        1, \text{if $h$ odd and $a = \actB$} \\
        \frac{1}{2}, \text{if $h$ even} \\
        0, \text{if $h$ odd and $a = \actB$}
    \end{cases}
\end{align*}
It is easy to verify in this case that, if $\vecmu^* := \{\mu^*_h\}_h$ is induced by $\vecpi^*$:
\begin{align*}
    \mu^*_h(\slA) &= \mu^*_h(\slB) = \mu^*_h(\srA) = \mu^*_h(\srB) = \sfrac{1}{4}, \text{ if $h$ odd}, \\
    \mu^*_h(\sleft) &= \mu^*_h(\sright) = \sfrac{1}{2}, \text{ if $h$ even}. 
\end{align*}
In this case, the induced rewards in odd steps are state-independent (it is the same for all states $\srA,\srB,\slA,\slB$), therefore the policy $\vecpi^*$ is the optimal best response to the population and a FH-MFG.

In fact, $\vecpi^*$ is unique up to modifications in zero-probability sets (e.g., modifying $\pi^*_h(\sleft)$ for odd $h$, for which $\Prob[s_h = \sleft] = 0$).
To see this, for \emph{any} policy $\vecpi \in \Pi_H$, it holds that 
\begin{align*}
    \mu_h^{\vecpi}(\sleft) &= \mu_h^{\vecpi}(\sright) = \sfrac{1}{2}, \text{ if $h$ even}, \\
    \mu_h^{\vecpi}(\slA) + \mu_h^{\vecpi}(\slB) &= \mu_h^{\vecpi}(\srA) + \mu_h^{\vecpi}(\srB) = \sfrac{1}{2}, \text{ if $h$ odd}, 
\end{align*}
as the action of the agent does not affect transition probabilities between $\sleft,\sright$ in even rounds.
Moreover, as odd stages, the action rewards terms only depend on the state apart from the positive additional term $\beta \vecone$, so the only optimal action will be $\actB$.
Finally, for $\alpha > 0$, the actions $\actA,\actB$ must be played with equal probability as otherwise the term $\alpha \fRep(\mu(\srA), \mu(\srB))$ will lead to the action with lower probability assigned by being optimal.

\textbf{Step 2: Population divergence in $N$-FH-MFG.}
We will analyze the empirical population distribution deviation from $\vecmu^*$, namely, we will lower bound $\Exop[\|\mu_h^* - \empdist{h}\|_1]$.
The results in this step will be valid for \emph{any} policy profile $(\vecpi^1, \ldots,\vecpi^N) \in \Pi$: we emphasize that at even $h$, $\empdist{h}$ is independent of agent policies in the $N$ player game.
In this step, we also fix $\sfrac{1}{2\varepsilon} = 8$.

We will analyze $\empdist{h}$ at all even steps $h=2m$ where $m \in \mathbb{N}_{\geq 0}$.
Define the sequence of random variables for all $m \in \mathbb{N}_{\geq 0}$ as $X_m := \empdist{2m}(\sleft)$.
Define $\setG := \{ \frac{k}{N} : k = 0,\ldots, N \}$.
Note that for all even $h = 2m$, it holds almost surely that $\widehat{\mu}_h(\sleft), \widehat{\mu}_h(\sright) \in \setG$.
By the definition of the MFG, it holds for any $m \geq 0, k \in [N]$ that
\begin{align*}
    &\Prob[ N X_0 = k ] = \binom{N}{k} 2^{-N}, \\
    &\Prob[ N X_{m+1} = k | X_m ] = \binom{N}{k} (\omega_\eps(X_m))^{k} (1-\omega_\eps(X_m))^k,
\end{align*}
that is, given $X_m$, $N X_{m+1}$ is binomially distributed with $N X_{m+1} \sim \operatorname{Binom}(N, \fProb(X_m))$ without any dependence on the actions played by agents.
Therefore
\begin{align*}
    \EE{X_{m+1} | X_{m}} = \fProb(X_m), \quad \Var {X_{m+1} | X_{m}} 
    \leq \frac{1}{4N}.
\end{align*}
We define the following set $\setG_* := \{ 0, 1 \} \subset \setG$.
By the definition of the mechanics, if $x \in \setG_*, m\in\mathbb{N}_{\geq 0}$, it holds for all $m' > m$ that $\Prob[ X_{m'} = X_m | X_m = x] = 1$, that is once the Markovian random process $X_m$ hits $\setG_*$, it will remain in $\setG_*$.
Furthermore, for $K := \lfloor \log_5 \sqrt{N} \rfloor$, and for $k = 0,\ldots, K$ define the level sets:
\begin{align*}
    \setG_{-1} &:= \setG, \quad \setG_k := \left\{ x \in \setG :\left|x - \frac{1}{2}\right| \geq \frac{5^k}{2\sqrt{N}}  \right\}.
\end{align*}
For all $k \geq K$, define $\setG_k := \setG_*.$

Firstly, we have that
\begin{align*}
\Prob[X_0 \in \setG_0] = &\Prob \left[ \left| \frac{1}{N} \sum_i \ind{s_0^i = \sleft} - \frac{1}{2}\right| \geq \frac{1}{2\sqrt{N}} \right]
    \\
    = &\Prob \left[ \left| \sum_i \ind{s_0^i = \sleft} - \frac{N}{2}\right| \geq \frac{\sqrt{N}}{2} \right] \geq \frac{1}{10},
\end{align*}
where in the last line we applied the anti-concentration result of Lemma~\ref{lemma:anticoncentration} on the sum of independent Bernoulli random variables $\ind{s_0^i = \sleft}$ for $i\in[N]$. 

Next, assume that for some $m \in 1, \ldots, K-1$ we have $p \in \setG_m$.
If $\fProb (p) \in \{0,1\}$, it holds trivially that $\Prob[ X_{m+1} \in \setG_{m+1}| X_m = p] = 1$.
Otherwise, if $\fProb (p) \in (0,1)$,
\begin{align*}
    \Prob&[ X_{m+1} \in \setG_{m+1}| X_m = p]\\
    = &\Prob \left[ |X_{m+1} -\frac{1}{2} | \geq \frac{5^{m+1}}{2\sqrt{N}} \, \middle| \, X_m  = p\right] \\
        \geq &\Prob \left[  |\fProb(p) -\frac{1}{2}| -  |X_{m+1} - \fProb(p)| \geq \frac{5^{m+1}}{2\sqrt{N}} \middle|  X_m = p \right].
\end{align*}
Since in this case $|\fProb(X_m) -\frac{1}{2}| = |\fProb(X_m) - \fProb(\frac{1}{2})| \geq \sfrac{1}{2\epsilon} |X_m - \fProb(\frac{1}{2})|$, we have
\begin{align*}
    \Prob&[ X_{m+1} \in \setG_{m+1}| X_m = p] \\
    \geq &\Prob \left[  |\fProb(p) -\frac{1}{2}| -  |X_{m+1} - \fProb(p)| \geq \frac{5^{m+1}}{2\sqrt{N}} \middle|  X_m = p \right] \\
    = &\Prob \left[ \Big| X_{m+1} - \fProb(p) \Big| \leq \left|\fProb(p) -\frac{1}{2}\right| - \frac{5^{m+1}}{2\sqrt{N}} \middle|  X_m = p \right]  \\
    \geq &\Prob \left[ \Big| X_{m+1} - \fProb(p) \Big| \leq  8\frac{5^{m}}{2\sqrt{N}} - \frac{5^{m+1}}{2\sqrt{N}} \middle|  X_m = p \right]  \\
    = &\Prob \left[ \Big| X_{m+1} - \fProb(p) \Big| \leq  3\frac{5^{m}}{2\sqrt{N}} \middle|  X_m = p \right]  \\
    \geq &1 - 2 \exp \left\{ - \frac{9}{50}25^{m+1}\right\} 
\end{align*}
where in the last line we invoked the Hoeffding concentration bound (Lemma~\ref{lemma:hoefd_ineq}).

Using the above result inductively for $m \in 0, \ldots, K$ it holds that
\begin{align*}
    \Prob[X_m \in \setG_m | X_0 \in \setG_0] \geq &\prod_{m'=1}^{m} \Prob[X_{m'} \in \setG_{m'} |X_{m'-1} \in \setG_{m'-1}  ] \\
        \geq& \prod_{m'=1}^{m} \left( 1 - 2 \exp \left\{ - \frac{9}{50} 25^{m'}\right\}  \right) \\
        \geq & \left( 1 - 2 \sum_{m'=0}^{\infty}  \exp \left\{ - \frac{9}{50} 25^{m'+1}\right\}  \right) \\
        \geq & \left( 1 - 2 \sum_{m'=0}^{\infty} \exp \left\{ - \frac{9}{2}m' - \frac{9}{2}\right\}  \right) \\
        \geq & \left(1 - \frac{2 e^{-\sfrac{9}{2}}}{1 - e^{-\sfrac{9}{2}}} \right) \geq \frac{9}{10}.
\end{align*}
Since for $k > K$, $\Prob[X_{k+1} \in \setG_* | X_k \in \setG_*] = 1$ and $ \Prob[X_0 \in \setG_0] \geq \sfrac{1}{10}$, it also holds that
\begin{align*}
    \Prob[X_m \in \setG_m, \forall m \geq 0] \geq \frac{9}{100}.
\end{align*}

Finally, we use the above lower bound on the probability to lower bound the expectation:
\begin{align*}
    \EE{ \| \empdist{2m} - \mu_{2m}\|_1} &\geq \Prob[X_m \in \setG_m] \, \EEc{\left\| \empdist{2m} - \mu_{2m} \right\|_1}{X_m \in \setG_m} \\
    &\geq \Prob[X_m \in \setG_m] \, \EEc{2|X_m - \sfrac{1}{2}|\,\,}{X_m \in \setG_m} \\
    &\geq \frac{9}{100} \, \min \left\{ \frac{5^m}{\sqrt{N}}, 1\right\}.
\end{align*}
For odd $h=2m+1$, we also have the inequality
\begin{align*}
    \EE{ \| \empdist{2m+1} - \mu_{2m+1}\|_1} &\geq \EE{ \| \empdist{2m} - \mu_{2m}\|_1} \\
    &\geq \frac{9}{100} \, \min \left\{ \frac{5^m}{\sqrt{N}}, 1\right\}.
\end{align*}
which completes the first statement of the theorem (as $5^{H/2} = \Omega(2^H)$).

\newcommand{\pierg}[1]{\pi^\text{br}_{#1}}

\textbf{Step 3: Hitting time for $\setG_*$.}
We will show that the empirical distribution of agent states almost always concentrates on one of $\sleft, \sright$ during the even rounds in the $N$-player game, and bound the expected waiting time for this to happen.
The distributions of agents over states $\sleft, \sright$ in the even rounds are policy independent (they are not affected by which actions are played): hence the results from Step 2 still hold for the population distribution and the expected time computed in this step will be valid for any policy.

For simplicity, we define the FH-MFG for the non-terminating infinite horizon chain, and we will compute value functions up to horizon $H$.
Define the (random) hitting time $\tau$ as follows:
\begin{align*}
    \tau := \inf \{ m \geq 0: \empdist{2m}(\sleft) \in \setG_* \} = \inf \{ m \geq 0: X_m \in \setG_* \}.
\end{align*}
Note that for any $p \in \setG$, it holds that $\Prob[ X_{m+1} \in \setG_*  | X_m = p ] = \empdist{2m}(\sleft)^N + \empdist{2m}(\sright)^N=p^N+(1-p)^N \geq 2^{-N}.$
Therefore for all $m$ it holds that $\Prob[\empdist{2m} \notin \setG_*] \leq \left(1 - 2^{-N}\right)^{m-1}$.
By the Borel-Cantelli lemma, we can conclude that $\tau < \infty$ almost surely, and in particular $T_\tau := \Exop[\tau | X_0 = x] < \infty$ for any $x \in \setG$.

Next, we compute the expected value $T_{\tau}$.
Define the following two quantities:
\begin{align*}
    T_{-1} &:= \sup_{x \in \setG_{-1} }\{ \Exop[\tau |X_0 = x] \} \\ %
    T_0 &:= \sup_{x \in \setG_{0} }\{ \Exop[\tau |X_0 = x] \}.
\end{align*}

First, we compute an upper bound for $T_0$.
Define the event:
\begin{align*}
    E_0 := \bigcap_{m' \in [K]} \{ X_{m'} \in \setG_{m'} \}.
\end{align*}
Then, $T_0$ is upper bounded by:
\begin{align*}
    T_0 = &\sup_{x \in \setG_{0} } \Exop[\tau |X_0 = x] \\
        = &\sup_{x \in \setG_{0} } \Exop[\tau |E_0, X_0 = x] \Prob[E_0| X_0 = x ] \\    &+ \Exop[\tau |E_0^c, X_0 = x] \Prob[E_0^c| X_0 = x ] \\
        \leq &\sup_{x \in \setG_{0} } \Exop[\tau |E_0, X_0 = x] \Prob[E_0| X_0 = x ] \\
            &+ \Exop[\tau |E_0^c, X_0 = x] \Prob[E_0^c| X_0 = x ] \\
        \leq &K \frac{9}{10} +  \left( K + T_{-1} \right) \frac{1}{10} = K + \frac{T_{-1}}{10}
\end{align*}
where in the last step we used the lower bound on $\Prob[E_0]$ from Step 2.
Similarly for $T_{-1}$, from the one-sided anti-concentration bound (Lemma~\ref{lemma:anticoncentration}) it holds that:
\begin{align*}
    T_{-1} \leq &\sup_{x \in \setG_{-1} } \Exop[\tau |X_0 = x] \\
         \leq &\Exop[\tau |x \in \setG_0, X_0 = x] \Prob[x \in \setG_0| X_0 = x ] \\    
         &+ \Exop[\tau |x \notin \setG_0, X_0 = x] \Prob[x \notin \setG_0| X_0 = x ] \\
    \leq & \frac{1}{20} (T_0 + 1) + \frac{19}{20} (T_{-1}+1),
\end{align*}
the last line following since $T_{-1} > T_{0}$ by definition.
Solving the two inequalities, we obtain
\begin{align*}
    T_\tau \leq T_{-1} \leq \frac{200}{9} + \frac{10K}{9} \leq 23 + \frac{5}{9}\log_5 {N}.
\end{align*}

\textbf{Step 4: Ergodic optimal response to $N$-players.}
Next, we formulate a policy $\vecpi^\text{br} = \{ \pierg{h}\}_{h=0}^{H-1} \in \Pi^H$ that is ergodically optimal for the $N$-player game and can exploit a population that deploys the unique FH-MFG-NE.
For all $h$, the optimal policy will be defined by:
\begin{align*}
    \pierg{h} (a|s) = \begin{cases}
        1, \text{ if $s=\sleft$, $a = \actA$}\\
        1, \text{ if $s=\sright$, $a = \actB$}\\
        1, \text{ if $s\notin \{ \sleft, \sright\}$, $a = \actB$}\\
        0, \text{ otherwise}
    \end{cases}
\end{align*}
Intuitively, $\pierg{h}$ becomes optimal once all the agents are concentrated in the same states during the even rounds, which happens very quickly as shown in Step 3.
Assume that agents $i=2, \ldots N$ deploy the unique FH-MFG-NE $\vecpi^i = \vecpi^*$, and for agent $i=1$,  $\vecpi^1 = \vecpi^\text{br}$.
We decompose the three components of the rewards for the first agent, as defined in the construction of the MFG (Step 1):
\begin{align*}
    \JfinNi{1}&(\vecpi^\text{br}, \vecpi^*, \ldots, \vecpi^*) \\
        &= \EE{\sum_{\substack{h \text{ odd} \\ 0 \leq h \leq H}} (1 - \alpha - \beta) R^{1,\fRew}_h + \alpha R^{1,\fRep}_h + \beta \ind{a^1_h = a_B}} \\
        &\geq (1 - \alpha - \beta) \EE{\sum_{\text{odd } h=0}^{H-1}  R^{1,\fRew}_h } + \beta \left\lfloor \frac{H}{2}\right\rfloor
\end{align*}
as by definition clearly $\EE{\ind{a^1_h = a_B}} = 1$ for all odd $h$ and $R^\fRep_h \geq 0$ almost surely.

We analyze the terms $R^{1,\fRew}_h$ when the first agent follows $\vecpi^\text{br}$.
By the definition of the dynamics and $\vecpi^\text{br}$, it holds that 
\begin{align*}
    R^{1,\fRew}_h = g_1(\empdist{h-1}(s_{h-1}^1),\empdist{h-1}(\widebar{s}_{h-1}^1) )
\end{align*}
where $\widebar{s}_{h-1}^1 := \sleft$ if $s_{h-1}^1 = \sright$ and $\widebar{s}_{h-1}^1 := \sright$ if $s_{h-1}^1 = \sleft$.
As $\Prob[s^1_{h-1}=\cdot, \ldots, s^N_{h-1}=\cdot]$ at even step $h-1$ is permutation invariant, it holds that $\Prob[s^1_{h-1}=\cdot|\empdist{h-1}=\mu] = \mu(\cdot)$ for any $\mu\in\setG$.
Therefore,
\begin{align*}
    \Exop[R^{1,\fRew}_h] = &\sum_{\substack{\mu\in\setG\\ s\in\{\sleft, \sright\}}} \Prob[\empdist{h-1}=\mu] \Prob[s^1_{h-1}=s|\empdist{h-1}=\mu] \\
        &\Exop[R^{1,\fRew}_h|s^1_{h-1}=s,\empdist{h-1}=\mu] \\
    = & \sum_{\substack{\mu\in\setG\\ s\in\{\sleft, \sright\}}} \Prob[\empdist{h-1}=\mu] \mu(s) g_1(\mu(s),\mu(\widebar{s}) ) \geq \sfrac{1}{2},
\end{align*}
as for any $\mu$, if $s$ is such that $\mu(s)\geq\mu(\widebar{s})$ then $g_1(\mu(s),\mu(\widebar{s}) ) = 1$.
Furthermore, by the definition of the hitting time $\tau$, for any odd $h \geq 1$, $\EEc{R^{\fRew}_h}{2\tau < h} = \EEc{R^{\fRew}_h}{\empdist{h-1}(\sleft) \in \setG_*} = 1$,
as after time $2\tau$ the action $\actA$ will be optimal with reward $R^{\fRew}_h = 1$ almost surely, as $\vecpi^{br}$ chooses action $\actA$ at even steps.

Finally, using the lower bound of $\sfrac{1}{2}$ for $R^{\fRew}_h$ when $h<2\tau$ and that $R^{\fRew}_h = 1$ when $h > 2\tau$, we obtain:
\begin{align*}
    \EE{\sum_{\substack{h \text{ odd} \\ 0 \leq h \leq H}}  R^{\fRew}_h } = & \EE{ \sum_{\substack{h \text{ odd} \\ 0 \leq h \leq \min \{ 2\tau, H \}}} R^{1,\fRew}_h + \sum_{\substack{h \text{ odd} \\ \min \{ 2\tau, H \} + 1 \leq h < H}} R^{1,\fRew}_h} \\
        \geq & \EE{\frac{1}{2} \min \left\{ \tau, \left\lfloor \frac{H}{2}\right\rfloor \right\} + \left(\left\lfloor \frac{H}{2}\right\rfloor - \min \left\{ \tau, \left\lfloor \frac{H}{2}\right\rfloor \right\}\right)} \\
        \geq &  \left\lfloor \frac{H}{2}\right\rfloor - \frac{1}{2}\EE{\min \left\{ \tau, \left\lfloor \frac{H}{2}\right\rfloor \right\}} \\
        \geq &  \left\lfloor \frac{H}{2}\right\rfloor - \frac{1}{2}\EE{\tau }  = \left\lfloor \frac{H}{2}\right\rfloor - \frac{T_\tau}{2}
\end{align*}

Merging the inequalities above, we obtain
\begin{align*}
    &\,\JfinNi{1}(\vecpi^\text{br}, \vecpi^*, \ldots, \vecpi^*) \geq (1 - \alpha - \beta) \left(\left\lfloor \frac{H}{2}\right\rfloor - \frac{T_\tau}{2}\right) + \beta \left\lfloor \frac{H}{2}\right\rfloor.
\end{align*}

\textbf{Step 5: Bounding exploitability.}
Finally, we will upper bound also the expected reward of the FH-MFG-NE policy $\vecpi^*$ and hence lower bound the exploitability.
Our conclusion will be that $\vecpi^*$ suffers from a non-vanishing exploitability for large $H$, as $\vecpi^{\text{br}}$ becomes the best response policy after $H \gtrsim \log N$.
In this step, we assume the probability space induced by all $N$ agents following FH-MFG-NE policy $\vecpi^{\text{br}}$.

We have the definition
\begin{align*}
    \JfinNi{1}&(\vecpi^*, \vecpi^*, \ldots, \vecpi^*) 
        = \Exop \left[ \sum_{h=0}^{H-1} R(s_h^1, a_h^1, \empdist{h})\right] \\
        \leq & (1-\alpha-\beta)\Exop\left[ \sum_{\text{odd } h = 0}^{H-1} R^{1,\fRew}_h \right] + (\alpha +\beta)\left\lfloor \frac{H}{2}\right\rfloor
\end{align*}
This time, when $h$ odd and $h>2\tau$, it holds that $\Exop[R^{\fRew}_h|h>2\tau] = \sfrac{1}{2}$ since $\vecpi^*$ takes actions $\actA,\actB$ with equal probability in even steps, yielding $R^{\fRew}_h=1$ and $R^{\fRew}_h=0$ respectively almost surely.
As before,
\begin{align*}
    \EE{\sum_{\substack{h \text{ odd} \\ 0 \leq h \leq H}}  R^{\fRew}_h } = & \EE{ \sum_{\substack{h \text{ odd} \\ 0 \leq h \leq \min \{ 2\tau, H \}}} R^{1,\fRew}_h + \sum_{\substack{h \text{ odd} \\ \min \{ 2\tau, H \} + 1 \leq h < H}} R^{1,\fRew}_h} \\
        \leq & \EE{\min \left\{ \tau, \left\lfloor \frac{H}{2}\right\rfloor \right\} + \frac{1}{2} \left(\left\lfloor \frac{H}{2}\right\rfloor - \min \left\{ \tau, \left\lfloor \frac{H}{2}\right\rfloor \right\}\right)} \\
        = & \frac{1}{2} \EE{\left\lfloor \frac{H}{2}\right\rfloor + \min \left\{ \tau, \left\lfloor \frac{H}{2}\right\rfloor \right\}} \\
        \leq & \frac{1}{2}\left\lfloor \frac{H}{2}\right\rfloor + \frac{1}{2} \Exop[\tau] = \frac{1}{2}\left\lfloor \frac{H}{2}\right\rfloor + \frac{1}{2} T_\tau.
\end{align*}

The statement of the theorem then follows by lower bounding the exploitability as follows:
\begin{align*}
    \ExpfinNi{1}&(\vecpi^*, \vecpi^*, \ldots, \vecpi^*) \\
        = &\max_{\vecpi} \JfinNi{1}(\vecpi, \vecpi^*, \ldots, \vecpi^*) - \JfinNi{1}(\vecpi^*, \vecpi^*, \ldots, \vecpi^*) \\
        \geq &\JfinNi{1}(\vecpi^\text{br}, \vecpi^*, \ldots, \vecpi^*) - \JfinNi{1}(\vecpi^*, \vecpi^*, \ldots, \vecpi^*) \\
        \geq & (1-\alpha-\beta)\left(\left\lfloor \frac{H}{2}\right\rfloor - \frac{T_\tau}{2} - \frac{1}{2}\left\lfloor \frac{H}{2}\right\rfloor - \frac{T_\tau}{2}\right) -\alpha\left\lfloor \frac{H}{2}\right\rfloor \\
        \geq & (1-\alpha-\beta) \left( \frac{H}{4} - 24 - \frac{5}{9}\log_5 {N} \right) - \alpha \left\lfloor \frac{H}{2}\right\rfloor
\end{align*}
The above inequality implies that if $H \geq \log_2 N$, then
\begin{align*}
    \ExpfinNi{1}&(\vecpi^*, \vecpi^*, \ldots, \vecpi^*) \\
        \geq & (1-\alpha-\beta) \left( \frac{1}{4} - \frac{5}{9\log_2 5} \right) H - \alpha \frac{H}{2} - 24,
\end{align*}
which implies $\ExpfinNi{1}(\vecpi^*, \vecpi^*, \ldots, \vecpi^*) \geq \Omega(H)$ by choosing $\alpha,\beta$ small constants as $\frac{1}{4} - \frac{5}{9\log_2 5} > 0$.

\subsection{Upper Bound for Stat-MFG: Extended Proof of Theorem~\ref{theorem:upper_approx_stat}}

Let $\mu^*, \pi^*$ be a $\delta$-Stat-MFG-NE.
As before, the proof will proceed in three steps:
\begin{itemize}
    \item \textbf{Step 1.} Bounding the expected deviation of the empirical population distribution from the mean-field distribution $\EE{\|\empdist{h}-\mu^*\|_1}$ for any given policy $\vecpi$.
    \item \textbf{Step 2. } Bounding difference of $N$ agent value function $\JstatN$ and the infinite player value function $\Vstat$ in the stationary mean-field game setting.
    \item \textbf{Step 3. } Bounding the exploitability of an agent when each of $N$ agents are playing the Stat-MFG-NE policy.
\end{itemize}

\textbf{Step 1: Empirical distribution bound.}
We first analyze the deviation of the empirical population distribution $\widehat{\mu}_t$ over time from the stable distribution $\mu^*$.
For this, we state the following lemma and prove it using techniques similar to Corollary D.4 of \cite{yardim2023policy}.
\begin{lemma}
    Assume that the conditions of Theorem~\ref{theorem:upper_approx_stat} hold, and that $(\mu^*, \pi^*) \in \Delta_\setS$ is a Stat-MFG-NE.
    Furthermore, assume that the $N$ agents follow policies $\{ \pi^i\}_{i=1}^N$ in the $N$-Stat-MFG, define $\Delta_{\widebar{\pi}} := \frac{1}{N}\sum_i \| \widebar{\pi} - \pi^i\|_1$.
    Then, or any $t \geq 0$, we have
\begin{align*}
   \Exop \left[ \| \mu^* - \widehat{\mu}_{t} \|_1 \right] \leq \frac{t K_a \Delta_\pi}{2} +  \frac{2(t+1)\sqrt{|\setS|}}{\sqrt{N}}.
\end{align*}
\end{lemma}
\begin{proof}
    $\mathcal{F}_t$ as the $\sigma$-algebra generated by the states of agents $\{s_t^i\}$ at time $t$.
    For $\widehat{\mu_0}$, we have by definitions that
    \begin{align*}
        \Exop\left[ \widehat{\mu_0} \right] &= \Exop\left[ \frac{1}{N}\sum_{i} \vece_{s_t^i} \right] = \mu^* \\
        \Exop\left[ \| \widehat{\mu_0} - \mu^* \|_2^2 \right] &= \Exop\left[ \frac{1}{N^2} \sum_{i} \left\|\left(\vece_{s_t^i} - \mu^* \right) \right\|_2^2 \right] \leq \frac{4}{N}
    \end{align*}
    where the last line follows by independence.
    The two above imply $\Exop\left[ \| \widehat{\mu_0} - \mu^*\|_1 \right] \leq \frac{2\sqrt{|\setS|}}{\sqrt{N}}$.
    
    Next, we inductively calculate:
    \begin{align}
        \Exop&\left[ \widehat{\mu}_{t+1} \middle | \mathcal{F}_t \right] = \Exop\left[ \frac{1}{N} \sum_{s'\in \setS} \sum_{i=1}^N \mathbbm{1}(s_{t+1}^i = s') \vece_{s'} \middle| \mathcal{F}_t \right] \notag \\
            & = \sum_{s'\in \setS} \vece_{s'} \sum_{i=1}^N \frac{1}{N}  \widebar{P}(s'|s_{t}^i, \pi^i(s_{t}^i), \widehat{\mu}_t), \label{eq:statmfg_muexpt}\\
     \Exop&[ \| \widehat{\mu}_{t+1} - \Exop[ \widehat{\mu}_{t+1}  | \mathcal{F}_t ] \|_2^2 | \mathcal{F}_t ] \notag \\ 
        &= \frac{1}{N^2}  \sum_{i=1}^N \Exop[ \| \vece_{s_{t+1}^i} -  \Exop[  \vece_{s_{t+1}^i}  | \mathcal{F}_t ] \|_2^2 | \mathcal{F}_t ] \leq \frac{4}{N}.\label{eq:statmfg_muvart}
    \end{align}
    We bound the $\ell_1$ distance to the stable distribution as
    \begin{align*}
        \Exop&\left[ \| \widehat{\mu}_{t+1} - \mu^* \|_1 | \mathcal{F}_t\right] \\
            &\leq \underbrace{\Exop\left[ \| \Exop\left[ \widehat{\mu}_{t+1}  | \mathcal{F}_t\right] | \mathcal{F}_t \right] - \mu^* \|_1} _{(\square)} + \underbrace{\Exop\left[\| \Exop\left[ \widehat{\mu}_{t+1}  | \mathcal{F}_t\right] - \widehat{\mu}_{t+1} \|_1 \mathcal{F}_t\right] }_{(\triangle)}.
    \end{align*}
    The two terms can be bounded separately using Inequalities~\eqref{eq:statmfg_muexpt} and \eqref{eq:statmfg_muvart}.
    \begin{align*}
        (\triangle) \leq & \sqrt{|\setS|} \Exop\left[\| \Exop\left[ \widehat{\mu}_{t+1}  | \mathcal{F}_t\right] - \widehat{\mu}_{t+1} \|_2 \mathcal{F}_t\right] \\
            \leq & \sqrt{|\setS|} \sqrt{\Exop\left[\| \Exop\left[ \widehat{\mu}_{t+1}  | \mathcal{F}_t\right] - \widehat{\mu}_{t+1} \|_2^2 \mathcal{F}_t\right]} \leq \frac{2\sqrt{|\setS|}}{\sqrt{N}}\\
        (\square) = & \left\| \sum_{s'\in \setS} \vece_{s'} \sum_{i=1}^N \frac{1}{N}  \widebar{P}(s'|s_{t}^i, \pi^i(s_{t}^i), \widehat{\mu}_t) - \mu^*\right\|_1 \\
            = & \left\| \sum_{s'\in \setS} \vece_{s'} \sum_{i=1}^N \frac{1}{N}  \widebar{P}(s'|s_{t}^i, \pi^i(s_{t}^i), \widehat{\mu}_t) - \Gamma_{pop}(\pi^*,\mu^*)\right\|_1 \\
            \leq & \left\| \sum_{i=1}^N \frac{1}{N}  \widebar{P}(\cdot|s_{t}^i, \pi^i(s_{t}^i), \widehat{\mu}_t) - \sum_{i=1}^N \frac{1}{N}  \widebar{P}(\cdot|s_{t}^i, \pi^*(s_{t}^i), \widehat{\mu}_t)\right\|_1 \\
                &+ \left\| \sum_{s'\in \setS} \widehat{\mu}_t(s') \widebar{P}(s'|s_{t}^i, \pi^i(s_{t}^i), \widehat{\mu}_t) - \Gamma_{pop}(\pi^*,\mu^*)\right\|_1 \\
            \leq & \frac{K_a}{2N} \sum_i \| \pi^* - \pi^i \|_1 + \left\| \Gamma_{pop}(\pi^*,\widehat{\mu}_t) - \Gamma_{pop}(\pi^*,\mu^*)\right\|_1 \\
            \leq &\frac{K_a \Delta_\pi}{2} + \| \mu^* - \widehat{\mu}_{t} \|_1
    \end{align*}
    Hence, by the law of total expectation, we can conclude
    \begin{align*}
        \Exop \left[ \| \mu^* - \widehat{\mu}_{t+1} \|_1 \right] \leq \Exop \left[ \| \mu^* - \widehat{\mu}_{t} \|_1 \right] + \frac{K_a \Delta_\pi}{2} + \frac{2\sqrt{|\setS|}}{\sqrt{N}}
    \end{align*}
    or inductively,
    \begin{align*}
        \Exop \left[ \| \mu^* - \widehat{\mu}_{t} \|_1 \right] \leq \frac{t K_a \Delta_\pi}{2} +  \frac{2(t+1)\sqrt{|\setS|}}{\sqrt{N}}.
    \end{align*}
\end{proof}

\textbf{Step 2: Bounding difference in value functions.}
Next, we bound the differences in the infinite-horizon 
\begin{lemma}
    \label{lemma:bound_n_mfg_stat_similar_reward}
    Suppose $N$-Stat-MFG agents follow the same sequence of policy $\pi^*$. 
    Then for all $i$,
    \begin{align*}
    |\JstatN&(\pi^*,\dots,\pi^*) - \Vstat(\mu^*, \pi^*)
    |
    \\
        &\leq \frac{\gamma}{1-\gamma} \left(L_\mu + \frac{L_s}{2}\right) \frac{2\sqrt{|\setS|}}{\sqrt{N}} 
    \end{align*}
\end{lemma}
\begin{proof}
    For ease of reading, in this proof expectations, probabilities, and laws of random variables will be denoted $\Exop_{\infty}, \Prob_{\infty}, \Law_{\infty}$ respectively over the infinite player finite horizon game and $\Exop_{N}, \Prob_{N}, \Law_{N}$ respectively over the $N$-player game.
    Due to symmetry in the $N$ agent game, any permutation $\sigma: [N] \rightarrow [N]$ of agents does not change their distribution, that is $\Law_N(s_t^1,\dots, s_t^N) = \Law_N(s_t^{\sigma(1)},\dots, s_t^{\sigma(N)})$.
    We can then conclude that:
    \begin{align*}
        \EEN{
        \rwd(s_t^1, a_t^1, \empdist{h})
        }
        &=
        \frac{1}{N}\sum_{i=1}^N
        \EEN{
        \rwd(s_t^i, a_t^i, \empdist{t})
        } \\
        &=
        \EEN{
            \sum_{s
            \in\setS}\empdist{t}(s)\rwdpol(s, \pi_t(s), \empdist{t}).
        }
    \end{align*}
    Therefore, we by definition:
    \begin{align*}
        \JstatNi{1}(\vecpi,\dots,\vecpi) = \EEN{\sum_{t=0}^{\infty}  \sum_{s
            \in\setS}\empdist{t}(s)\rwdpol(s, \pi^*(s), \empdist{t})}.
    \end{align*}
    Next, in the Stat-MFG, we have that for all $t \geq 0$,
    \begin{align*}
        \Prob_\infty(s_t = \cdot) &= \mu^*, \\ \Prob_\infty(s_{t+1} = \cdot) &= \sum_{s\in\setS}\Prob_\infty(s_{t} = s)\Prob_\infty(s_{t} = \cdot|s_{t} = s) \\
        &= \Gamma_P(\Prob_\infty(s_{t} = s), \pi^*) = \mu^*,
    \end{align*}
    so by induction $\Prob_\infty(s_t = \cdot) = \mu^*$.
    Then we can conclude that
    \begin{align*}
        \Vstat(\mu^*, \pi^*) &= \EEinf{\sum_{t=0}^{\infty} \gamma^t R(s_t, \pi^*(s_t), \mu_t)} \\
            &= \sum_{t=0}^{\infty} \gamma^t \sum_{s\in\setS} \mu^*(s) R(s, \pi^*(s), \mu^*),
    \end{align*}
    by a simple application of the dominated convergence theorem.
    We next bound the differences in truncated expect reward until some time $T > 0$:
    \begin{align*}
            \biggl| &\EEN{\sum_{t=0}^{T} \gamma^t \sum_{s
            \in\setS}\empdist{t}(s) \rwdpol(s, \pi^*(s), \empdist{t})} \\
                &- \sum_{t=0}^{T} \gamma^t\sum_{s\in\setS} \mu_t(s) R(s, \pi^*(s), \mu_t)\biggr| \\
        \leq & \EEN{ \sum_{t=0}^{T} \gamma^t \left| \sum_{s
            \in\setS}\left(\empdist{t}(s)\rwdpol(s, \pi^*(s), \empdist{t}) - \mu^*(s) R(s, \pi^*(s), \mu^*) \right) \right|} \\
        \leq &\EEN{ \sum_{t=0}^{T} \gamma^t \left(\frac{L_s}{2}\|\mu^* -\empdist{t}\|_1 + L_\mu \|\mu^* -\empdist{t}\|_1\right)} \\
        \leq & \sum_{t=0}^{T} \gamma^t \left(L_\mu + \frac{L_s}{2}\right) \EEN{ \|\mu^* -\empdist{t}\|_1} \\
        \leq &\frac{1}{(1-\gamma)^2} \left(L_\mu + \frac{L_s}{2}\right) \frac{2\sqrt{|\setS|}}{\sqrt{N}}
    \end{align*}
    Taking $T \rightarrow \infty$ and applying once again the dominated convergence theorem the result is obtained.
\end{proof}

\textbf{Step 3: Bounding difference in policy deviation.}
Finally, to conclude the proof of the main theorem of this section, we will prove that the improvement in expectation due to single-sided policy changes are at most of order $\cO\left(\delta + \frac{1}{\sqrt{N}}\right)$.

\begin{lemma}
    Suppose we have two policy sequences $\pi^*, \pi \in \Pi$ and $\mu^* \in \Delta_\setS$ such that $\Gamma_P(\mu^*, \pi^*) = \mu^*$ and $\Gamma_P(\cdot, \pi^*)$ is non-expansive.
    Then,
    \begin{align*}
    &\left|\JstatNi{1}(\pi',\pi^*,\dots,\pi^*) - \Vstat(\mu^*, \pi')
    \right|\\
    &\leq 
        \sum_{t=0}^{\infty} \gamma^t
        \biggl(
        \lmu 
        \EE{
            \|\empdist{t}-\mu^{\vecpi}_t\|_1
        }
        +
        \kmu
        \sum_{t'=0}^{t-1}
        \EE{\|\empdist{t'}-\mu_{t'}^{\vecpi}\|_1}
        \biggr) \\
    &\leq \left( \frac{K_a}{2N} + \frac{2\sqrt{|\setS|}}{\sqrt{N}} \right)   \frac{\sfrac{L_\mu}{2} + K_\mu }{(1-\gamma)^3} 
    \end{align*}    
\end{lemma}
\begin{proof}
    For the truncated game $T$, it still holds by the derivation in the FH-MFG that:
    \begin{align*}
        |&\EEN{\rwd(s_{t}^1, a_{t}^1,\empdist{t})}
    -
    \EEinf{\rwd(s_{t}, a_{t},\mu_{t}^{\vecpi})}| \\
    &\leq \frac{L_\mu}{2}\Exop_N\biggl[
                \|\mu^{\vecpi}_{t}-
                \empdist{t}\|_1
            \biggr] +  \kmu\sum_{t'=0}^{t-1}
            \Exop_N\biggl[
                \|\mu^{\vecpi}_{t'}-
                \empdist{t'}\|_1
            \biggr].
    \end{align*}
    We take the limit $T\rightarrow\infty$ and apply the dominated convergence theorem to obtain the state bound, also noting that $\sfrac{1}{2}\cdot\sum_t (t+1)(t+2)\gamma^t  \leq \frac{1}{(1-\gamma)^3}$. 
\end{proof}

\textbf{Conclusion and Statement of the Result.}
Finally, if $\mu^*, \pi^*$ is a $\delta$-Stat-MFG-NE, by definition we have that:
    By definition of the Stat-MFG-NE, we have:
    \[
    \delta\geq \Expfin(\vecpi_\delta)
    =
    \max_{\pi' \in \Pi} \Vstat( \mu^*, \pi') - \Vstat ( \mu^*, \pi^* )
    \]
Then using the two bounds from Steps 2,3 and the fact that $\pi^*$ $\delta$-optimal with respect to $\mu^*$:
\begin{align*}
    &\max_{\pi' \in \Pi} \JfinNi{1}( \pi', \pi^*, \ldots, \pi^*) - \JfinNi{1} ( \pi^*, \pi^*,\ldots,\pi^*) \\
    &\leq 2\delta + \left( \frac{K_a}{2N} + \frac{2\sqrt{|\setS|}}{\sqrt{N}} \right)  \frac{\sfrac{L_\mu}{2} + K_\mu }{(1-\gamma)^3} + \frac{L_\mu + \sfrac{L_s}{2}}{(1-\gamma)^2} \biggl(\frac{2\sqrt{|\setS|}}{\sqrt{N}} \biggr)
\end{align*}

\subsection{Lower Bound for Stat-MFG: Extended Proof of Theorem~\ref{theorem:lower_approx_stat}}

Similar to the finite horizon case, we define constructively the counter-example: the idea and the nature of the counter-example remain the same.
However, minor details of the construction are modified, as it will not hold immediately that all agents are on states $\{ \sleft, \sright\}$ on even times $t$, and that the Stat-MFG-NE is unique as before.

\textbf{Defining the Stat-MFG.}
We use the same definitions for $\setS, \setA, \fRew, \fRep,\fProb$ as in the FH-MFG case.
Define the convenience functions $Q_L, Q_R$ as
\begin{align*}
    Q_L(\mu) &:= \frac{\mu(\slA) + \mu(\slB)}{\max\{\mu(\slA) + \mu(\slB) + \mu(\srA) + \mu(\srB), \sfrac{4}{9} \}}, \\
    Q_R(\mu) &:= \frac{\mu(\srA) + \mu(\srB)}{\max\{\mu(\slA) + \mu(\slB) + \mu(\srA) + \mu(\srB), \sfrac{4}{9} \}}.
\end{align*}
We define the transition probabilities:
\begin{align*}
    \text{If } s &\in \{\slA, \slB, \srA, \srB\},  \forall \mu, a:  \\ P(s' | s, a, \mu) &= \begin{cases}
        \fProb(Q_L(\mu)), \text{ if } s' = \sright, s \in \{\slA, \slB\} \\
        \fProb(Q_R(\mu)), \text{ if } s' = \sleft, s \in \{\slA, \slB\} \\
        \fProb(Q_L(\mu)), \text{ if } s' = \sright, s \in \{\srA, \srB\}  \\
        \fProb(Q_R(\mu)), \text{ if } s' = \sleft, s \in \{\srA, \srB\} 
    \end{cases},
\end{align*}
and define $P(\sleft, a, \mu), P(\sright, a, \mu)$ as before.
With previous Lipschitz continuity results, it follows that $P\in \PL_{\sfrac{9}{8\varepsilon}}$.

Similarly, we modify the reward function $R$ as follows:
\begin{align*}
    R(\sleft, \actA, \mu) = &R(\sleft, \actB, \mu) = 0, \\
    R(\sright, \actA, \mu) = &R(\sright, \actB, \mu) = 0, \\
    \begin{pmatrix}
        R(\slA, \actA, \mu) \\
        R(\slB, \actA, \mu)
    \end{pmatrix} = & (1-\alpha-\beta)\fRew\big(Q_L(\mu), Q_R(\mu) \big) +\alpha \fRep(\mu(\slA), \mu(\slB))  \\
    \begin{pmatrix}
        R(\slA, \actB, \mu) \\
        R(\slB, \actB, \mu)
    \end{pmatrix} = & (1-\alpha-\beta)\fRew\big(Q_L(\mu), Q_R(\mu) \big) + \fRep(\mu(\slA), \mu(\slB)) \\
        &+ \beta \vecone \\
    \begin{pmatrix}
        R(\srA, \actA, \mu) \\
        R(\srB, \actA, \mu)
    \end{pmatrix} = &(1-\alpha-\beta)\fRew\big(Q_R(\mu), Q_L(\mu) \big) + \alpha \fRep(\mu(\srA), \mu(\srB)) \\
    \begin{pmatrix}
        R(\srA, \actB, \mu) \\
        R(\srB, \actB, \mu)
    \end{pmatrix} = &(1-\alpha-\beta)\fRew\big(Q_R(\mu), Q_L(\mu) \big) + \alpha \fRep(\mu(\srA), \mu(\srB)) \\
        &+ \beta\vecone,
\end{align*}
simple computation shows that $R\in\RL_3$.
In this proof, unlike the $N$-FH-SAG case, $\alpha$ will be chosen as a function of $N$, namely $\alpha = \mathcal{O}(e^{-N})$.

\textbf{Step 1: Solution of the Stat-MFG.}
We solve the infinite agent game: let $\mu^*,\pi^*$ be an Stat-MFG-NE.
By simple computation, one can see that for any stationary distribution $\mu^*$ of the game, probability must be distributed equally between groups of states $\{\sleft, \sright \}$ and $\{\slA, \slB, \srA, \srB \}$, that is,
\begin{align*}
    \mu^*(\sleft) + \mu^*(\sright) = \sfrac{1}{2} , \\
    \mu^*(\slA) + \mu^*(\slB) + \mu^*(\srA) + \mu^*(\srB) = \sfrac{1}{2}.
\end{align*}
It holds by the stationarity equation $\Gamma_P(\mu^*, \pi^*) = \pi^*$ that
\begin{align*}
    \mu^*(\sleft) = &\mu^*(\slA) + \mu^*(\slB), \\
    \mu^*(\sright) = &\mu^*(\srA) + \mu^*(\srB), \\
    \mu^*(\sleft) = & \sum_{s\in\setS} \mu^*(s) \pi^*(a|s) P(\sleft|s,a,\mu^*) \\
                   = & P(\sleft|\slA,\actA,\mu^*), \\
    \mu^*(\sright) = &\sum_{s\in\setS} \mu^*(s) \pi^*(a|s)P(\sright|s,a,\mu^*)  \\
                    = & P(\sright|\slA,\actA,\mu^*),
\end{align*}
as $P(\sright|s, a,\mu^*) = P(\sright|s,a,\mu^*)$ and similarly $P(\sleft|s, a,\mu^*) = P(\sleft|s,a,\mu^*)$ for any $s\in\{\slA,\slB, \srA, \srB\}, a\in\setA$.
If $\mu^*(\sleft) > \sfrac{1}{4}$, then by definition $P(\sleft|\slA,\actA,\mu^*) < \sfrac{1}{4}$, and similarly if $\mu^*(\sleft) < \sfrac{1}{4}$, then by definition $P(\sleft|\slA,\actA,\mu^*) > \sfrac{1}{4}$.
So it must be the case that $\mu^*(\sleft) = \mu^*(\sright) = \sfrac{1}{4}$.
Then the unique Stat-MFG-NE must be 
\begin{align*}
    \pi^*(a|s) := &
        \begin{cases}
            1, \text{if $a = \actB, s \in \{ \slA, \slB, \srA, \srB\}$} \\
            \frac{1}{2}, \text{if $s\in\{ \sleft, \sright\}$} \\
            0, \text{if $a = \actA, s \in \{ \slA, \slB, \srA, \srB\}$},
        \end{cases} \\
    \mu^*(\srA) &= \mu^*(\slA) = \mu^*(\srB) = \mu^*(\slB) = \sfrac{1}{8},
\end{align*}
as otherwise the action $\argmin_{a\in\setA} \pi^*(a|\sright)$ will be a better response in state $\sright$ and the action $\argmin_{a\in\setA} \pi^*(a|\sleft)$ will be optimal in state $\sright$.

\textbf{Step 2: Expected population deviation in $N$-Stat-SAG.}
We fix $\sfrac{1}{2\varepsilon} = 3$, define the random variable $\widebar{N} := N (\empdist{0}(\sright) + \empdist{0}(\sleft))$.
We will analyze the population under the event $\widebar{N} := \{ \left|\sfrac{\widebar{N}}{N} - \sfrac{1}{2}\right| \leq \sfrac{1}{18}\}$,
which holds with probability $\Omega(1 - e^{-N^2})$ by the Hoeffding inequality.
Under the event $\widebar{E}$, it holds that $\empdist{t}(\slA) + \empdist{t}(\slA) + \empdist{t}(\slA) + \empdist{t}(\slA) > \sfrac{4}{9}$ almost surely at all $t$.

Fix $N_0 \in \mathbb{N}_{>0}$ such that $\left|\sfrac{N_0}{N} - \sfrac{1}{2}\right| \leq \sfrac{1}{18}$, in this step we will condition on $E_0 := \{\widebar{N} := N_0 \}$.
Once again define the random process $X_m$ for $m\in\mathbb{N}_{\geq 0}$ such that
\begin{align*}
    X_m := \begin{cases}
            \frac{\empdist{2m}(\sleft)}{\empdist{2m}(\sleft) + \empdist{2m}(\sright)}, \text{ if $m$ odd} \\
            \frac{\empdist{2m}(\sright)}{\empdist{2m}(\sleft) + \empdist{2m}(\sright)}, \text{ if $m$ even}
        \end{cases}
\end{align*}
with the modification at odd $m$ necessary because of the difference in dynamics $P$ (oscillating between $\sleft,\sright$) from the FH-SAG case.
It still holds that $X_m$ is Markovian, and given $X_m$ we have $N_0 X_{m+1} \sim \operatorname{Binom}(N_0, \fProb(X_m))$.
As before, $X_m$ is independent from the policies of agents.

Define $K := \lfloor \log_2 \sqrt{N_0} \rfloor$, $\setG := \{ \sfrac{k}{N_0} : k = 0,\ldots, N_0 \}$, $\setG_* := \{ 0, 1 \} \subset \setG$ and the level sets once again as
\begin{align*}
    \setG_{-1} := \setG, \quad \setG_k &:= \left\{ x \in \setG :\left|x - \frac{1}{2}\right| \geq \frac{2^k}{2\sqrt{N_0}}  \right\} \text{ when $k \leq K$}, \\
    \setG_{K+1} &:= \setG_*.
\end{align*}
As before, using the Markov property, Hoeffding, and the fact that $|\fProb(x) - \sfrac{1}{2}| \geq \sfrac{1}{2\epsilon} | x - \sfrac{1}{2} |$ we obtain $\forall k \in 0, \ldots, K-1$, $\forall m$ that
\begin{align*}
    \Prob[X_{m+1} \in \setG_0 | X_{m} \in \setG_{-1}, E_0] &\geq \sfrac{1}{20} \\
    \Prob[X_{m+1} \in \setG_{k+1} |X_{m} \in \setG_{k}, E_0] &\geq \alpha_k := 1 - 2 \exp \left\{ - \frac{1}{8}4^{k+1}\right\}, 
\end{align*}
hence from the analysis before we have the lower bound
\begin{align*}
    \Exop[|X_m -\sfrac{1}{2}| \, |E_0] \geq C_1 \, \min \left\{ \frac{2^m}{\sqrt{N_0}}, 1\right\},
\end{align*}
for some absolute constant $C_2 > 0$.

\textbf{Step 3. Exploitability lower bound.}
As in the case of FH-MFG, the ergodic optimal policy is given by 
\begin{align*}
    \widebar{\pi} (a|s) = \begin{cases}
        1, \text{ if $s=\sleft$, $a = \actA$}\\
        1, \text{ if $s=\sright$, $a = \actA$}\\
        1, \text{ if $s\notin \{ \sleft, \sright\}$, $a = \actB$}\\
        0, \text{ otherwise}
    \end{cases}
\end{align*}
We define the shorthand functions
\begin{align*}
\setS^* := \{\sleft, \sright\}, \quad
&Q(\mu):= (Q_L(\mu), Q_R(\mu)),\\
Q_{\text{min}}(\mu) := \min\{Q_L(\mu), Q_R(\mu)\}, \quad &Q_{\text{max}} := \max\{Q_L(\mu), Q_R(\mu)\}.
\end{align*}

We condition on $E_{\setS^*} := \{s^1_0 \in \setS^* \}$, that is the first agent starts from states $\{\sleft, \sright\}$, the analysis will be similar under event $E^c_{\setS^*}$.
As in the case of FH-MFG, due to permutation invariance, it holds for any odd $t$ and $\mu\in\{\mu' \in \Delta_{\setS^*}: N_0 \mu'  \in \mathbb{N}_{>0}^2\}$ that
\begin{align*}
    &\Prob[s_t^1 \in \{\slA,\slB\}|E_0, E_{\setS^*}, Q(\empdist{t}) = \mu] = Q_L(\mu) \\
	&\Prob[s_t^1 \in \{\srA,\srB\}|E_0, E_{\setS^*}, Q(\empdist{t}) = \mu] = Q_R(\mu),
\end{align*}
therefore expressing the error component due to $\fRew$ as $R^{1,\fRew}_t$ and expressing some repeating conditionals as $\bullet$:
\begin{align*}
    \widebar{G}_t^\mu :=&\Exop\left[R^{1,\fRew}_t \middle|E_0, E_{\setS^*}, Q(\empdist{t}) = \mu, a_t^1 \sim \widebar{\pi}(s_t^1), \substack{a_t^i \sim \pi^*(s_t^i),\\ \text{when $i\neq 1$}}\right] \\
    = &\sum_{s\in\setS^*} \Prob[s^1_{t}=s|Q(\empdist{t})=\mu, \bullet] \Exop[R^{1,\fRew}_t|s^1_{t}=s,Q(\empdist{t})=\mu, \bullet] \\
    =&
    \frac{Q_{\text{max}}(\mu)}{Q_{\text{max}}(\mu)} Q_{\text{max}}(\mu) + \frac{Q_{\text{min}}(\mu)}{Q_{\text{max}}(\mu)} Q_{\text{min}}(\mu).
\end{align*}
Similarly, since $\pi^*(a|s) =\sfrac{1}{2}$ for any $s\in\setS^*$, it holds that
\begin{align*}
    G_t^\mu := &\Exop\left[R^{1,\fRew}_t \middle|E_0, E_{\setS^*}, Q(\empdist{t}) = \mu, \substack{a_t^i \sim \pi^*(s_t^i),\\ \text{ $\forall i$}}\right] \\
    = & \frac{1}{2} \frac{Q_{\text{min}}(\mu)}{Q_{\text{max}}(\mu)} + \frac{1}{2} \frac{Q_{\text{max}}(\mu)}{Q_{\text{max}}(\mu)}.
\end{align*}
Therefore, given the population distribution between $\slA, \slB$ and $\srA, \srB$, the expected difference in rewards for the two policies is
\begin{align*}
    \widebar{G}_t^\mu - G_t^\mu = &\left( Q_{\text{max}}(\mu) - \frac{1}{2}\right) + \left(Q_{\text{min}}(\mu) - \frac{1}{2} \right) \frac{Q_{\text{min}}(\mu)}{Q_{\text{max}}(\mu)} \\
    = &\left( Q_{\text{max}}(\mu) - \frac{1}{2}\right) + \left(\frac{1}{2} -  Q_{\text{max}}(\mu)\right) \frac{Q_{\text{min}}(\mu)}{Q_{\text{max}}(\mu)} \\
    = &\left( Q_{\text{max}}(\mu) - \frac{1}{2}\right) \left( 1 - \frac{Q_{\text{min}}(\mu)}{Q_{\text{max}}(\mu)}\right) \\
    \geq &2\left( Q_{\text{max}}(\mu) - \frac{1}{2}\right)^2.
\end{align*}
Therefore from above, we conclude that
\begin{align*}
    \Exop[\widebar{G}_t^{\empdist{t}} - G_t^{\empdist{t}} \, |E_0] &\geq \Exop[2|X_{\frac{t-1}{2}} - \sfrac{1}{2}|^2 \, |E_0, E_{\setS^*}] 
    \geq 2C_1^2 \, \min \left\{ \frac{2^{t}}{2N_0}, 1\right\}.
\end{align*}
Using the lower bound above, the conditional expected difference in discounted total reward is
\begin{align*}
    \Exop&\big[\sum_{t=0}^{\infty} \gamma^{t} R(s^1_t, a^1_t, \empdist{t}) |E_0, E_{\setS^*}, a_t^1 \sim \widebar{\pi}(s_t^1), \substack{a_t^i \sim \pi^*(s_t^i),\\ \text{when $i\neq 1$}}\big] \\
        & \quad - \Exop\big[\sum_{t=0}^{\infty} \gamma^{t} R(s^1_t, a^1_t, \empdist{t}) |E_0, E_{\setS^*}, \substack{a_t^i \sim \pi^*(s_t^i),\\ \text{$\forall i$}}\big] \\
    &\geq (1-\alpha-\beta)\sum_{k=0}^{\infty} 2C_1^2 \, \gamma^{2k+1}\min \left\{ \frac{2^{2k}}{N_0}, 1\right\} -\frac{2\alpha}{1-\gamma} \\
    &\geq \frac{ C_2 }{N_0} \sum_{k=0}^{\lfloor \log_{4} N_0 \rfloor} (4\gamma^2)^k + \frac{ C_3 }{N_0} \sum_{k=\lfloor \log_{4} N_0 \rfloor}^{\infty} \gamma^{2k} - \frac{2\alpha}{1-\gamma}
    \\
    &\geq \frac{ C_4 ((4\gamma^2)^{\log_{4} N_0} - 1)}{N_0} + C_5 \frac{(\gamma^2)^{\log_{4} N_0}N_0^{-1}}{1-\gamma^2} - \frac{2\alpha}{1-\gamma} 
    \\
    & \geq C_6 N_0^{\log_{2} \gamma} + C_7 \frac{N_0^{\log_{2} \gamma - 1}}{1-\gamma} -\frac{2\alpha}{1-\gamma}.
\end{align*}
Taking expectation over $N_0$ (using $\Exop[\widebar{N}|E^*] = \sfrac{N}{2}$ and Jensen's):
\begin{align*}
    \Exop&\big[\sum_{t=0}^{\infty} \gamma^{t} R(s^1_t, a^1_t, \empdist{t}) |E^*, E_{\setS^*}, a_t^1 \sim \widebar{\pi}(s_t^1), \substack{a_t^i \sim \pi^*(s_t^i),\\ \text{when $i\neq 1$}}\big] \\
        &- \Exop\big[\sum_{t=0}^{\infty} \gamma^{t} R(s^1_t, a^1_t, \empdist{t}) |E^*, E_{\setS^*}, \substack{a_t^i \sim \pi^*(s_t^i),\\ \text{$\forall i$}}\big] \\
    \geq & C_6 N_0^{\log_{2} \gamma} + C_7 \frac{N_0^{\log_{2} \gamma - 2}}{1-\gamma} -\frac{2\alpha}{1-\gamma}
\end{align*}

While the analysis above assumes event $E_{\setS^*}$, the same analysis lower bound follows with a shift between even and odd steps when $s^1_0 \notin \setS^*$, hence 
\begin{align*}
    \Exop&\big[\sum_{t=0}^{\infty} \gamma^{t} R(s^1_t, a^1_t, \empdist{t}) |E^*, a_t^1 \sim \widebar{\pi}(s_t^1), \substack{a_t^i \sim \pi^*(s_t^i),\\ \text{when $i\neq 1$}}\big] \\
    &- \Exop\big[\sum_{t=0}^{\infty} \gamma^{t} R(s^1_t, a^1_t, \empdist{t}) |E^*, \substack{a_t^i \sim \pi^*(s_t^i),\\ \text{$\forall i$}}\big] \\
    \geq & C_6 N_0^{\log_{2} \gamma} + C_7 \frac{N_0^{\log_{2} \gamma - 2}}{1-\gamma} -\frac{2\alpha}{1-\gamma}
\end{align*}

Finally, we conclude the proof with the observation
\begin{align*}
    &\max_{\pi} \JstatNi{1}(\pi, \vecpi^*, \ldots, \vecpi^*) - \JfinNi{1}(\vecpi^*, \vecpi^*, \ldots, \vecpi^*) \\
     \geq & \JstatNi{1}(\widebar{\pi}, \vecpi^*, \ldots, \vecpi^*) - \JfinNi{1}(\vecpi^*, \vecpi^*, \ldots, \vecpi^*) \\
    \geq & C_6 N_0^{\log_{2} \gamma} + C_7 \frac{N_0^{\log_{2} \gamma - 2}}{1-\gamma} -\frac{2\alpha}{1-\gamma} - (1-\gamma)^{-1} \Prob[\widebar{E}^c], 
\end{align*}
where $\Prob[\widebar{E}^c] = O(e^{-N^2})$ and we pick $\alpha = \cO(e^{-N})$.

\section{Intractability Results}

\subsection{Fundamentals of PPAD}

We first introduce standard definitions and tools, mostly taken from \cite{daskalakis2009complexity, goldberg2011survey, papadimitriou1994complexity}.

\paragraph{Notations.}
For a finite set $\Sigma$, we denote by $\Sigma^n$ the set of tuples $n$ elements from $\Sigma$, and by $\Sigma^* = \bigcup_{n\geq 0} \Sigma^n $ the set of finite sequences of elements of $\Sigma$.
For any $\alpha \in \Sigma$, let $\alpha^n \in \Sigma^n$ denote the $n$-tuple $(\underbrace{\alpha, \ldots, \alpha} _ {\text{$n$ times}})$.
For $x \in \Sigma^*$, by $|x|$ we denote the length of the sequence $x$.
Finally, the following function will be useful, defined for any $\alpha > 0$:
\begin{align*}
    u_\alpha: &\mathbb{R} \rightarrow [0, \alpha]\\
    u_\alpha(x) &:= \max\{ 0, \min \{ \alpha, x\}\} =  \begin{cases}
        \alpha, \text{ if $x \geq \alpha$}, \\
        x, \text{ if $0\leq x \leq \alpha$}, \\
        0, \text{ if $x \leq 0$}.
    \end{cases}
\end{align*}

We define a search problem $\setS$ on alphabet $\Sigma$ as a relation from a set $\setI_\setS \subset \Sigma^*$ to $\Sigma^*$ such that for all $x \in \setI_\setS$, the image of $x$ under $\setS$ satisfies $\setS_x \subset \Sigma^{|x|^k}$ for some $k \in \mathbb{N}_{> 0}$, and given $y \in \Sigma^{|x|^k}$m whether $y\in \setS_x$ is decidable in polynomial time.

Intuitively speaking, PPAD is the complexity class of search problems that can be shown to always have a solution using a ``parity argument'' on a directed graph.
The simplest complete example (the example that defines the problem class) of PPAD problems is the computational problem \textsc{End-of-The-Line}. 
The problem, formally defined below, can be summarized as such: given a directed graph where each node has in-degree and out-degree at most one and given a node that is a source in this graph (i.e., no incoming edge but one outgoing edge), find another node that is a sink or a source.
Such a node can be always shown to exist using a simple parity argument.

\begin{definition}[\textsc{End-of-The-Line} \cite{daskalakis2009complexity}]
    The computational problem \textsc{End-of-The-Line} is defined as follows: given two binary circuits $S, P$ each with $n$ input bits and $n$ output bits such that $P(0^n) = 0^n \neq S(s^n)$, find an input $x \in \{0,1\}^n$ such that $P(S(x)) \neq x \text{ or } S(P(x)) \neq x \neq 0^n$.
\end{definition}

The obvious solution to the above is to follow the graph node by node using the given circuits until we reach a sink: however, this can take exponential time as the graph size can be exponential in the bit descriptions of the circuits.
It is believed that \textsc{End-of-The-Line} is difficult \cite{goldberg2011survey}, that there is no efficient way to use the bit descriptions of the circuits $S, P$ to find another node with degree 1.

\subsection{Proof of Intractability of Stat-MFG}

We reduce any \CompGC{} problem to the problem \CompStat{} for some simple transition function $P \in \setP^{\text{Sim}}$.

Let $(\setV,\setG)$ be a generalized circuit to be reduced to a stable distribution computation problem.
Let $V = |\setV| \geq 1$.
We will define a game that has at most $V + 1$ states and $|\setA| = 1$ actions, that is, agent policy will not have significance, and it will suffice to determine simple transition probabilities $P(s'|s,\mu)$ for all $s,s'\in\setS, \mu \in \Delta_\setS$.

\newcommand{\sbase}{s_\text{base}}

The proposed system will have a base state $\sbase \in \setS$ and 1 additional state $s_v$ associated with the gate whose output is $v \in \setV$.
Our construction will be sparse: only transition probabilities in between states associated with a gate and $\sbase$ will take positive values.
We define the useful constants $\theta := \frac{1}{8V}, B:= \frac{1}{4}$.

Given an (approximately) stable distribution $\mu^*$ of $P$, for each vertex $v$ we will read the satisfying assignment for the \CompGC{} problem by the value $u_1(\theta^{-1}\mu^*(s_v))$.
For each possible gate, we define the following gadgets.

\paragraph{Binary assignment gadget.}
For a gate of the form $G_{\leftarrow}(\zeta||v)$, we will add one state $s_v$ such that
\begin{align*}
    \text{If $\zeta = 1:$ } &\begin{cases}
        P(\sbase|s_v, \mu) = 1, \\
        P(s_v|s_v, \mu) = 0, \\
        P(s_v | \sbase, \mu) = \frac{\theta}{\max\{B, \,\mu(\sbase)\}}
    \end{cases} \\
    \text{If $\zeta = 0:$ } &\begin{cases}
        P(\sbase|s_v, \mu) = 1, \\
        P(s_v|s_v, \mu) = 0, \\
        P(s_v | \sbase, \mu) = 0 
    \end{cases}
\end{align*}

\paragraph{Weighted addition gadget.}
Next, we implement the addition gadget $G_{\times,+}(\alpha, \beta | v_1, v_2 | v)$ for $\alpha,\beta \in [-1, 1]$.
In this case, we also add one state $s_v$ to the game, and define the transition probabilities:
\begin{align*}
    P(\sbase|s_v, \mu) &= 1, \\
    P(s_v|s_v, \mu) &= 0, \\
    P(s_v | \sbase, \mu) &=  \frac{u_\theta(\alpha u_\theta(\mu(v_1)) + \beta u_\theta(\mu(v_2)))}{\max\{B, \,\mu(\sbase)\}}
\end{align*}

\paragraph{Brittle comparison gadget.}
For the comparison gate $G_{<}(|v_1, v_1|v)$, we also add one state $s_v$ to the game.
Define the function $ p_\delta: [-1,1] \rightarrow [0,1]$
\begin{align*}
    p_\delta (x, y) := u_1 \left( \frac{1}{2} +  \delta^{-1} (x - y) \right),
\end{align*}
for any $\delta > 0$.
In particular, if $x \geq y + \delta$, then $p_\delta (x, y) = 1$, and if $x \leq y - \delta$, then $p_\delta (x, y) = 0$.
We define the probability transitions to and from $s_v$ as
\begin{align*}
    P(s_v | \sbase, \mu) &= \frac{ \theta p_{8\varepsilon} (\theta^{-1} u_\theta(\mu(s_1)), \theta^{-1} u_\theta(\mu(s_2)))}{\max\{B, \,\mu(\sbase)\}} , \\ 
    P(s_v | s_v, \mu) &= 0, \\
    P(\sbase | s_v, \mu) &= 1.
\end{align*}

Finally, after all $s_v$ have been added, we complete the definition of $P$ by setting
\begin{align*}
    P(\sbase|\sbase,\mu) = 1 - \sum_{s' \in \setS} P(s'|\sbase,\mu).
\end{align*}

We first verify that the above assignment is a valid transition probability matrix for any $\mu\in \Delta_\setS$.
It is clear from definitions that for any $\mu,s\neq \sbase$, $P(\cdot|s,\mu)$ is a valid probability distribution as long as $8\varepsilon < 1$.
Moreover, for any $s \neq \sbase$, it holds that $0 \leq P(s|\sbase,\mu) \leq \frac{\theta}{B} < 1$, and it also holds that
\begin{align*}
    P(\sbase|\sbase,\mu) = 1 - \sum_{s' \in \setS} P(s'|\sbase,\mu) \geq 1 - \frac{V \theta}{B} \geq 0
\end{align*}
so $P(\cdot|\sbase,\mu)$ is a valid probability transition matrix.
Finally, the defined transition probability function $P$ is Lipschitz in the components of $\mu$, and $P$ can be defined as a composition of simple functions, hence $P \in \PSim$.
Finally, in this defined MFG, it holds that $V+1 = |\setS|$, since for each gate in the generalized circuit we defined one additional state.

\paragraph{Error propagation.}
We finally analyze the error propagation of the stationary distribution problem in terms of the generalized circuit.
Without loss of generality we assume $\varepsilon < \frac{1}{8}$.
First, for any solution of the \CompStat{} problem $\mu^*$, whenever $\varepsilon < \frac{1}{8}$, it must hold that:
\begin{align*}
    \left| \mu^*(\sbase) - \sum_{s' \in \setS} \mu^*(s) P(\sbase| s, \mu^*)\right| \leq \frac{1}{8|\setS|},
\end{align*}
hence (using $V < |\setS|$) we have the lower bound on $\mu^*(\sbase)$ given by:
\begin{align*}
    \mu^*(\sbase) \geq &\sum_{s \in \setS} \mu^*(s) P(\sbase| s, \mu^*) - \frac{1}{8V} \\
    \geq & \mu^*(\sbase) P(\sbase| \sbase, \mu^*) + \sum_{s \neq \sbase} \mu^*(s) P(\sbase| s, \mu^*) - \frac{1}{8V} \\
    \geq &\mu^*(\sbase) \left(1 - \frac{V\theta}{B}\right) + \sum_{s \neq \sbase} \mu^*(s) - \frac{1}{8V} \\
    \geq &\mu^*(\sbase) \left(1 - \frac{V\theta}{B}\right) + (1 - \mu^*(\sbase)) - \frac{1}{8V} \\
    & \implies \mu^*(\sbase) \geq \frac{1 - \frac{1}{8V}}{1 + \frac{V\theta}{B}} \geq B = \frac{1}{4}.
\end{align*}
We will show that a solution of the \CompStat{} can be converted into a $\varepsilon'$-satisfying assignment 
\begin{align*}
    v \rightarrow u_1\left(\frac{\mu^*(s_v)}{\theta}\right),
\end{align*}
for some appropriate $\varepsilon'$ to be defined later.

\textbf{Case 1: Binary assignment error.}
First, assume $G_{\leftarrow}(\zeta||v) \in \setG$
If $\zeta = 1$, since $\mu^*$ is a $\varepsilon$ stable distribution we have 
\begin{align*}
    | \mu^*(s_v) - \mu^*(\sbase) P (s_v|\sbase, \mu^*)| &\leq \frac{\varepsilon}{|\setS|} \\
    \left| \mu^*(s_v) - \mu^*(\sbase) \frac{\theta}{\max\{B , \mu^*(\sbase)\}} \right| &\leq \frac{\varepsilon}{|\setS|}  \\
    \left| \mu^*(s_v) -  \theta \right| &\leq \frac{\varepsilon}{|\setS|} \\
    \left| \frac{\mu^*(s_v)}{\theta} -  1 \right| &\leq \frac{\varepsilon}{\theta|\setS|} \leq \frac{\varepsilon}{\theta V} \leq8\varepsilon,
\end{align*}
where we used the fact that $\frac{\theta}{\max\{B , \mu^*(\sbase)\}} =  \mu^*(\sbase) $.
and it follows by definition that $|u_1\left(\frac{\mu^*(s_v)}{\theta}\right) -  1 | \leq 8\varepsilon$, since the map $u_1$ is 1-Lipschitz and therefore can only decrease the absolute value on the left.
Likewise, if $\zeta = 0$,
\begin{align*}
    | \mu^*(s_v) - \sum_{s\in \setS}\mu^*(s) P (s_v|s, \mu^*)| &\leq \frac{\varepsilon}{|\setS|} \\
    | \mu^*(s_v)| &\leq \frac{\varepsilon}{|\setS|} \\
     \left| \frac{\mu^*(s_v)}{\theta} \right| &\leq \frac{\varepsilon}{\theta|\setS|} \leq 8\varepsilon
\end{align*}
and once again $u_1\left(\frac{\mu^*(s_v)}{\theta}\right) \leq 8\varepsilon$.

\textbf{Case 2: Weighted addition error.}
Assume that $G_{\times,+}(\alpha, \beta | v_1, v_2 | v) \in \setG$, and set $\square := u_\theta(\alpha u_\theta(\mu(v_1)) + \beta u_\theta(\mu(v_2)))$.
Using the fact that $\| \mu^* - \Gamma_P(\mu^*)\|\leq \frac{\varepsilon}{|\setS|}$,
\begin{align*}
    |\mu^*(s_v) - \sum_{s\in \setS} \mu^*(s) P(s_v|s,\mu^*) | &\leq \frac{\varepsilon}{|\setS|}, \\
    \left|\mu^*(s_v) - \mu^*(\sbase) \frac{u_\theta(\alpha u_\theta(\mu(v_1)) + \beta u_\theta(\mu(v_2)))}{\max\{B, \,\mu(\sbase)\}} \right| &\leq \frac{\varepsilon}{|\setS|}, \\
    \left|\frac{\mu^*(s_v)}{\theta} - \frac{\square}{\theta} \right| &\leq \frac{\varepsilon}{|\setS|\theta},
\end{align*}
which implies 
\begin{align*}
    \left|u_1\left(\frac{\mu^*(s_v)}{\theta}\right) - u_1\left(\alpha u_1\left(\frac{\mu^*(v_1)}{\theta}\right) + \beta u_1\left(\frac{\mu^*(v_2)}{\theta}\right)\right)\right| \leq 8\varepsilon.
\end{align*}

\textbf{Case 3: Brittle comparison gadget.}
Finally, we analyze the more involved case of the comparison gadget.
Assume $G_{<}(|v_1, v_2|v) \in \setG$.
The stability conditions for $s_v$ yield:
\begin{align*}
    |\mu^*(s_v) - \mu^*(\sbase) P(s_v|\sbase, \mu^*) | \leq &\frac{\varepsilon}{|\setS|} \\
    |\mu^*(s_v) - \theta p_{8\varepsilon} (\theta^{-1} u_\theta(\mu^*(v_1)), \theta^{-1} u_\theta(\mu^*(v_2))) | \leq &\frac{\varepsilon}{|\setS|}
\end{align*}
We analyze two cases: $u_1(\theta^{-1}\mu^*(v_1)) \geq u_1(\theta^{-1}\mu^*(v_2)) + 8\varepsilon$ and $u_1(\theta^{-1}\mu^*(v_1)) \leq u_1(\theta^{-1}\mu^*(v_2)) - 8\varepsilon$.
In the first case, we obtain
\begin{align*}
    \theta^{-1}u_{\theta}(\mu^*(v_1)) \geq\theta^{-1}u_{\theta}(\mu^*(v_2)) + 8\varepsilon,
\end{align*}
which implies by the definition of $p_{8\varepsilon}$
\begin{align*}
    |\mu^*(s_v) - \theta | \leq &\frac{\varepsilon}{|\setS|} \\
    |u_1(\theta^{-1}\mu^*(s_v)) - 1 | \leq &\frac{\varepsilon}{|\setS|\theta} \\
    u_1(\theta^{-1}\mu^*(s_v)) \geq & 1 -\frac{\varepsilon}{|\setS|\theta} \geq 1 - 8\varepsilon.
\end{align*}
In the second case $u_1(\theta^{-1}\mu^*(v_1)) \leq u_1(\theta^{-1}\mu^*(v_2)) - 8\varepsilon$, it follows by a similar analysis that
\begin{align*}
    u_1(\theta^{-1}\mu^*(s_v)) \leq & \frac{\varepsilon}{|\setS|\theta} \leq 8\varepsilon.
\end{align*}

Hence, in the above, we reduced the \CompGCeps{$8\varepsilon$} problem to the \CompStateps{$\varepsilon$} problem, completing the proof that \CompStat{} is \ppadhard{}.
The fact that \CompStat{} is in \ppad{} on the other hand easily follows from the fact that \CompStat{} is the fixed point problem for the (simple) operator $\Gamma_P$, reducing it to the \CompEOL{} problem by a standard construction \cite{daskalakis2009complexity}.

\subsection{Proof of Intractability of FH-MFG}

\newcommand{\stbv}[1]{s_{#1, \text{base}}}
\newcommand{\stvone}[1]{s_{#1,1}}
\newcommand{\stvzero}[1]{s_{#1,0}}

As in the previous section, we reduce any \CompGC{} problem $(\setG, \setV)$ to the problem \CompFHeps{$\varepsilon^2$}{$2$} for some simple reward $R\in \setR^{\text{Sim}}$.
Once again let $V = |\setV|$.

Associated with each $v \in \setV$ we define $\stvone{v}, \stvzero{v}, \stbv{v} \in \setS$.
The initial distribution is defined as
\begin{align*}
    \mu_0(\stbv{v}) = \frac{1}{V}, \forall v \in \setV,
\end{align*}
and we define two actions for each state: $\setA = \{a_1, a_0 \}$.
The state transition probability matrix is given by
\begin{align*}
    P(s|\stbv{v}, a) &= \begin{cases}
        1, \text{ if } a = a_1, s = \stvone{v}, \\
        1, \text{ if } a = a_0, s = \stvzero{v}, \\
        0, \text{ otherwise.}
    \end{cases} \\
    P(\stbv{v}| s, a) &= 0, \forall v\in\setV, s\in\setS, a\in\setA,
\end{align*}
and an $\varepsilon$ satisfying assignment $p:\setV \rightarrow [0,1]$ will be read by $p(v) = \pi_1^*(a_1|\stbv{v})$ for the optimal policy $\vecpi^* = \{ \pi_h \}_{h=0}^1$.
We will specify population-dependent rewards $R \in \RL^{\text{Simple}}$, since $R$ will not depend on the particular action but only the state and population distribution, we will concisely denote $R(s,a,\mu) = R(s,\mu)$.
It will be the case that
\begin{align*}
    R(\stbv{v}, \mu) = 0, \forall v \in \setV, \mu \in \Delta_\setS.
\end{align*}
We assign $R(\stvone{v},\mu) = R(\stvzero{v}, \mu) = 0, \forall \mu$ for any vertex $v$ of the generalized circuit that is not the output of any gate in $\setG$.

\paragraph{Binary assignment gadget.}
For any binary assignment gate $G_{\leftarrow}(\zeta||v)$, we assign
\begin{align*}
    R(\stvone{v}, \mu) &= \zeta, \\
    R(\stvzero{v}, \mu) &= 1- \zeta, \forall \mu \in \Delta_\setS.
\end{align*}

\paragraph{Weighted addition gadget.}
For any gate $G_{\times,+}(\alpha, \beta | v_1, v_2 | v)$, 
\begin{align*}
    R(\stvone{v}, \mu) &= u_1(u_1(\alpha V \mu(\stvone{v_1}) + \beta V \mu(\stvone{v_2})) - V\mu(\stvone{v})), \\
    R(\stvzero{v}, \mu) &= u_1(V\mu(\stvone{v}) -u_1(\alpha V\mu(\stvone{v_1}) + \beta V\mu(\stvone{v_2}))),
\end{align*}
for all $\mu \in \Delta_\setS$.

\paragraph{Brittle comparison gadget.}
For any gate $G_{<}(|v_1, v_2|v)$, we define the rewards for states $\stvone{v}, \stvzero{v}$ as
\begin{align*}
    R(\stvone{v}, \mu) &= u_1(V\mu(\stvone{v_2}) - V\mu(\stvone{v_1})), \\
    R(\stvzero{v}, \mu) &= u_1(V\mu(\stvone{v_1}) - V\mu(\stvone{v_2})), \forall \mu\in \Delta_\setS.
\end{align*}

Now assume that $\vecpi^* = \{ \pi^*_h \}_{h=0}^1$ is a solution to the \CompFHeps{$\varepsilon^2$}{$2$} problem and $\vecmu^* = \Lambda_{P,\mu_0}^2(\vecpi^*)$, that is, assume that for all $\vecpi \in \Pi^2$,
\begin{align*}
    \Vfin(\vecmu^*, \vecpi) - \Vfin(\vecmu^*, \vecpi^*) \leq \frac{\varepsilon^2}{V}.
\end{align*}
Firstly, if $\mu_1^*$ is induced by $\vecpi^*$, it holds that $\forall v \in \setV$,
\begin{align*}
    \mu_1^*(\stbv{v}) = 0, \quad
    &\mu_1^*(\stvone{v}) = \frac{1}{V} \pi_0^*(\stvone{v}|\stbv{v}), \\
    \mu_1^*&(\stvzero{v}) = \frac{1 - \pi_0^*(\stvone{v}|\stbv{v})}{V} .
\end{align*}

\newcommand{\pibrzero}{\pi_0^{\text{br}}}
\newcommand{\pibrone}{\pi_1^{\text{br}}}
\newcommand{\pibrvec}{\vecpi^{\text{br}}}

Furthermore, a policy $\pibrvec \in \Pi_2$ that is the best response to $\vecmu^* := \{ \mu^*_0, \mu^*_1 \}$ can be always formulated as:
\begin{align*}
    \pibrzero (a_1 | \stbv{v}) &= \begin{cases}
        1, \text{ if $R(\stvone{v}, \mu_1^*) > R(\stvone{v}, \mu_1^*)$}, \\
        0, \text{otherwise}
    \end{cases} \\
    \pibrzero (a_0 | \stbv{v}) &= 1- \pibrzero (a_1 | \stbv{v}), \\
    \pibrone (a_1 | \stbv{v}) &= 1, \\
    \pibrone (a_0 | \stbv{v}) &= 0.
\end{align*}
By the optimality conditions, we will have
\begin{align*}
    \Vfin(\vecmu^*, \pibrvec) - \Vfin(\vecmu^*, \vecpi^*) \leq \frac{\varepsilon^2}{V}.
\end{align*}
Furthermore, for any $v\in\setV$ it holds that
\begin{align*}
    &\Vfin(\vecmu^*, \pibrvec) - \Vfin(\vecmu^*, \vecpi^*) \\
    & = \sum_{v \in \setV} \mu_0(\stbv{v}) [ \max_{s \in \{\stvone{v}, \stvzero{v} \}} R(s, \mu^*_1)    \\
        & \qquad - \pi^*_0(a_1|\stbv{v}) R(\stvone{v}, \mu^*_1) - \pi^*_0(a_0|\stbv{v}) R(\stvzero{v}, \mu^*_1)] \\
    \geq & \frac{1}{V} \max_{s \in \{\stvone{v}, \stvzero{v} \}} R(s, \mu^*_1) \\
    \qquad &- \frac{1}{V} \pi^*_0(a_1|\stbv{v}) R(\stvone{v}, \mu^*_1) - \frac{1}{V} \pi^*_0(a_0|\stbv{v}) R(\stvzero{v}, \mu^*_1)
\end{align*}
as the summands are all positive.
We prove that all gate conditions are satisfied case by base.
Without loss of generality, we assume $\varepsilon < 1$ below.

\textbf{Case 1.} 
It follows that for any $v\in\setV$ such that $G_{\leftarrow}(\zeta||v) \in \setG$, we have
\begin{align*}
    \frac{1}{V} - \frac{1}{V} \pi^*_0(a_1|\stbv{v}) \zeta - \frac{1}{V} \pi^*_0(a_0|\stbv{v}) (1 - \zeta)  &\leq \frac{\varepsilon^2}{V} \\
    1 - \pi^*_0(a_1|\stbv{v}) \zeta -  (1- \pi^*_0(a_1|\stbv{v})) (1 - \zeta)  &\leq \varepsilon^2\\
    \zeta (1 - 2 \pi^*_0(a_1|\stbv{v}) ) + \pi^*_0(a_1|\stbv{v}) \leq \varepsilon^2 \leq \varepsilon.
\end{align*}
The above implies $\pi^*_0(a_1|\stbv{v}) \geq 1 - \varepsilon$ if $\zeta = 1$, and if $\zeta = 0$, it implies $\pi^*_0(a_1|\stbv{v}) \leq \varepsilon$.

\textbf{Case 2.}
For any $v\in\setV$ such that $G_{\times,+}(\alpha, \beta | v_1, v_2 | v) \in \setG$, denoting in short
\begin{align*}
    \square &:= u_1(\alpha V\mu_1^*(\stvone{v_1}) + \beta V\mu_1^*(\stvone{v_2})) \\
        &= u_1(\alpha \pi^*_0(a_1|\stvone{v_1}) + \beta \pi^*_0(a_1|\stvone{v_2})), \\
    p_1 &:= \pi^*_0(a_1|\stbv{v})\\
    p_0 &:= \pi^*_0(a_0|\stbv{v})
\end{align*}
we have
\begin{align*}
    \frac{1}{V} &\max\big\{u_1(V\mu_1^*(\stvone{v}) - \square), u_1( \square - V\mu_1^*(\stvone{v}))\big \} \\
        &- \frac{1}{V} \pi^*_0(a_1|\stbv{v}) u_1( \square - V\mu_1^*(\stvone{v})) \\
        &- \frac{1}{V} \pi^*_0(a_0|\stbv{v}) u_1(V\mu_1^*(\stvone{v}) - \square) \leq \varepsilon^2,
\end{align*}
or equivalently
\begin{align*}
    \max\big\{u_1(p_1 - \square), &u_1( \square - p_1 )\big \} - p_1 u_1( \square - p_1 ) - p_0 u_1(p_1 - \square) \leq \varepsilon^2.
\end{align*}
First, assume it holds that $p_1 \leq \square$, then:
\begin{align*}
    u_1(\square - p_1) - p_1 u_1( \square - p_1 ) \leq &\varepsilon^2 \\
    (1 - p_1) (\square - p_1) \leq &\varepsilon^2.
\end{align*}
The above implies that either $p_1 \geq 1 - \varepsilon$ or $u_1(\square - p_1) \leq \varepsilon$, both cases implying $| \square - p_1 | \leq \varepsilon$ since we assume $\square \geq p_1$.
To conclude case 2, assume that $\square < p_1$, then 
\begin{align*}
    u_1(p_1 - \square) - (1-p_1) u_1(p_1 - \square) &\leq \varepsilon^2, \\
    p_1 (p_1 - \square) &\leq \varepsilon^2,
\end{align*}
then either $p_1 \leq \varepsilon$ or $p_1 - \square \leq \varepsilon$, either case implying once again $| \square - p_1 | \leq \varepsilon$.

\textbf{Case 3.}
Finally, for any $v \in \setV$ such that $G_{<}(|v_1, v_2|v) \in \setG$,
\begin{align*}
     \frac{1}{V} & \max\Big\{ u_1(\mu(\stvone{v_2}) - \mu(\stvone{v_1})), u_1(\mu(\stvone{v_1}) - \mu(\stvone{v_2})) \Big\} \\
        &- \frac{1}{V} \pi^*_0(a_1|\stbv{v}) u_1(\mu(\stvone{v_1}) - \mu(\stvone{v_2})) \\
     &- \frac{1}{V} \pi^*_0(a_0|\stbv{v}) u_1(\mu(\stvone{v_2}) - \mu(\stvone{v_1})) \leq \varepsilon 
\end{align*}
hence once again using the shorthand notation:
\begin{align*}
    \triangle &:= V\mu^*_1(\stvone{v_2}) - V\mu^*_1(\stvone{v_1}) = \pi^*_0(a_1|\stvone{v_2}) - \pi^*_0(a_1|\stvone{v_1})\\
    p_1 &:= \pi^*_0(a_1|\stbv{v})\\
    p_0 &:= \pi^*_0(a_0|\stbv{v})
\end{align*}
we have the inequality:
\begin{align*}
u_1(|\triangle|) - p_1 u_1(\triangle) - p_0 u_1(-\triangle) &\leq \varepsilon^2 \\
    u_1(|\triangle|) - p_1 u_1(\triangle) - (1- p_1)u_1(-\triangle) &\leq \varepsilon^2.
\end{align*}
First assume $\triangle \geq \varepsilon$, then
\begin{align*}
    u_1(\triangle) (1 - p_1) \leq \varepsilon^2 \implies
    1 - \varepsilon \leq p_1,
\end{align*}
and conversely if $\triangle \leq - \varepsilon$,
\begin{align*}
    u_1(-\triangle) p_1 \leq \varepsilon^2 \implies
    p_1 \leq \varepsilon,
\end{align*}
concluding that the comparison gate conditions are $\varepsilon$ satisfied for the assignment $v \rightarrow \pibrzero(a_1|\stbv{v})$.

The three cases above conclude that $v \rightarrow \pibrzero(a_1|\stbv{v})$ is an $\varepsilon$-satisfying assignment for the generalized circuit $(\setV, \setG)$, concluding the proof that \CompFHeps{$\varepsilon_0$}{$2$} is \ppadhard{} for some $\varepsilon_0 >0$.
The fact that \CompFHeps{$\varepsilon_0$}{$2$} is in \ppad{} follows from the fact that the NE is a fixed point of a simple map on space $\Pi_2$, see for instance \cite{huang2023statistical}.

\subsection{Proof of Intractability of \texorpdfstring{\CompFHLinear{2}}{2-FH-Linear}}

Our reduction will be similar to the previous section, however, instead of reducing a $\varepsilon$-\textsc{GCircuit} to an MFG, we will reduce a $2$ player general sum normal form game, $2$-\textsc{Nash}, to a finite horizon mean field game with linear rewards with horizon $H=2$ (\CompFHLinear{2}). 
Let $\varepsilon > 0, K_1, K_2 \in \mathbb{N}_{>0}, A,B\in\mathbb{R}^{K_1, K_2}$ be given for a \CompTwoNash{} problem.
We assume without loss of generality that $K_1 > 1$, as otherwise, the solution of \CompTwoNash{} is trivial.

\newcommand{\sbaseA}{s_{\text{base}}^1}
\newcommand{\sbaseB}{s_{\text{base}}^2}

This time, we define finite horizon game with $K_1 + K_2 + 2$ states, denoted $\setS := \{\sbaseA, \sbaseB, s^1_1,\ldots, s^1_{K_1}, s^2_{1}, \ldots, s^2_{K_2}\}$.
Without loss of generality, we can assume $K_1 \leq K_2$.
The action set will be defined by $\setA = [K_2] = \{ 1, \ldots, K_2\}$.
The initial state distribution will be given by $\mu_0(\sbaseA) = \mu_0(\sbaseB) = \sfrac{1}{2}$, with $\mu_0(s) = 0$ for all other states.
We define the transitions for any $s \in \setS, a, a' \in \setA$ as:
\begin{align*}
    P(s|\sbaseA, a) &= \begin{cases}
        1, \text{ if $s = s^1_a$ and $a \leq K_1$}, \\
        1, \text{ if $s = s^1_a$ and $a > K_1$}, \\
        0, \text{otherwise.}
    \end{cases} \\
    P(s|\sbaseB, a) &= \begin{cases}
        1, \text{ if $s = s^2_a$}, \\
        0, \text{otherwise.}
    \end{cases} \\
    P(s|s^1_a, a') = \begin{cases}
        1, \text{ if $s = s^1_a$}, \\
        0, \text{otherwise.}
    \end{cases} 
    & \qquad 
    P(s|s^2_a, a') = \begin{cases}
        1, \text{ if $s = s^2_a$}, \\
        0, \text{otherwise.}
    \end{cases}
\end{align*}
Finally, we will define the linear reward function as for all $a \in [K_2]$:
\begin{align*}
    R(\sbaseA, a, \mu) &= 0, \\
    R(\sbaseB, a, \mu) &= 0, \\
    R(s^1_a, a, \mu) &= \begin{cases}
        0, \text{if $a > K_1$}, \\
        \frac{1}{2} + \frac{1}{2} \sum_{a' \in [K_2]} \mu(s^2_{a'}) A_{a, a'} 
    \end{cases} \\
    R(s^2_a, a, \mu) &=  \frac{1}{2} + \frac{1}{2} \sum_{a' \in [K_1]} \mu(s^1_{a'}) B_{a', a}.
\end{align*}
In words, the states $\sbaseA, \sbaseB$ represent the two players of the \CompTwoNash{}, and an agent starting from one of the initial base states $\sbaseA, \sbaseB$ of the FH-MFG at round $h=0$ will be placed at $h=1$ at a state representing the (pure) strategies of each player respectively.

Given the game description above, assume $\vecpi^* = \{ \pi_h^*\}_{h=0}^1$ is an $\varepsilon$ solution of the \CompFHLinear{2}.
Then, it holds for the induced distribution $\vecmu^* := \{ \mu^*_h \}_{h=0}^1 = \Lambdaop$ that:
\begin{align*}
    \mu_0^* &= \mu_0, \\
    \mu^*_1 (s) &= \sum_{s',a' \in \setS \times \setA} \mu_0(s') \pi^*(a'|s') P(s|s',a') \\
        &= \begin{cases}
        \frac{1}{2} \pi_0(i|\sbaseA), \text{ if $s=s^1_i$, for some $i\in[K_1]$}, \\
        \frac{1}{2} \pi_0(i|\sbaseB), \text{ if $s=s^2_i$, for some $i\in[K_2]$}, \\
        \frac{1}{2} - \frac{1}{2}\sum_{i\in[K_1]}\pi_0(i|\sbaseA), \text{ if $s=\sbaseA$,} \\
        0, \text{otherwise.}
    \end{cases}
\end{align*}
By definition of the $\varepsilon$ finite horizon Nash equilibrium, 
\begin{align*}
    \Expfin (\vecpi^*) &:= \max_{\vecpi' \in \Pi^H} \Vfin( \Lambda^H_P (\vecpi^*), \vecpi') - \Vfin ( \Lambda^H_P (\vecpi^*), \vecpi ) \leq \varepsilon,
\end{align*}
in particular, it holds for any $\vecpi \in \Pi_2$ that
\begin{align}
     \Vfin( \vecmu^*, \vecpi) - \Vfin ( \vecmu^*, \vecpi^* ) \leq \varepsilon. \label{eq:linfh_comp}
\end{align}
By direct computation, the value functions $\Vfin$ can be written directly in this case for any $\pi$:
\begin{align*}
    \Vfin( \vecmu^*, \vecpi) = &\frac{1}{2} \sum_{a \in [K_1]} \pi_0(a|\sbaseA) \left(\frac{1}{2} + \frac{1}{2} \sum_{a' \in [K_2]} \mu_1^*(s^2_{a'}) A_{ a, a'} \right) \\
        &+\frac{1}{2} \sum_{a' \in [K_2]}\pi_0(a'|\sbaseB) 
    \left(\frac{1}{2} + \frac{1}{2} \sum_{a \in [K_1]} \mu_1^*(s^1_{a}) B_{a, a'} \right)
    \\
     = & \frac{1}{4}\left( 1 + \sum_{a\in[K_1]} \pi_0(a|\sbaseA)\right) \\
        &+ \frac{1}{8} \sum_{a \in [K_1]} \sum_{a' \in [K_2]} \pi_0(a|\sbaseA)
     \pi^*_0(a'|\sbaseB)A_{ a, a'} \\
     &+ \frac{1}{8} \sum_{a \in [K_1]} \sum_{a' \in [K_2]} \pi_0(a'|\sbaseB) 
     \pi^*_0(a|\sbaseA) B_{a, a'} 
\end{align*}
We analyze two different cases, accounting for a possible imbalance between the strategy spaces of the two players, $[K_1]$ and $[K_2]$.

\textbf{Case 1.} Assume $K_1 = K_2$.
Then, $\Vfin(\vecmu^*, \vecpi)$ simplifies to
\begin{align}
    \Vfin(\vecmu^*, \vecpi) = &\frac{1}{2} + \frac{1}{8} \sum_{a \in [K_1]} \sum_{a' \in [K_2]} \pi_0(a|\sbaseA)
     \pi^*_0(a'|\sbaseB)A_{ a, a'} \notag 
     \\
     &+ \frac{1}{8} \sum_{a \in [K_1]} \sum_{a' \in [K_2]} \pi_0(a'|\sbaseB) 
     \pi^*_0(a|\sbaseA) B_{a, a'}. \label{eq:k1k2eq_vfindiff}
\end{align}
Take an arbitrary mixed strategy $\sigma_1 \in \Delta_{[K_1]}$ and define the policy $\vecpi_A = \{\pi_{A,h} \}_{h=0}^1 \in \Pi^2$ so that
\begin{align*}
    \pi_{A,0}(\sbaseA) = \sigma_1, \quad
    \pi_{A,0}(\sbaseB) = \pi^*_0(\sbaseB), \quad 
    \pi_{A,1} = \pi^*_1.
\end{align*}
Then, placing $\vecpi_A$ in equations \eqref{eq:k1k2eq_vfindiff} and \eqref{eq:linfh_comp}, it follows that
\begin{align}
    \sum_{a \in [K_1]} &\sum_{a' \in [K_2]} \sigma_1(a)
     \pi^*_0(a'|\sbaseB)A_{ a, a'} \notag
     \\
     &- \sum_{a \in [K_1]} \sum_{a' \in [K_2]} \pi^*_0(a|\sbaseA)
     \pi^*_0(a'|\sbaseB)A_{ a, a'} \leq 8\varepsilon. \label{eq:k1k2eq_res1}
\end{align}
Similarly, for any $\sigma_2 \in \Delta_[K_2]$, replacing $\vecpi$ in equations \eqref{eq:k1k2eq_vfindiff} and \eqref{eq:linfh_comp} with a policy $\vecpi_B$ such that 
\begin{align*}
    \pi_{B,0}(\sbaseA) = \pi^*_0(\sbaseA), \quad
    \pi_{B,0}(\sbaseB) = \sigma_2, \quad
    \pi_{B,1} = \pi^*_1,
\end{align*}
we obtain
\begin{align}
    \sum_{a \in [K_1]} &\sum_{a' \in [K_2]} \sigma_2(a) 
     \pi^*_0(a'|\sbaseA) B_{a, a'} \notag \\ 
        &- \sum_{a \in [K_1]} \sum_{a' \in [K_2]} \pi^*_0(a'|\sbaseB) 
     \pi^*_0(a|\sbaseA) B_{a, a'} \leq 8\varepsilon. \label{eq:k1k2eq_res2}
\end{align}
Hence, the resulting equations \eqref{eq:k1k2eq_res1}, \eqref{eq:k1k2eq_res2} imply that in this case the strategy profile $(\pi^*_0(\sbaseA), \pi^*_0(\sbaseB))$ is a $8\varepsilon$-Nash equilibrium for the normal form game defined by matrices $A,B$.

\textbf{Case 2.}
Next, we analyze the case when $1 < K_1 < K_2$.
If $\sum_{a'\in [K_1]} \pi^*_0(a'|\sbaseA) = 0$, then the policy
\begin{align*}
    \pi'_{0}(1|\sbaseA) = 1, \quad
    \pi'_{0}(\sbaseB) = \pi^*_0(\sbaseB), \quad 
    \pi'_{1} = \pi^*_1.
\end{align*}
yields an exploitability of at least $\sfrac{1}{4}$, so by taking $\varepsilon$ smaller than $\sfrac{1}{4}$ we can discard this possibility.

Otherwise, we define a policy $\vecpi_C = \{\pi_{C,h} \}_{h=0}^1 \in \Pi^2$ such that
\begin{align*}
    \pi_{C,0}(a|\sbaseA) = \begin{cases}
        \frac{\pi^*_0(a|\sbaseA)}{\sum_{a'\in [K_1]} \pi^*_0(a'|\sbaseA) }, \text{ if $a \in [K_1]$}, \\
        0, \text{otherwise.}
    \end{cases}\\
    \pi_{C,0}(\sbaseB) = \pi^*_0(\sbaseB), \quad
    \pi_{C,1} = \pi^*_1,
\end{align*}
and replace $\vecpi$ in Equation \eqref{eq:linfh_comp} with $\vecpi_C$ to obtain:
\begin{align*}
     &\frac{1}{4} - \frac{1}{4}S \\
        &+
     \frac{1}{8} \left(S^{-1} - 1\right)\sum_{a \in [K_1]} \sum_{a' \in [K_2]} \pi^*_0(a|\sbaseA)
     \pi^*_0(a'|\sbaseB)A_{ a, a'} \leq \varepsilon
\end{align*}
where $S := \sum_{a'\in [K_1]} \pi^*_0(a'|\sbaseA) < 1$, hence
\begin{align*}
    1-S = \sum_{a'\in [K_2] - [K_1]} \pi^*_0(a'|\sbaseA) \leq 4\varepsilon.
\end{align*}
Now for some $\sigma_1 \in \Delta_{[K_1]}$, once again take the policy $\vecpi_{A}$ defined in Case 1, and use Inequality \eqref{eq:linfh_comp} to obtain:
\begin{align*}
    \frac{1}{4} (1 - S) + &\frac{1}{8} \sum_{a \in [K_1]} \sum_{a' \in [K_2]} \sigma_1(a)
     \pi^*_0(a'|\sbaseB)A_{ a, a'} \\
        &- \frac{1}{8} \sum_{a \in [K_2]} \sum_{a' \in [K_2]} \pi^*_0(a|\sbaseA)
     \pi^*_0(a'|\sbaseB)A_{ a, a'}
    \leq \varepsilon \\
    &\sum_{a \in [K_1]} \sum_{a' \in [K_2]} \sigma_1(a)
     \pi^*_0(a'|\sbaseB)A_{ a, a'} \\
    &- \sum_{a \in [K_1]} \sum_{a' \in [K_2]} \pi^*_0(a|\sbaseA)
     \pi^*_0(a'|\sbaseB)A_{ a, a'}
    \leq 8\varepsilon.
\end{align*}
Here, using the definition of $\vecpi_C$, as $\pi_{C,0}(a|\sbaseA) \geq \pi^*_0(a|\sbaseA)$ for $a \in [K_1]$, we obtain:
\begin{align*}
    &\sum_{a \in [K_1]} \sum_{a' \in [K_2]} \sigma_1(a)
     \pi_{C,0}(a'|\sbaseB)A_{ a, a'}
    \\
    &- \sum_{a \in [K_1]} \sum_{a' \in [K_2]} \pi_{C, 0}(a|\sbaseA)
     \pi_{C,0}(a'|\sbaseB)A_{ a, a'}
    \leq 8\varepsilon.
\end{align*}

Next take $\vecpi_B$ as defined above in Case 1 for any arbitrary $\sigma_2 \in \Delta_{[K_2]}$ and use the Inequality~\ref{eq:linfh_comp}:
\begin{align*}
    \sum_{a' \in [K_2]} &\sum_{a \in [K_1]} \sigma_2(a')
     \pi^*_0(a|\sbaseA) B_{ a, a'} \\
        &- \sum_{a \in [K_1]} \sum_{a' \in [K_2]} \pi^*_0(a|\sbaseA)
     \pi^*_0(a'|\sbaseB) B_{ a, a'}
    \leq 8\varepsilon \\
    \sum_{a \in [K_1]} &\sum_{a' \in [K_2]} \sigma_2(a')
     \pi_{C,0}(a|\sbaseA) B_{ a, a'} \\
    &- \sum_{a \in [K_1]} \sum_{a' \in [K_2]} \pi_{C,0}(a|\sbaseA)
     \pi_{C,0}(a'|\sbaseB) B_{ a, a'}
    \leq \frac{8\varepsilon}{S} \leq \frac{8\varepsilon}{1 - 4\varepsilon}.
\end{align*}
Assuming without loss of generality that $\varepsilon < \frac{1}{8}$, it follows that $\pi_{C,0}(\sbaseA), \pi_{C,0}(\sbaseB)$ is a $16\varepsilon$ solution to the \CompTwoNash{}.

\end{document}